\documentclass[12pt,a4paper]{amsart}
\usepackage[english]{babel}

\usepackage[latin1]{inputenc}
\usepackage{amsmath}
\usepackage{amscd}
\usepackage{amstext}
\usepackage{color}
\usepackage{amsbsy}
\usepackage{amsopn}
\usepackage[tiny,nohug,heads=vee]{diagrams}
\diagramstyle[labelstyle=\scriptstyle]
\usepackage{amsthm}
\usepackage{amsxtra}
\usepackage{upref}
\usepackage{amsfonts}
\usepackage{amssymb}
\usepackage{mathrsfs}
\usepackage{euscript}
\usepackage{array}
\usepackage{stmaryrd}
\usepackage{hyperref}
\usepackage{verbatim}

\theoremstyle{plain}
      \newtheorem{theorem}{Theorem}
      \newtheorem{lemma}[theorem]{Lemma}
      
      \newtheorem{proposition}[theorem]{Proposition}
      
      \theoremstyle{definition}

      \theoremstyle{ex}
      \newtheorem{ex}[theorem]{Example}
 \theoremstyle{remark}
      \newtheorem{remark}[theorem]{Remark}
      \theoremstyle{proof}

\newcounter{Step}
\newenvironment{step}[0]{\bigskip\addtocounter{Step}{1}\noindent\textbf{Step \theStep :} }{\
  \begin{flushright} \end{flushright}}

\def\N{\mbox{I\hspace{-.15em}N} }
\def\R{\mbox{I\hspace{-.15em}R} }

\def\C{\hspace{.17em}\mbox{l\hspace{-.47em}C} }

\def\o{\otimes}

\begin{document}
\title[Functional properties of Generalized H\"ormander spaces of distributions II]{Functional properties of Generalized H\"ormander spaces of distributions II :\\
{ Multilinear maps and applications to spaces of functionals with wave front set conditions}}
\author{Yoann Dabrowski}
\address{ Universit\'{e} de Lyon\\ 
Universit\'{e} Lyon 1\\
Institut Camille Jordan UMR 5208\\
43 blvd. du 11 novembre 1918\\
F-69622 Villeurbanne cedex\\
France
}
\email{ dabrowski@math.univ-lyon1.fr}
\subjclass[2010]{Primary 46F05; Secondary : 81T20}
\keywords{Wave front set, Microcausal functionals, Retarded brackets, Infinite dimensional Poisson algebras.}
\begin{abstract}
We continue our study and applications of  generalized H\"ormander spaces of distributions $\mathcal{D}'_{\gamma,\Lambda}$ with $C^\infty$ wavefront
set included in a cone $\Lambda$ and the union of $H^s$-wave front sets in a second cone $\gamma\subset \Lambda$. We give hypocontinuity results and failure of continuity of tensor multiplication maps between these spaces and deduce hypocontinuity results for various compositions on spaces of multilinear maps. We apply this study to a generalization of microcausal functionals from algebraic quantum field theory with derivatives controlled by spaces either of the form $\mathcal{D}'_{\gamma,\Lambda}$ or some $\epsilon$-tensor product of them. We prove nuclearity and completeness results and give general results to build  Poisson algebra structures (with at least hypocontinuous bilinear products). We also apply our general framework to build retarded products with field dependent propagators.
\end{abstract}

\maketitle

\begin{center}\section*{\textsc{Introduction}}
\end{center}

In 1992, Radzikowski~\cite{Radzikowski-92-PhD,Radzikowski} showed 
the wave front set of distributions 
to be  a key
concept to define quantum fields in curved spacetime.
This idea was fully developed into a renormalized
scalar field theory in curved spacetimes by
Brunetti and Fredenhagen~\cite{Brunetti2}, followed by
Hollands and Wald~\cite{Hollands2} and later extended to more general fields, for instance Dirac fields~\cite{Kratzert-00,Hollands-01,Antoni-06,
Dappiaggi-09,Sanders-10-Dirac,Rejzner-11},
gauge fields~\cite{Hollands-08,RejznerYangMills,RejznerYangMills2}
and even some attempts of quantization of gravitation~\cite{RejznerBF}. Moreover, recent simplifications strongly improved the mathematical understanding of the theory \cite{KM}.

Following those developments, the natural space where quantum field
theory seem to take place is not the space of distributions
$\mathcal{D}'$, but the space $\mathcal{D}'_\Gamma$ of distributions having
their wave front set in a 
specified closed cone $\Gamma$.
This space and its simplest properties were described
by H{\"o}rmander in 1971~\cite{Hormander-71}. Since recent developments \cite{DutschBF}, most papers in algebraic quantum field theory used microcausal functionals where the natural space to control the wave front set is rather the dual of the previous space $\mathcal{E}'_\Lambda$ with control of the wave front set by an open cone $\Lambda=-\Gamma^c$, we started recently in \cite{BrouderDabrowski} the investigation of functional analytic properties of these spaces.  The first paper of this series \cite{Dab14a} then computed their completion and bornologification.

We called ``dual wave front set" the union of usual $H^s$-wave front sets (see e.g. \cite[p8]{Delort}) of  a distribution : $DWF(u)=\bigcup_{s\in \R}WF_s(u)$ so that $WF(u)=\overline{DWF(u)}$ is recovered as its closure (see  \cite{Dab14a} for more details). With this definition the completion $\widehat{\mathcal{E}'_\Lambda}$ is nothing but : $$\widehat{\mathcal{E}'_\Lambda}=\{u\in \mathcal{E}', DWF(u)\subset \Lambda\}.$$
We thus introduced and studied the main functional analytic properties of the following spaces, for $\gamma\subset \Lambda\subset \overline{\gamma}$ cones, we define :
$$\mathcal{E}'_{\gamma,\Lambda}(U)=\{u\in \mathcal{E}'(U) : DWF(u)\subset \gamma, WF(u)\subset \Lambda\},$$
$$\mathcal{D}'_{\gamma,\Lambda}(U)=\{u\in \mathcal{D}'(U) : DWF(u)\subset \gamma, WF(u)\subset \Lambda\}.$$

This class of spaces is stable by topological and bornological duality, completion, bornologification and their main properties  established in \cite{Dab14a} are summarized in section 1.1 below for the reader's convenience.

After this general study, our second goal is to provide tools to study spaces of functionals with wave front set conditions on their derivatives. For this we gather  miscellaneous results on tensor products of our spaces and spaces of multilinear maps in part 1 of this second paper in the series. The third paper of this series will investigate more systematically those tensor products, following the general advice of Grothendieck, and mostly because general results don't give a full understanding of those tensor products. As for the completion above, there is a need for concrete (microlocal) representation to obtain a fully satisfactory functional analytic understanding. In this paper, we will be content to generalize hypocontinuity results of \cite{BrouderDangHelein} to our spaces in proposition \ref{hypocontinuity}. We will deduce some non-continuity results showing the use of hypocontinuity is unavoidable in proposition \ref{NONContinuity}. This is to be contrasted with vague statements of ``continuity" meaning sequential continuity spread out in the literature, and which become dangerous when mixed with projective tensor products statements which are related to full continuity.
We will finally give general relations of tensor products with our generalized H\"{o}rmander spaces. This will be crucial to use standard information on propagators and to relate our functionals with usual microcausal functionals of algebraic quantum field theory. Finally, we define in section 4 composition maps on spaces of multilinear maps. These spaces thus form a kind of topological operad.

Concerning spaces of functionals, we first solve technical problems in getting complete nuclear topologies on multilocal functionals and variants of microcausal functionals in  Theorem \ref{AppliedFA}. This problem remained open in \cite{RibeiroBF} when some nuclear topology was found on ordinary microcausal functionals.

However, we go beyond this functional analytic study of already used spaces to suggest a more algebraic and functional analytic approach on them to build efficiently Poisson algebras in infinite dimension with at least hypocontinuity results for the binary operations. Our idea about controlling, by wave front set conditions on derivatives, spaces of functionals is that it is much easier to control derivatives in spaces of multilinear maps on our spaces rather than by our spaces themselves. This especially restore some continuity of products (instead of hypocontinuity, in the full support case, see Theorem \ref{AppliedFA}) and this reduces definition of maps like retarded or  Poisson products (Peierls brackets) even with field dependent propagators to applications of composition of multilinear maps, see theorem \ref{AppliedComposition}. We emphasize that our theorem is a scheme of result to build various maps on spaces of functionals. Especially, as an application, we give a construction of retarded products in the case of field dependent propagators, fixing an issue in \cite{RibeiroBF}. This is the goal of section \ref{Retarded}.

Let us now describe in more detail the content of this paper, as a guiding summary of our tools for the reader.

Section 1 gathers preliminary material, mostly coming from \cite{Dab14a}, notably in subsection 1.1.
Subsection 1.2 recalls various notation and definitions of tensor products we will use extensively later.

Subsection 1.3 is new and completes the proofs of some results already  stated in our first paper on the series. It proves approximation properties for our spaces, identifies completions with more concrete quasi-completions and give a crucial property of vector valued distributions based on our spaces to be tested scalarly by duality, the so-called property $(\epsilon)$ of Schwartz.

As explained before, part 1 starts the study of tensor products building on previous results in \cite{BrouderDangHelein} and our general functional analytic properties. Section 2 contains continuity, and hypocontinuity results (most are generalizations of those in the quoted paper but one, proposition \ref{hypocontinuityImproved}, is strictly stronger even in the closed cone case they consider). Those results are based on known stability properties of hypocontinuity. This section also contains the advertised non-continuity result. The issue comes from zero sections where cones are not the most natural set of control for classical tensor products. The non-continuity result is thus based on our improved continuity and proves using duality results and nuclearity that continuity would imply an isomorphism of the hypocontinuous tensor product with some space of distributions. But plenty of distributions in this space, typically an example of Hormander with wave front set contained in one direction, with one side on zero sections, cannot come from the image of the hypocontinuous product. Since the cones considered in this non-continuity result are of the form appearing in physics for microcausal functionals, this result is of physical significance.

Section 3 gives inclusions between $\epsilon$-products and our generalized H\"ormander spaces of distributions, in waiting for an exact microlocal characterization of the former in the third paper of the series. 
Section 4 is concerned with multilinear maps.

 Part 2 gives the already described applications to functionals. Schwartz' papers on vector valued distributions \cite{Schwartz3} play a key role here as well as the notion of convenient smoothness \cite{KrieglMichor} and our composition of multilinear maps from section 4. Section 5 defines our general spaces of functionals and gives completeness, nuclearity and results of (hypo)continuity of products. Section 6 contains a general theorem to construct hypocontinuous bilinear maps and  section 7 applies this to retarded products.

\medskip 

\subsection*{Acknowledgments} The author is grateful to Katarzyna Rejzner and Christian Brouder for helpful discussions on the physical motivation of part 2. He also thanks Christian Brouder for many comments on previous versions of this paper that helped improving its exposition. 

Finally, the author acknowledges the support and hospitality of the Erwin Schr\"odinger Institute during the workshop ``Algebraic Quantum Field Theory: Its Status and Its Future" in May 2014. He also acknowledges the organizers for the stimulating program. Again, he thinks (and hopes) the physical relevance of the content of part 2 greatly benefited from the participation at this workshop.

\section{Preliminaries}
We start by recalling the setting of the first paper of the series without giving the definition of the topologies but we still state the main results we will most often use in this second paper.

Let $U\subset M$ an open set in a smooth connected manifold (implicitly assumed orientable 
 $\sigma$-compact without boundary) of dimension $d$.
We assume given on $M$ a complete Riemannian metric $D$ giving the topology (so that, by Hopf-Rinow theorem, closed balls for $D$ are compact.)

 Let $E\mapsto U$ a smooth real vector bundle (with finite dimensional fiber of real dimension $e$).

We will use a fixed smooth partition of unity $f_i$, indexed by $I$, subordinated to a covering $U_i$, with $\overline{U_i}\subset U$ compact and $(\overline{U_i})_{i\in I}$ locally finite 
 $U_i$ smoothly isomorphic via a chart $\varphi_i:U_i\to \R^d$ to an open set in $\R^d$, extending homeomorphically to $\overline{U_i}$ and  trivializing  $E\mapsto U$ 

 We consider $\mathcal{D}'(U;E)$ the space of distributional sections of $E$ with support in $U$ and $\mathcal{E}'(U;E)$ the space of distributional sections of $E$ with compact support in $U$.  Note that $\mathcal{D}'(U_i,E)\simeq \mathcal{D}'(\varphi_i(U_i),\R^e)\simeq (\mathcal{D}'(\varphi_i(U_i)))^e $ (see e.g. \cite[(3.16) p 234]{GrosserKOS}), via the map written above $u\mapsto u\circ \varphi_i^{-1}$.
 Likewise we call  $\mathcal{E}(U;E)$ the space of smooth sections of $E$ on $U$ and $\mathcal{D}(U;E)$ the space of smooth sections of $E$ with compact support in $U$. We write as usual $E'$ for the dual bundle, and $E^*=E'\otimes \R^*$ its version twisted by the density bundle $\R^*$ (cf e.g. \cite[chapter 3]{GrosserKOS}). We of course don't write $E$ in any notation when $E$ is the trivial line bundle over $\R$. Note that, since we don't make explicit our trivialization of vector bundles, we make the choice for those of $E,E^*$  so that for $g\in\mathcal{D}(U_i),v\in \mathcal{D}(U)$: 
$$\langle f_iu, gv\rangle=\langle (f_iu)\circ \varphi_i^{-1},(gv)\circ \varphi_i^{-1}\rangle.$$

This reduces duality pairings to those on $\R^n$ and the emphasis on the difficult analytic part of pullback rather than on the multiplication by smooth map part involved in change of bundle trivialization.

\subsection{Duality and functional analytic results} Recall that all the definition of topologies are given in section 3 there but they are also characterized by some of the properties given in the next result. The notions related to support properties are defined there in section 2. The most used support conditions $\mathcal{C}$ will be $\mathcal{K}$ the family of compact sets, $\mathcal{F}$ the family of all closed sets and on globally hyperbolic manifolds those explained in \cite[Ex 16]{Dab14a} (cf also \cite{Sanders}) :
timelike-compact closed sets $\mathcal{K}_T=\mathcal{K}_F\cap \mathcal{K}_P$, future-compact closed sets $\mathcal{K}_F,$ past-compact closed sets $\mathcal{K}_P,$ spacelike-compact closed sets $\mathcal{SK},$ future-spacelike-compact 
   closed sets $\mathcal{SK}_F$ and  past-spacelike-compact 
    closed sets $\mathcal{SK}_P$. In these cases the class $(\mathscr{O}_\mathcal{C})^o$ in duality formulas is described as follows $(\mathscr{O}_\mathcal{K})^o=\mathcal{F}, (\mathscr{O}_\mathcal{F})^o=\mathcal{K}$, $\mathcal{SK}_F= (\mathscr{O}_{\mathcal{K}_P})^o$, $\mathcal{SK}_P= (\mathscr{O}_{\mathcal{K}_F})^o,$
 $\mathcal{K}_P= (\mathscr{O}_{\mathcal{SK}_F})^o$, $\mathcal{K}_F= (\mathscr{O}_{\mathcal{SK}_P})^o, $ $(\mathscr{O}_{\mathcal{SK}})^{o}=(\mathscr{O}_{\mathcal{SK}_F})^{o}\cap (\mathscr{O}_{\mathcal{SK}_P})^{o}=\mathcal{K}_T,(\mathscr{O}_{\mathcal{K}_T})^{o}
 =\mathcal{SK}.$ They satisfy the assumptions below since they are all enlargeable, by definition polar, and either $\mathcal{C}$ or $(\mathscr{O}_\mathcal{C})^o$ is countably generated in the sense of \cite[Ex 3, Def 13]{Dab14a}.
 
 Recall the notation for our spaces in the vector bundle case :
$$\mathcal{D}'_{\gamma,\Lambda}(U,\mathcal{C};E)=\{u\in \mathcal{D}'(U;E)\ | \ WF(u)\subset \Lambda , DWF(u)\subset \gamma, \text{supp}(u)\in \mathcal{C}\}.$$ 

The following theorem is our main result from the first part \cite[Prop 34]{Dab14a}.

\begin{theorem}\label{FAGeneral2}Let $\gamma$ 
a cone  and  $\mathcal{C}=\mathcal{C}^{oo}$ an enlargeable polar family of closed sets in $U$ and let $\lambda=-\gamma^c.$
 The bounded sets on $\mathcal{D}'_{\gamma,\overline{\gamma}}(U,\mathcal{C};E)$ coincide for $\mathcal{I}_{ppp}$ and $\mathcal{I}_{iii}$ and this last inductive limit is regular. 
Moreover on $\mathcal{D}'_{\lambda,\lambda}(U,(\mathscr{O}_\mathcal{C})^o;E^*)$ we have $\mathcal{I}_{b}=\mathcal{I}_{ibi}=\mathcal{I}_{pbp}.$ 
  $(\mathcal{D}'_{\lambda,\overline{\lambda}}(U,(\mathscr{O}_\mathcal{C})^o;E^*),\mathcal{I}_{ibi})$ is the strong and Mackey dual of $(\mathcal{D}'_{\gamma,\overline{\gamma}}(U,\mathcal{C};E), \mathcal{I}_{iii}^{born}=\mathcal{I}_{ppp}^{born})$, the bornologification $\mathcal{I}_{ppp}^{born},$ the completion of $(\mathcal{D}'_{\lambda,{\lambda}}(U,(\mathscr{O}_\mathcal{C})^o;E^*),\mathcal{I}_{b})$ and also for any cone $\lambda\subset \Lambda\subset \overline{\lambda}$ the completion of $(\mathcal{D}'_{\lambda,{\Lambda}}(U,(\mathscr{O}_\mathcal{C})^o;E^*),\mathcal{I}_{ibi})$, which is also nuclear. 
  
  Thus, $(\mathcal{D}'_{\lambda,\overline{\lambda}}(U,(\mathscr{O}_\mathcal{C})^o;E^*),\mathcal{I}_{ibi})$ is  complete ultrabornological nuclear, 
especially a Montel space. Likewise, $(\mathcal{D}'_{\lambda,\lambda}(U,(\mathscr{O}_\mathcal{C})^o;E^*), \mathcal{I}_b=\mathcal{I}_{ibi})$ is the bornologification of $\mathcal{I}_{iii},$ so that the situation is summarized in the two following commuting diagrams where $i$ is the canonical injection from a space to its completion, $b$ the canonical map  with bounded inverse between a bornologification and the original space :
$$\begin{diagram}[inline]
&&(\mathcal{D}'_{\gamma,\overline{\gamma}}(U,\mathcal{C};E),\mathcal{I}_{ibi})
\\
&\ldTo^b&\uInto^i 
\\
(\mathcal{D}'_{\gamma,\overline{\gamma}}(U,\mathcal{C};E),\mathcal{I}_{ppp})& &(\mathcal{D}'_{\gamma,{\gamma}}(U,\mathcal{C};E),\mathcal{I}_{ibi})
\\
\uInto^i&\ldTo^b& 
\\
(\mathcal{D}'_{\gamma,{\gamma}}(U,\mathcal{C};E),\mathcal{I}_{ppp})& &
\end{diagram}\quad \quad \quad \begin{diagram}[inline]
(\mathcal{D}'_{\lambda,\overline{\lambda}}(U,(\mathscr{O}_\mathcal{C})^o;E^*),\mathcal{I}_{ibi})&&
\\
\dTo^b& \luInto^i&
\\
(\mathcal{D}'_{\lambda,\overline{\lambda}}(U,(\mathscr{O}_\mathcal{C})^o;E^*),\mathcal{I}_{ppp})& &(\mathcal{D}'_{\lambda,{\lambda}}(U,(\mathscr{O}_\mathcal{C})^o;E^*),\mathcal{I}_{ibi})
\\
&\luInto^i&\dTo^b
\\& & 
(\mathcal{D}'_{\lambda,{\lambda}}(U,(\mathscr{O}_\mathcal{C})^o;E^*),\mathcal{I}_{ppp})
\end{diagram}$$

Spaces symmetric with respect to the middle vertical line are Mackey duals of one another. All the spaces involved are nuclear locally convex spaces, and in each of them, bounded sets which are closed in the completion are metrisable compact sets and are equicontinuous sets from the stated dualities. When $\lambda,\gamma$ are $\mathbf{\Delta_2^0}$-cones (i.e. both $F_\sigma$,$G_\delta$) and if we assume
either $\mathcal{C}$ or $(\mathscr{O}_\mathcal{C})^o$  countably generated, all the space involved are moreover quasi-LB spaces of class $\mathfrak{G}$.  

\end{theorem}

The following extension of \cite[Prop 6.1]{BrouderDangHelein} was proven in \cite[Prop 36]{Dab14a}.
\begin{proposition}\label{pullback}Let $\gamma\subset \Lambda\subset \overline{\gamma}$ be cones on $U_2$ and $f:U_1\to U_2$ a smooth map. Define the cone  $df^*\gamma=\{(x,df^{*}(x)(\xi)):(f(x),\xi)\in\gamma\}$ and  for an enlargeable polar family of closed sets $\mathcal{C}$ define $f^{-1}(\mathcal{C})=\{f^{-1}(C), C\in \mathcal{C},\},$ and its polar enlargeable variant $f_e^{-1}(\mathcal{C})=\{(f^{-1}(C))_{(1-1/n)\epsilon(C)}, C\in \mathcal{C},\epsilon_i(C)>0\}^{oo}$ (depending on any function $\epsilon:\mathcal{C}\to ]0,1[^I$)

Assume $df^*\gamma\subset \dot{T}^*U_1$, and $ df^*\Lambda\subset \dot{T}^*U_1$. Then we have  continuous maps for $\mathcal{I}$ either $\mathcal{I}_{ppp}$ or $\mathcal{I}_{ibi}$ : $$f^*:(\mathcal{D}'_{\gamma,{\Lambda}}(U_2,\mathcal{C}),\mathcal{I})
\to (\mathcal{D}'_{df^*\gamma,df^*{\Lambda}}(U_1,f_e^{-1}(\mathcal{C})),\mathcal{I}),$$
$$f^*:(\mathcal{D}'_{\gamma,\overline{\gamma}}(U_2,\mathcal{C}),\mathcal{I})
\to (\mathcal{D}'_{df^*\gamma,\overline{df^*{\gamma}}}(U_1,f_e^{-1}(\mathcal{C})),\mathcal{I}).$$
Especially, $DWF(f^*u)\subset df^*(DWF(u)).$
\end{proposition}

We will also need some identification of topologies obtained by computing equicontinuous sets in respective duals,  we obtained them  in \cite[lemma 28]{Dab14a} for $\lambda,\gamma\subset\Lambda\subset \overline{\gamma}$ any cones:
\begin{equation}\label{Alternativeppp}(\mathcal{D}'_{\gamma,\Lambda}(U,\mathcal{C};E),\mathcal{I}_{H,p pp})\simeq\underleftarrow{\lim}_{ \lambda_1\supset \gamma}\underleftarrow{\lim}_{\lambda_2\supset \lambda_1\cup \Lambda}\left(\mathcal{D}'_{\lambda_1,\overline{\lambda_1}\cap \lambda_2}(U,\mathcal{C};E),\mathcal{I}_{H,pmp}\right),\end{equation}
\begin{equation}\label{Alternativepip}(\mathcal{D}'_{\gamma,\overline{\gamma}}(U,\mathcal{C};E),\mathcal{I}_{H,p ip})\simeq\underleftarrow{\lim}_{ \lambda_1\supset \gamma}\left(\mathcal{D}'_{\lambda_1,\overline{\lambda_1}}(U,\mathcal{C};E),\mathcal{I}_{H,pip}\right),\end{equation}
\begin{equation}\label{Alternativeiii}(\mathcal{D}'_{\lambda,\lambda}(U,(\mathscr{O}_\mathcal{C})^o;E^*),\mathcal{I}_{iii})\simeq\underrightarrow{\lim}_{\Pi\subset \lambda}(\mathcal{D}'_{\Pi,\Pi}(U,(\mathscr{O}_\mathcal{C})^o;E^*),\mathcal{I}_{iii}).\end{equation}

Here $\lambda_i$ indexes projective limits over open cones, $\Pi$ inductive limits over closed cones.

We finally gather in the next lemma the descriptions of bounded, absolutely convex compact and equicontinuous sets in the spaces above which are scattered  in the proofs of \cite{Dab14a}. These are the most relevant bornologies for the tensor products we will consider. Recall that we introduced in \cite[section 3]{Dab14a} the following technical space $$\mathcal{D}'_{\overline{\gamma}}(U,\mathcal{F}:\delta;E)=\{u\in \mathcal{D}'(U;E)\ | \ 
\forall n ,DWF_n(u)\subset \Gamma_n \},$$ for an increasing sequence of closed cones $\Gamma_n\subset \gamma$ gathered in $\delta=(\Gamma_n)$. Modulo, equivalence by inclusion, this is basically a space where we fix the $H^n $ wave front set (which is similar to $DWF_n$, but different) in a family of closed cones included in the cone $\gamma.$

\begin{lemma}\label{BoundedRelated}With the setting of theorem \ref{FAGeneral2}, on 
$(\mathcal{D}'_{\gamma,\overline{\gamma}}(U,\mathcal{C};E),\mathcal{I}_{ibi})$ and 
$(\mathcal{D}'_{\gamma,\overline{\gamma}}(U,\mathcal{C};E),\mathcal{I}_{ppp}),$
bounded sets, precompact sets, sets included in an absolutely convex compact set and equicontinuous sets induced by their duals coincide. We have for them the same description as for bounded sets for
$(\mathcal{D}'_{\gamma,{\gamma}}(U,\mathcal{C};E),\mathcal{I}_{b})$ and 
$(\mathcal{D}'_{\gamma,{\gamma}}(U,\mathcal{C};E),\mathcal{I}_{iii})$, namely sets bounded in the sense of the bornology induced by the regular inductive limit $(\mathcal{D}'_{\gamma,\overline{\gamma}}(U,\mathcal{C};E),\mathcal{I}_{iii})$, i.e. uniformly supported in a $C\in \mathcal{C}$ and included and bounded for some $\delta=(\Gamma_n)\in \Delta(\gamma)$ in $\mathcal{D}'_{\overline{\gamma}}(U,\mathcal{F}:\delta;E).$ 

On $(\mathcal{D}'_{\gamma,{\gamma}}(U,\mathcal{C};E),\mathcal{I}_{b})$ and 
$(\mathcal{D}'_{\gamma,{\gamma}}(U,\mathcal{C};E),\mathcal{I}_{iii})$, equicontinuous sets from their dual coincides with sets included in absolutely convex compact sets which are those associated to the bornological inductive limit associated to $\mathcal{I}_{iii}$, namely sets of distributions uniformly supported on some $C\in \mathcal{C}$ and bounded in $\mathcal{D}'_{\Pi,{\Pi}}(U,\mathcal{F};E)$ for some closed cone $\Pi\subset \gamma.$

Finally, the bornologification maps $b:\mathcal{D}'_{\gamma,\overline{\gamma}}(U,\mathcal{C};E),\mathcal{I}_{ibi}\to \mathcal{I}_{ppp}, \mathcal{D}'_{\gamma,{\gamma}}(U,\mathcal{C};E),\mathcal{I}_{ibi}\to \mathcal{I}_{ppp}$ are proper.
\end{lemma}

Especially, all the spaces above are defined from those bounded in a bornological inductive limit associated to $\mathcal{I}_{iii}$ on some space.

\begin{proof}
From the beginning of the proof of \cite[Prop 34]{Dab14a} a set bounded in  $(\mathcal{D}'_{\gamma,\overline{\gamma}}(U,\mathcal{C};E),\mathcal{I}_{ppp})$ is bounded in the regular inductive limit $\mathcal{I}_{iii}$ which corresponds to the  description above. For general reasons about bounded sets in completion and bornologification, this gives the description of bounded sets in any of the spaces above. Since from \cite[Prop 34]{Dab14a}, $(\mathcal{D}'_{\gamma,\overline{\gamma}}(U,\mathcal{C};E),\mathcal{I}_{ibi})$ is nuclear, bounded sets are exactly precompact=``relatively compact" sets. Thus the inverse image of a compact set for $\mathcal{I}_{ppp}$ is closed since the bornologification map is continuous, bounded since both topologies have same bounded sets thus compact by the above nuclearity. This gives that the bornologification map is proper. For the second proper map we reason by diagram chasing, a compact set is also compact in the completion thus from the complete case in the bornologification, and since the topology of the bornologification without completion agrees with the one of the completeion of the bornologification (for our spaces), it is compact there.

In $\mathcal{D}'_{\gamma,\overline{\gamma}}(U,\mathcal{C};E)$, for  $\mathcal{I}_{ibi}$ equicontinuous sets and strongly bounded sets coincides as in any dual of a bornological space (this is explained from Hogbe-Nlend's results  in \cite[lemma 27]{Dab14a}), and moreover, $\mathcal{I}_{ibi}$  is known to be the strong dual topology by \cite[Prop 34]{Dab14a}. Finally, on a dual, the equicontinuous sets induced by a space and its completion are known to be the same thus the equicontinuous sets are the same for $\mathcal{I}_{ppp}.$

For $(\mathcal{D}'_{\gamma,{\gamma}}(U,\mathcal{C};E),\mathcal{I}_{iii})$ equicontinuous sets are identified in \cite[lemma 28]{Dab14a}, and as above they are the same for $\mathcal{I}_{ibi}$ whose dual is the completion of the dual. In $(\mathcal{D}'_{\gamma,{\gamma}}(U,\mathcal{C};E),\mathcal{I}_{b})$ and 
$(\mathcal{D}'_{\gamma,{\gamma}}(U,\mathcal{C};E),\mathcal{I}_{iii})$ closed equicontinuous sets are compact as in any dual of a nuclear space \cite[p 519]{Treves}, and conversely absolutely convex compact sets (which coincides for both topologies from the proper map property shown above) are shown to be  equicontinuous at the end of the proof of \cite[Prop 34]{Dab14a}.
This completes our proof.
\end{proof}

\subsection{Reminder on tensor products of locally convex spaces.}

We will be mostly interested in five tensor products, projective, injective, $\epsilon$, $\gamma$ and bornological (with variants).

If $E,F$ are locally convex spaces, $E\otimes_\pi F $ (resp. $E\otimes_\beta F $, $E\otimes_{\beta e} F $ when $E=E_1',F=F_1'$ are duals, $E\otimes_{\gamma} F $) is the algebraic tensor product equipped with the finest locally convex topology on it making $E\times F\to E\otimes_\pi F $ continuous (resp $E\times F\to E\otimes_\beta F $ hypocontinuous, resp. $E\times F\to E\otimes_{\beta e} F $ $\epsilon$-hypocontinuous, i.e. hypocontinuous for equicontinuous parts of the implicitly used dualities $E=E_1',F=F_1'$, resp. $E\times F\to E\otimes_{\gamma} F $ $\gamma$-hypocontinuous, i.e. hypocontinuous for absolutely convex compact parts). We write as usual $E\hat{\otimes}_\pi F , E\hat{\otimes}_\beta F, E\hat{\otimes}_{\beta e} F, E\hat{\otimes}_\gamma F $ the corresponding completions. We will also sometimes use the terminology $\sigma_1-\sigma_2$ hypocontinuous (especially with $\sigma_i$ one of the letters above) to say hypocontinuous for bounded sets in $\sigma_i$ on the corresponding side \cite{Schwartz3}.

One should note that at various places in the (bornology, convenient vector space) literature, see e.g. \cite{KrieglMichor}, there is another (different) definition of the bornological tensor product as the finest l.c.s topology making $E\times F\to E\otimes_\beta' F $ bounded (on product of bounded sets). From the characterization of neighborhoods of zero (see e.g. \cite[Prop 11.3.4 p 390]{PerrezCarreras}), it is easy to see that $E\otimes_\beta' F=E^{born}\otimes_\beta F^{born}$ is merely the bornological/hypocontinuous product of bornologifications.

If $E,F$ are barrelled (resp. bornological) so are $E\otimes_\beta F, E\otimes_{\beta e} F,E\otimes_{\gamma} F $ \cite[Prop 11.3.6 p 390]{PerrezCarreras} \cite[Th 5]{DefantGovaerts}  and  $E\hat{\otimes}_\beta F,E\hat{\otimes}_\gamma F $ are barrelled in the barelled case \cite[Prop 4.2.1 p 103]{PerrezCarreras}.

The $\epsilon$-product has been extensively used and studied by Laurent Schwartz \cite[section 1]{Schwarz}.
By definition $E\epsilon F= (E'_c\otimes_{\beta e} F'_c)'$ is the set of $\epsilon$-hypocontinuous bilinear forms on the duals $E'_c$ equipped with the Arens topology 
 of uniform convergence on absolutely convex compact subsets of $E.$
 
 When the bounded sets on $E'_c,F'_c$ coincide with the equicontinuous parts this is of course the same as $E\epsilon F= (E'_c\otimes_{\beta} F'_c)'$ and when $E,F$ have their $\gamma$ topology (i.e. $(E'_c)'_c$ so that as explained in \cite{Schwarz} equicontinuous sets of the dual coincide with absolutely convex compact sets), this is the same as $E\epsilon F= (E'_c\otimes_{\gamma} F'_c)'.$ Note that all our spaces above in theorem \ref{FAGeneral2} have their $\gamma$ topology since they have all their Mackey topology. 
  The topology on $E\epsilon F$ is the topology of uniform convergence on products of equicontinuous sets in $E', F'$ (as the initial topology of $E=(E'_c)'$, by Mackey Thm is the topology of uniform convergence on equicontinuous parts in $E'$).
If $E,F$ are quasi-complete spaces (resp. complete spaces , resp. complete spaces  with the approximation property) 
so is $E\epsilon F$ (see \cite[Prop 3 p29, Corol 1 p 47]{Schwarz}).
 
$E\otimes_\epsilon F$ is the topology on $E\otimes F$ induced by $E\epsilon F$ (see \cite[Prop 11 p46]{Schwarz} and $E\hat{\otimes}_\epsilon F\simeq E\epsilon F$ if $E,F$ are complete and $E$ has the approximation property. All the above tensor products are commutative and the ${\otimes}_\pi,{\otimes}_\epsilon, {\otimes}_\beta'$ products are associative.

We end this section by a technical result in our context that will be used at the end in the application section 7. This gives a glimpse of the general study of tensor products in the third paper of this series.

\begin{proposition}\label{QLBEpsilon}Let $\gamma 
\subset \dot{T}^*U$ a $\mathbf{\Delta_2^0}$-cone and $\mathcal{C}$ an enlargeable polar family of closed sets of $U$ with either $\mathcal{C}$ or $(\mathscr{O}_\mathcal{C})^o$ countably generated, and similarly $\Gamma\subset \dot{T}^*V$ a closed cone on a similar open $V$. Let $E\to U,F\to V$ vector bundle maps and $\mathcal{I}\in \{\mathcal{I}_{ppp},\mathcal{I}_{ibi}\}$
 then $(\mathcal{D}'_{\gamma,\overline{\gamma}}(U,\mathcal{C};E),\mathcal{I})\epsilon (\mathcal{D}'_{\Gamma,\Gamma}(U,\mathcal{F};F),\mathcal{I}_{ppp})$ is a quasi-LB space thus a strictly webbed space.
\end{proposition}

\begin{proof}
From Theorem \ref{FAGeneral2}, $(\mathcal{D}'_{\gamma,\overline{\gamma}}(U,\mathcal{C};E),\mathcal{I})$ is a complete quasi-LB space thus a strictly webbed space \cite[Th 1]{ValdiviaquasiLB}. As in \cite[Corol 13]{BrouderDabrowski} (using the (PLN) property showed in \cite[Prop 8]{Dab14a}
), it is easy to show that   $(\mathcal{D}'_{\Gamma,\Gamma}(U,\mathcal{F};F),\mathcal{I}_{ppp})$ is a (PLS)-space thus, by definition, can be written $\underleftarrow{\lim}_{n\in \N}X_n$ with $X_n$ the dual of a Fr\'echet-Schwartz space. Now by completeness and nuclearity,  $(\mathcal{D}'_{\gamma,\Lambda}(U,\mathcal{C};E),\mathcal{I})\epsilon (\mathcal{D}'_{\Gamma,\Gamma}(U,\mathcal{F};F),\mathcal{I}_{ppp})\simeq \underleftarrow{\lim}_{n\in \N}(\mathcal{D}'_{\gamma,\Lambda}(U,\mathcal{C};E),\mathcal{I})\hat{\o}_\epsilon X_n$. Thus this is a countable inductive limit of webbed spaces (using \cite[Prop V.2.4 p 91]{DeWilde} for the tensor product and e.g. \cite[\S 45.4.(7) p 63]{Kothe2} for the countable projective limit). From \cite[Th 2]{ValdiviaquasiLB}, since our space is complete, it is also quasi-(LB).
\end{proof}

\subsection{Consequences for approximation properties and vector valued distributions.}
We gather here results that are useful in themselves in functional analysis, but essential when dealing with vector valued distributions \cite{Schwarz}. The study of vector valued distributions based on H\"ormander spaces of distributions is motivated by their use in quantum field theory, notably in \cite{Rejzner,RejznerYangMills}. We will also use them in part 2 below.

\begin{proposition}\label{ApproximationProp}
Let $\gamma\subset \Lambda\subset \overline{\gamma}$ cones  and  $\mathcal{C}=\mathcal{C}^{oo}$ an enlargeable polar family of closed sets in $U$ and assume
either $\mathcal{C}$ or $(\mathscr{O}_\mathcal{C})^o$  countably generated. Then $(\mathcal{D}'_{\gamma,\overline{\gamma}}(U,\mathcal{C};E),\mathcal{I}_{ibi})$, 
$(\mathcal{D}'_{\gamma,\overline{\gamma}}(U,\mathcal{C};E),\mathcal{I}_{ppp})$,$(\mathcal{D}'_{\gamma,{\gamma}}(U,\mathcal{C};E),\mathcal{I}_{b})$ and 
$(\mathcal{D}'_{\gamma,{\gamma}}(U,\mathcal{C};E),\mathcal{I}_{iii})$ all have the sequential approximation property (in the strongest variant with uniform convergence on compact subsets of the completion \cite[p399]{Jarchow}).
\end{proposition}
\begin{proof}
We use \cite[Prop 2 p 7]{Schwarz}. Of course from the same proof as in \cite[Prop 1 p 6]{Schwarz}, $F=\mathcal{D}(U;E)$ has the wanted sequential approximation property (since it is complete, the uniform convergence on compact sets of the completion coincide with Schwartz' notion of convergence on absolutely convex compact sets and Grothendieck's notion of convergence on precompact sets). If $\mathcal{C}$ is countably generated, \cite[lemma 22]{Dab14a} 
 gives us a subsequence $L_{k_n}$ of $L_n$ such that for any $B$ bounded in
$G=(\mathcal{D}'_{\gamma,\overline{\gamma}}(U,\mathcal{C};E),\mathcal{I}_{ibi}=\mathcal{I}_{iii}^{born})$, $L_{k_n}(B)$ is bounded in the same sense, and $L_{k_n}(u)\to u$. Thus since $(\mathcal{D}'_{\gamma,\overline{\gamma}}(U,\mathcal{C};E),\mathcal{I}_{ibi})$ is barrelled,  from \cite[p 141]{Kothe2} $L_{k_n}\to Id$ uniformly on compact sets of $G.$ This gives the expected sequential approximation of $Id$ by $L(G,F)$ in $L(G,G)$ in this case. Composing with the bornologification map at the target and knowing the stronger continuity of each map from \cite[lemma 22]{Dab14a}, and since the compact sets on which we want to converge uniformly are the same by the proper map property in lemma \ref{BoundedRelated}, this gives the same result for  $G=(\mathcal{D}'_{\gamma,\overline{\gamma}}(U,\mathcal{C};E),\mathcal{I}_{ppp}).$ Finally these two spaces are the completion of the other two, thus since the uniform convergence sets agree as the topology, we also conclude those cases. For the case where $(\mathscr{O}_\mathcal{C})^o$ is countably generated, we have the corresponding result for the dual of $(\mathcal{D}'_{\gamma,\overline{\gamma}}(U,\mathcal{C};E),\mathcal{I}_{ibi})$ whose Arens dual is still this space, thus the sequential approximation property in Schwartz' sense, which coincides with the stated one in this complete case. The consequence without bornologification and without completion then follows similarly.
\end{proof}

We gather here other consequences of \cite[lemma 22]{Dab14a}
.
\begin{proposition}\label{quasiCompletion}
Let $\gamma$ a cone  and  $\mathcal{C}=\mathcal{C}^{oo}$ an enlargeable polar family of closed sets in $U$. For any bounded set $B$ in $(\mathcal{D}'_{\gamma,\overline{\gamma}}(U,\mathcal{C};E),\mathcal{I}_{ibi})$, 
$(\mathcal{D}'_{\gamma,\overline{\gamma}}(U,\mathcal{C};E),\mathcal{I}_{ppp})$,$(\mathcal{D}'_{\gamma,{\gamma}}(U,\mathcal{C};E),\mathcal{I}_{b})$ and
$(\mathcal{D}'_{\gamma,{\gamma}}(U,\mathcal{C};E),\mathcal{I}_{iii})$, there is a set $B'\subset \mathcal{D}(U;E)$ and bounded in the same space with $B\subset \overline{B'}$ (even in the Mackey-sequential completion). 

As a consequence, $(\mathcal{D}'_{\gamma,\overline{\gamma}}(U,\mathcal{C};E),\mathcal{I}_{ibi})$ is the quasi-completion of $(\mathcal{D}'_{\gamma,{\gamma}}(U,\mathcal{C};E),\mathcal{I}_{b})$ and $(\mathcal{D}'_{\gamma,\overline{\gamma}}(U,\mathcal{C};E),\mathcal{I}_{ppp})$ the quasi-completion of $(\mathcal{D}'_{\gamma,{\gamma}}(U,\mathcal{C};E),\mathcal{I}_{iii}).$
\end{proposition}
\begin{proof}
In any case $B$ is always bounded in $(\mathcal{D}'_{\gamma,\overline{\gamma}}(U,\mathcal{C};E),\mathcal{I}_{ibi})$, and from \cite[lemma 22]{Dab14a}
 $B'=L_{k(B)}(B)$ is bounded for $\mathcal{I}_{iii}$ thus for $\mathcal{I}_{ibi}$ since this is the bornologification of the previous one. Thus $B'$ is bounded in any of these spaces and the conclusion follows from the lemma. 
\end{proof}

In \cite[p 53]{Schwarz}, Schwartz says a space of distributions $\mathscr{H}$ (in $\R^n$ but the manifold case is identical, i.e. $\mathscr{H}\hookrightarrow \mathcal{D}'(U,E)$) has property $(\epsilon)$ if for any quasi-complete separated locally convex vector space $F$, the $\epsilon$ product (cf. the previous subsection for a reminder) $\mathscr{H}\epsilon F$ is the same as the space of distributions $T\in \mathcal{D}'(U,E)\epsilon F\simeq L(F_c',\mathcal{D}'(U,E))$ such that for any $f\in F',$ $T(f)\in \mathscr{H}.$ (Recall $F_c'$ is the Arens dual of $F$, i.e. $F'$ given the topology of uniform convergence on compact absolutely convex sets in $F$). Said otherwise $\mathscr{H}$-valued distributions defined as $\mathscr{H}\epsilon F$  can be tested scalarly to belong to  $\mathscr{H}$ once known to be in $\mathcal{D}'(U,E)$-valued distributions.

\begin{proposition}\label{PropEpsilon}
Let $\gamma$ cones  and  $\mathcal{C}=\mathcal{C}^{oo}$ an enlargeable polar family of closed sets in $U$ and assume
either $\mathcal{C}$ or $(\mathscr{O}_\mathcal{C})^o$  countably generated. Then $(\mathcal{D}'_{\gamma,\overline{\gamma}}(U,\mathcal{C};E),\mathcal{I}_{ibi})$, 
and $(\mathcal{D}'_{\gamma,\overline{\gamma}}(U,\mathcal{C};E),\mathcal{I}_{ppp})$ have property $(\epsilon)$.
\end{proposition}

\begin{proof}
One uses (an obvious manifold variant of) \cite[proposition 15 p54]{Schwarz}. Both our spaces are quasi-complete nuclear thus closed bounded sets are compact. We know that them and their Arens=Mackey duals (since they are semi-Montel) are strictly normal (see \cite[lemma 22]{Dab14a} 
 and Theorem \ref{FAGeneral2} for the Mackey duals.) It remains to check they have a system of neighborhoods of zero which is closed in them with the topology induced by $\mathcal{D}'(U,E)$, for which, by Schwartz' remark after his proof, it suffices to check there is, for every equicontinuous set $B$ of their duals, a set $B'\subset \mathcal{D}(U,E^*)$ and equicontinuous in their duals such that $B\subset \overline{B'}$ the weak closure. For $(\mathcal{D}'_{\gamma,\overline{\gamma}}(U,\mathcal{C};E),\mathcal{I}_{ibi})$  for which equicontinuous sets are the same as bounded sets, this is our preceding proposition. In $(\mathcal{D}'_{\gamma,\overline{\gamma}}(U,\mathcal{C};E),\mathcal{I}_{ppp})',$ equicontinuous sets are in the bornological inductive limit associated to the definition of $ \mathcal{I}_b$ namely from the proof of Theorem \ref{FAGeneral2} again, the bounded sets in the bornological inductive limit associated to $\mathcal{I}_{iii}$ (this is also explained in detail in lemma \ref{BoundedRelated}). Thus from \cite[lemma 22]{Dab14a}
, we can take $B'=L_{k(B)}(B)$ (note this lemma only uses the inductive limit bornology).
\end{proof}

\part{A first look at tensor products and operators on Generalized H\"ormander spaces of distributions}
We start here our investigation of tensor products and spaces of operators on our (slightly) Generalized H\"ormander spaces of distributions. Our intent is to prove here only what we need for our applications to generalizations of spaces of functionals used in algebraic quantum field theory \cite{Brunetti2,DutschBF,Rejzner,RibeiroBF}.  A more systematic study will be carried out in the third paper of this series.

\section{Preliminaries on continuity of tensor products}
The strongest result is obtained when tensor product is not considered with optimal wave front set condition but we can get a continuous map (and not only a hypocontinuous map).

\begin{proposition}\label{continuity}
Let $\gamma_i
\subset \dot{T}^*U_i$ be  closed 
cones, $E_i\mapsto U_i$ fiber bundles. 
 
Let $\gamma=\dot{T}^*U-\gamma_1^c\times\gamma_2^c$, 
$U=U_1\times U_2$,
then the bilinear map $.\otimes. : (\mathcal{D}'_{\gamma_1,\gamma_1}(U_1,\mathcal{F};E_1),\mathcal{I}_{iii})\times (\mathcal{D}'_{\gamma_2,\gamma_2}(U_2,\mathcal{F};E_2),\mathcal{I}_{iii})\to (\mathcal{D}'_{\gamma,\gamma}(U,\mathcal{F};E_1\otimes E_2),\mathcal{I}_{iii})$ is continuous.
\end{proposition}
\begin{proof}
 and \cite[lemma 27]{Dab14a} 
, we can be content with the cases $\alpha=iii$.
We only consider trivial bundles on open sets of $\R^d,$ the general case being identical.

 Since the support  is $\mathcal{F}(U_i)$, we are reduced to the usual normal topology $\mathcal{I}_H$ from \cite{BrouderDabrowski}. From the well-known continuity of tensor product of distributions with their strong topology and exhaustion lemmas (any closed $\text{supp}(f)\times V\subset \gamma_1^c\times \gamma_2^c$ can be covered by finitely many parts as below), it suffices to consider seminorms of the form $P_{k,f_1\otimes f_2, V_1\times V_2}$ with $\text{supp}(f_i)\times V_i\subset \gamma_i^c$ and of course we have~: $$P_{k,f_1\otimes f_2, V_1\times V_2}(u_1\otimes u_2)=\sup_{(\xi,\eta)\in  V_1\times V_2}(1+|(\xi,\eta)|)^k |\mathcal{F}(f_1u_1)(\xi)||\mathcal{F}(f_2u_2)(\eta)|\leq P_{k,f_1, V_1}(u_1)P_{k,f_2, V_2}(u_2).$$\end{proof}

We also have a weaker but much more useful hypocontinuity result when we take optimal wave front set conditions, which essentially comes from \cite[Th. 5.5]{BrouderDangHelein}. This is especially important once considered our next result that explains why we don't have continuity in most cases different from the previous one.
Note that above and below $E_1\otimes E_2\to U_1\times U_2$ always denotes the exterior tensor product of vector bundles.

\begin{proposition}\label{hypocontinuity}
Let $\gamma_i\subset \Lambda^{(i)}\subset \overline{\gamma_i} 
\subset \dot{T}^*U_i$  
cones and $\mathcal{C}_i$ enlargeable polar families of closed sets of $U_i$. Let $E_i\to U_i$ vector bundle maps. 
Let $\pi: \dot{T}^*U_i\to U_i$ the first projection and then  let $$\gamma=\gamma_1\times\gamma_2\cup \gamma_1\times U_2\times\{0\}\cup U_1\times\{0\}\times \gamma_2=:\gamma_1\dot{\times}\gamma_2,$$ $\Lambda=\Lambda^{(1)}\dot{\times}\Lambda^{(2)}$ and $\mathcal{C}=(\mathcal{C}_1\times \mathcal{C}_2)^{oo}.$
The bilinear maps $$.\otimes. : (\mathcal{D}'_{\gamma_1,\Lambda^{(1)}}(U_1,\mathcal{C}_1;E_1),\mathcal{I}_{\alpha_1})\times \mathcal{D}'_{\gamma_2,\Lambda^{(2)}}(U_2,\mathcal{C}_2;E_2),\mathcal{I}_{\alpha_2})\to (\mathcal{D}'_{\gamma,\Lambda}(U_1\times U_2, \mathcal{C};E_1\otimes E_2),\mathcal{I}_{\beta})$$ are hypocontinuous (if not otherwise specified) in the following cases  :
\begin{enumerate}
\item[(i)]$\gamma_i=\Lambda^{(i)}$ closed cones and $\alpha_j=\beta=iii$,
\item[(ii)]
$\Lambda^{(i)},{\gamma_i}$ any cones as above and $\alpha_j=\beta=ibi.$ 
\item[(iii)]$\gamma_i=\Lambda^{(i)}$ any cone 
$\alpha_1=b,iii$ , $\alpha_2=\beta=iii$ in which case we only have $\gamma$-hypocontinuity, but full hypocontinuity if $\gamma_1$ closed for $\alpha_1=b$.
 \item[(iv)]$\gamma_1=\Lambda^{(1)}$ any cone and $\alpha_1=b$ and $\gamma_2,\Lambda^{(2)}$ any cones and $\alpha_2=\beta=ppp$ 
 in which case we only have $(\gamma-\beta)$-hypocontinuity, but full hypocontinuity if $\gamma_1$ closed.
\end{enumerate}

\end{proposition}
The various cases give (some kind of) hypocontinuity for most pairs of main topologies in Theorem \ref{FAGeneral2}, most notably for pairs of dual topologies and for identical topologies, with the exception of $\mathcal{I}_{ppp}$ in the complete case (where there is only a result with its dual topology), because the non-complete case is only given with $\gamma$-hypocontinuity and thus not enough bounded sets to go to the completion.

Note that case (ii) includes $\gamma_i=\Lambda^{(i)}$ open cones and $\alpha_j=\beta\in\{iii,iip,ipi,ipp\}$ since those topologies coincide with $\mathcal{I}_b$ in this case except on the target space since $\gamma$ may not be open, but it is still stronger (see \cite[Prop 33]{Dab14a}
) and for the same reason, case (ii) also includes $\overline{\gamma_i}=\Lambda^{(i)}$  with $\gamma_i$ open $\alpha_j=\beta=ppp$ in the case since $\mathcal{I}_{ppp}=\mathcal{I}_{ppp}^{born}=\mathcal{I}_{ibi}$ in this case from Theorem \ref{FAGeneral2}.
\begin{proof}
Note first that $\gamma^c=(\dot{T}^*U_1-\gamma_1)\times\dot{T}^*U_2\cup (\dot{T}^*U_1)\times (\dot{T}^*U_2-\gamma_2)\cup \gamma_1^c\times U_2\times\{0\}\cup U_1\times\{0\}\times \gamma_2^c$ so that $\gamma$ is again a $\mathbf{\Delta_2^0}$-cone if $\gamma_i$'s are (since this formula and the defining one keep the $F_\sigma$ property).

Using \cite[Th 23, lemma 27]{Dab14a} 
  to identify the topologies with what we called normal topology in \cite{BrouderDabrowski}, 
Case (i) is exactly \cite[Th. 5.5]{BrouderDangHelein} at least in the case $\mathcal{C}_i=\mathcal{F}(U_i)$, and trivial vector bundles. Note also that $\gamma$ is indeed a closed cone when $\gamma_i$'s are.  We don't explain more the manifold case which is a direct consequence of their result.

One can use a well-known result for hypocontinuity and inductive limits (see \cite[Th 1 c]{Melnikov}), since the inductive limit on support are strict regular (especially quasiregular in the sense of Melnikov), it suffices to get hypocontinuity on each term of the inductive limit, and since the topology is induced there from the one without support condition and the support condition of tensor products is well known, the result is obvious.

For case (ii),  we start with $\gamma_i=\Lambda^{(i)}$ in which case $\mathcal{I}_{ibi}=\mathcal{I}_{b}$ and
we proved in Theorem \ref{FAGeneral2} that the spaces considered are ultrabornological thus barrelled and we can again be content to prove separate continuity using \cite[Th 41.2 p424]{Treves}.
Thus if we take $u_1\in \mathcal{D}'_{\Pi_1,\Pi_1}(U_1,\mathcal{C}_1)$ and $\Pi_1\subset \gamma_1$ closed cone we have to check that $u_2\mapsto u_1\otimes u_2$ is continuous from  $(\mathcal{D}'_{\Pi_2,\Pi_2}(U_2,\mathcal{C}_2),\mathcal{I}_{H,iii}^{born})\to (\mathcal{D}'_{\gamma,\gamma=\Lambda}(U_1\times U_2, \mathcal{C}),\mathcal{I}_{b})$ using the definition of continuity on inductive limits, but for this it suffices to check it is bounded $(\mathcal{D}'_{\Pi_2,\Pi_2}(U_2,\mathcal{C}_2),\mathcal{I}_{H,iii})\to (\mathcal{D}'_{\Pi_1\dot{\times} \Pi_2}(U_1\times U_2, \mathcal{C}),\mathcal{I}_{H,iii})$ (since then it will be bounded thus continuous between bornologifications). 
The boundedness statement is included in the above hypocontinuity already checked.


We continue with the case $\Lambda^{(i)}=\overline{\gamma_i}$, note that $\Lambda=\overline{\gamma}$ in this case. 
 Since the source spaces are the completions of those of our first subcase in case (ii) by our 
 Theorem \ref{FAGeneral2}, we use \cite[\S 40.3.(4)]{Kothe2} for which we only need to check the target space is complete (which is given at the same place
 )
  and the closure of the bounded sets generates the bounded sets of the completion, and this is checked in proposition \ref{quasiCompletion}. 
For the general case, since the source spaces have the previous case as their completion, we already know the map is built with value in $\mathcal{D}'_{\gamma,\overline{\gamma}}(U,\mathcal{F};E_1\otimes E_2).$

Now the topology induced on our desired target by the completion is of course $\mathcal{I}_{ibi}$ by the same completion property at the target from Theorem \ref{FAGeneral2} 
and it suffices to check the extended map is valued in the stated space. For, thanks to the computation of the completion, we only have to compute wave front sets (since the DWF condition is not changed at the completion level). But for any element in pair in the source, wave front sets are closed cones, and applying (i) to them and increasingness of $\dot{\times}$ in both arguments, one concludes to the expected inclusion $WF(u\otimes v)\subset \Lambda.$

In case (iii), we use the inductive limit description for $\mathcal{I}_b$ and \eqref{Alternativeiii}. We checked in the lemma \ref{BoundedRelated} 
 that absolutely convex compact sets are equicontinuous and that such sets are only in the bornology of the quoted bornological inductive limits. We can thus use a result in the spirit of \cite[Th 1 c]{Melnikov} (replacing the strong assumption for bounded sets he uses namely a regularity of the inductive limit by the right assumption for his result which is to take an hypocontinuity with respect to the bornologies of the inductive limits). Then the required hypocontinuities for applying this result all follow from the closed case (i).
 [Note that the case with $\alpha_1=b$  also follows directly  from the case $\alpha_1=iii$ by composition with the bornologification map.]

Noting the first space with topology $\mathcal{I}_b$ is ultrabornological thus barrelled, we can use \cite[\S 40.2.(3,4)]{Kothe2} to get asymmetric hypocontinuity with respect to bounded sets on the right hand side (starting only from separate continuity). And since in the case $\gamma_1$ closed, bounded and equicontinuous sets coincide in the left hand side, one gets full hypocontinuity in this case. This gives the special result in the closed case.

Finally in case (iv) note first taking an equicontinuous or an absolutely convex compact set in the left hand side, thus a bounded set with wave front set controlled by a closed cone $\Gamma\subset\gamma_1$ so that the statement reduces to prove hypocontinuity in the case $\gamma_1$ closed in reasoning as before with a variant of \cite[Th 1 c]{Melnikov}.
We thus stick to this case. Starting from (iii) with same $\gamma_i$ gives the case $\gamma_2=\Lambda^{(2)}$ (since $\mathcal{I}_{ppp}=\mathcal{I}_{iii}$ in this case in the source and $\mathcal{I}_{iii}$ is stronger in the targe).
Next we deal with the case $\overline{\gamma_2}=\Lambda^{(2)},$
 and fixing as target space the completion of the space we want we have from (iii) hypocontinuity. We can apply \cite[\S 40.3.(4)]{Kothe2}, (we thus don't need to  extend on the left hand side) and get hypocontinuity using again proposition \ref{quasiCompletion}. 
 To get the same thing with the stated target space, we only have to check the space of value is the right one in  controlling the wave front set (this is the only thing changing while we got to the completion before), but this is again obvious using (i). This concludes this second case. Finally, for the general case, we deduce as before from the previous case the map by functoriality and only have to change the space of value which means computing the wave front set of arrival since all the spaces have the same completion (see Theorem \ref{FAGeneral2}.)

\end{proof}

The advertised non-continuity result will be a  consequence of an improvement of the previous result. We will investigate more its meaning in the third paper of this series.

We concentrate on the closed cone case, with $U_i\subset \R^{d_i}$ and without vector bundles until the end of this subsection, for which we need to introduce extra seminorms. We leave the manifold slight generalization to the reader. Let $\gamma_i
\subset \dot{T}^*U_i$  
be 
cones and $\gamma=\gamma_1\dot{\times}\gamma_2$.

Define  for $D>0,k\in\N^*$
$$C_{D,k,1}=\{(\xi,\eta)\in\R^{d_1}\times\R^{d_2},D|\eta|^k\geq  |\xi|\}$$

$$C_{D,k,2}=\{(\xi,\eta)\in\R^{d_1}\times\R^{d_2},D|\xi|^k\geq  |\eta|\}$$
 and for a set $A$, $f\in \mathcal{D}(U_1 \times U_2)$
 
 $$p_{l,A,f}(u)=\text{sup}_{(\xi,\eta)\in A}(1+|\xi|+|\eta|)^l|\mathcal{F}(uf)(\xi,\eta)|$$
 
 Then we call $\mathcal{I}_{\beta}$, the topology on $(\mathcal{D}'_{\gamma,\Lambda}(U_1\times U_2, \mathcal{F}),\mathcal{I}_{\beta})$  for $\gamma=\gamma_1\dot{\times}\gamma_2$ generated by the seminorms for $\mathcal{I}_{ppp}$ and, for $\varphi_i\in \mathcal{D}(U_i)$ and  cones $V_i$ such that $(\text{supp}(\varphi_i)\times V_i)\cap \gamma_i=\emptyset$, any $D,k,l$, the seminorms
$p_{l,C_{D,k,1}\cap (\R^{d_1}\times V_{2}),\varphi_1\otimes\varphi_2}$ and 
$p_{l,C_{D,k,2}\cap (V_1\times\R^{d_2}),\varphi_1\otimes\varphi_2}.$

\begin{proposition}\label{hypocontinuityImproved}
Let $\gamma_i
\subset \dot{T}^*U_i$  
be 
cones and $\gamma=\gamma_1\dot{\times}\gamma_2$. 
$(\mathcal{D}'_{\gamma,\overline{\gamma}}(U_1\times U_2, \mathcal{F}),\mathcal{I}_{\beta})$ is complete and the tensor product bilinear map is $\gamma$-hypocontinuous :$$.\otimes. : (\mathcal{D}'_{\gamma_1,\gamma_1}(U_1,\mathcal{F}),\mathcal{I}_{iii})\times (\mathcal{D}'_{\gamma_2,\gamma_2}(U_2,\mathcal{F}),\mathcal{I}_{iii})\to (\mathcal{D}'_{\gamma,\overline{\gamma}}(U_1\times U_2, \mathcal{F}),\mathcal{I}_{\beta}).$$ 
\end{proposition}

\begin{proof}
For the completeness statement, taking any Cauchy net  $u_n$ for $\mathcal{I}_\beta$ it converges to some $u$ for $\mathcal{I}_{ppp}$ from the completeness of this topology and $\mathcal{F}(u_nf)(\xi,\eta)$ thus converges pointwise, and since moreover, from completeness of our supplementary decay seminorms on spaces of continuous functions, they have limits in those spaces, the limit is necessarily   $\mathcal{F}(uf)(\xi,\eta)$ so that we have indeed the full convergence for $\mathcal{I}_\beta.$ We could of course formulate this in terms of closed subspace of a product of complete spaces.

It remains to check the improved bonds near the cones $\gamma_1\times U_2\times\{0\}$ and $U_1\times\{0\}\times\gamma_2$ for hypocontinuity. Note that as in the proof of proposition \ref{hypocontinuity} (iii) it suffices to prove the closed $\gamma_i$ case and reason by inductive limits. 

By symmetry we are content to bound (in the setting before the proposition):
\begin{align*}p_{l,C_{D,k,1}\cap (\R^{d_1}\times V_{2}),\varphi_1\otimes\varphi_2}(u_1\otimes u_2):=\text{sup}_{(\xi,\eta)\in C_{D,k,1}, \eta\in V_2}(1+|\xi|+|\eta|)^l|\mathcal{F}(u_1\varphi_1)(\xi)||\mathcal{F}(u_2\varphi_2)(\eta)|
\end{align*}
in two cases to prove hypocontinuity of the tensor product with those seminorms added. First if we assume $u_1\in B$ some bounded set in $\mathcal{D}'(U_1)$ thus say with $\mathcal{F}(u_1\varphi_1)$ polynomially bounded of order $m$, in this case, one gets (using $(1+|\xi|+|\eta|)\leq (1+|\xi|)(1+|\eta|)\leq (1+D|\eta|^k)(1+|\eta|)\leq \max(1,|D|)(1+|\eta|)^{k+1}$ on $C_{D,k,1}$):
\begin{align*}\sup_{u_1\in B}&p_{l,C_{D,k,1}\cap (\R^{d_1}\times V_{2}),\varphi_1\otimes\varphi_2}(u_1\otimes u_2)\\&\leq \max(1,D)^{l+m}\left( \sup_{\xi}(1+|\xi|)^{-m}|\mathcal{F}(u_1\varphi_1)(\xi)|\right)
P_{(k+1)(m+l),\varphi_2,V_2}(u_2).\end{align*}

Second, if we assume $u_2\in B$ some bounded set in $\mathcal{D}_{\gamma_2}'(U_2)$ one introduces $B'=\{\varphi_1 e_{\xi}(1+|\xi|+|\eta|)^l|\mathcal{F}(u_2\varphi_2)(\eta)| : (\xi,\eta)\in C_{D,k,1}, \eta\in V_2 , u_2\in B\}$ so that obviously $$\sup_{u_2\in B}p_{l,C_{D,k,1}\cap (\R^{d_1}\times V_{2}),\varphi_1\otimes\varphi_2}(u_1\otimes u_2)\leq P_{B'}(u_1)$$
and it remains to prove $B'$ bounded in $\mathcal{D}(U_1).$ First $B'$ is indeed uniformly supported on $\text{supp}(\varphi_1).$

We have to bound \begin{align*}\sup_{f\in B'}\sup_{x\in K}\sum_{|\alpha|\leq n}|f^{\alpha}(x)|&\leq \sup_{(\xi,\eta)\in C_{D,k}, \eta\in V_2 , u_2\in B}C_{n,\varphi_1}(1+|\xi|)^n(1+|\xi|+|\eta|)^l|\mathcal{F}(u_2\varphi_2)(\eta)|\\&\leq \max(1,D)^{l+n}C_{n,\varphi_1}\sup_{ u_2\in B}P_{(k+1)(l+n),V_2,\varphi_2}(u_2)\end{align*} for some compact $K$ say containing a neighborhood of $\text{supp}(\varphi_1),$ and some constant $C_{n,\varphi_1}$ depending on the norms of $\varphi_1$. But now the last term is finite by definition of boundedness of $B$. Those two estimates together prove the stated hypocontinuity.

\end{proof}

Our next non-continuity result will be based on the argument that if we could have a full continuity, we would have an isomorphism contradicting the improved target topology for hypocontinuity in our previous proposition.

\begin{proposition}\label{NONContinuity}
Let $\Gamma_i\subset \dot{T}^*U_i$  closed cones and $\Lambda_i=-\Gamma_i^c\subset \dot{T}^*U_i$ open cones (complements being taken in $\dot{T}^*U_i,\dot{T}^*U$), $\mathcal{C}_i$ polar enlargeable families of closed sets on $U_i$. Let also $U=U_1\times U_2$, $\mathcal{C}=(\mathscr{O}_{(\mathscr{O}_{\mathcal{C}_1}^o\times \mathscr{O}_{\mathcal{C}_2}^o)^{oo}})^o$,  $\Gamma=(\Gamma_1^c\dot{\times} \Gamma_2^c)^c\subset \dot{T}^*U$,  and $\lambda=(\Lambda_1^c\dot{\times} \Lambda_2^c)^c$ (which is open)
 $\Lambda=-\Gamma^c
=\Lambda_1\dot{\times}\Lambda_2\subset  \lambda,$  and $\gamma=-\lambda^c=\Gamma_1\dot{\times} \Gamma_2\subset \Gamma.$

Then the tensor map $$[(\mathcal{D}'_{\Gamma_1,\Gamma_1}(U_1,\mathcal{F}),\mathcal{I}), (\mathcal{D}'_{\Gamma_2,\Gamma_2}(U_2,\mathcal{F}) ,\mathcal{I})]\to (\mathcal{D}'_{\Gamma,\Gamma}(U,\mathcal{F}),\mathcal{I}),$$ 
$$[(\mathcal{D}'_{\Lambda_1,\Lambda_1}(U_1,\mathcal{C}_1),\mathcal{I}), (\mathcal{D}'_{\Lambda_2,\Lambda_2}(U_2,\mathcal{C}_2) ,\mathcal{I})]\to (\mathcal{D}'_{\lambda,\lambda}(U,\mathcal{C}),\mathcal{I}),$$ 
are NOT continuous as soon as $(\Gamma_1,\Gamma_2)\not\in \{(\emptyset,\emptyset),(\dot{T}^*U_1,\dot{T}^*U_2)\}$, for any topology $\mathcal{I}$ among $\mathcal{I}_{ppp},\mathcal{I}_{ibi}.$ 
\end{proposition}
Of course, from the stated inclusions of cones,  our previous proposition proved they are (at least $\gamma$)-hypocontinuous. Note that the target cone $\Gamma$ is only slightly smaller than the one of proposition \ref{continuity}, the only difference being in zero sections, explaining this result is quite optimal for full continuity.
\begin{proof}
Assume for contradiction that the maps are continuous, so that one gets continuous maps (the second one by completion) \begin{align}\label{contMap}(\mathcal{D}'_{\Gamma_1,\Gamma_1}(U_1,\mathcal{F}),\mathcal{I})\hat{\otimes}_\pi (\mathcal{D}'_{\Gamma_2,\Gamma_2}(U_2,\mathcal{F}) ,\mathcal{I})\to (\mathcal{D}'_{\Gamma,\Gamma}(U,\mathcal{F}),\mathcal{I}),\end{align}

(resp. \begin{align}\label{contMap2}(\mathcal{D}'_{\Lambda_1,\overline{\Lambda_1}}(U_1,\mathcal{C}_1),\mathcal{I})\hat{\otimes}_\pi (\mathcal{D}'_{\Lambda_2,\overline{\Lambda_2}}(U_2,\mathcal{C}_2) ,\mathcal{I})\to (\mathcal{D}'_{\lambda,\overline{\lambda}}(U,\mathcal{C}),\mathcal{I}),\ \ \ \ \ \ \  \ )\end{align}
and moreover from the hypocontinuity in proposition \ref{hypocontinuity} (ii),(iii) we also have a continuous (
note that we know $\epsilon$-hypocontinuity is the same as $\gamma$-hypocontinuity in the stated case) \begin{align}\label{hypoMap}(\mathcal{D}'_{\Lambda_1,{\Lambda_1}}(U_1,\mathcal{K}),\mathcal{I}){\otimes}_{\beta e} (\mathcal{D}'_{\Lambda_2,{\Lambda_2}}(U_2,\mathcal{K}) ,\mathcal{I})\to (\mathcal{D}'_{\Lambda,\overline{\Lambda}}(U,\mathcal{K}),\mathcal{I})\end{align}(resp. \begin{align}\label{hypoMap2}(\mathcal{D}'_{\Gamma_1,\Gamma_1}(U_1,(\mathscr{O}_{\mathcal{C}_1})^o),\mathcal{I}_{ibi}){\otimes}_\beta (\mathcal{D}'_{\Gamma_2,\Gamma_2}(U_2,(\mathscr{O}_{\mathcal{C}_2})^o) ,\mathcal{I}_{ibi})\to (\mathcal{D}'_{\gamma,\gamma}(U,(\mathscr{O}_{\mathcal{C}})^o),\mathcal{I}_{ibi}).\ \ \ \ \ \ \  \ )\end{align}
We used the completion here in order to be able to go to the completed tensor product (no problem for $\gamma$ since it is closed).

We first show this would imply the topological or algebraic isomorphisms :
$$(\mathcal{D}'_{\Gamma_1,\Gamma_1}(U_1,\mathcal{F}),\mathcal{I})\hat{\otimes}_\pi (\mathcal{D}'_{\Gamma_2,\Gamma_2}(U_2,\mathcal{F}) ,\mathcal{I})\simeq_{alg} (\mathcal{D}'_{\Gamma,\Gamma}(U,\mathcal{F}),\mathcal{I}),$$
$$(\mathcal{D}'_{\Lambda_1,{\Lambda_1}}(U_1,\mathcal{K}),\mathcal{I})\hat{\otimes}_{\beta e} (\mathcal{D}'_{\Lambda_2,{\Lambda_2}}(U_2,\mathcal{K}) ,\mathcal{I})\simeq_{top} (\mathcal{D}'_{\Lambda,\overline{\Lambda}}(U,\mathcal{K}),\mathcal{I}),$$
(resp. $$(\mathcal{D}'_{\Lambda_1,\overline{\Lambda_1}}(U_1,\mathcal{C}_1),\mathcal{I})\hat{\otimes}_\pi (\mathcal{D}'_{\Lambda_2,\overline{\Lambda_2}}(U_2,\mathcal{C}_2) ,\mathcal{I})\simeq_{alg} (\mathcal{D}'_{\lambda,\overline{\lambda}}(U,\mathcal{C}),\mathcal{I}),$$ $$(\mathcal{D}'_{\Gamma_1,\Gamma_1}(U_1,(\mathscr{O}_{\mathcal{C}_1})^o),\mathcal{I}_{ibi})\hat{\otimes}_\beta (\mathcal{D}'_{\Gamma_2,\Gamma_2}(U_2,(\mathscr{O}_{\mathcal{C}_2})^o) ,\mathcal{I}_{ibi})\simeq_{top} (\mathcal{D}'_{\gamma,\gamma}(U,(\mathscr{O}_{\mathcal{C}})^o),\mathcal{I}_{ibi}).\ \ \ \ \ \ \  \ )$$
and we will get a contradiction from these isomorphisms.

We start by proving the stated isomorphisms. From \eqref{contMap}, \eqref{contMap2} we have one map for the algebraic isomorphisms and we have to build the inverse. 
Since the spaces for which we take projective tensor products are complete nuclear and thus have the approximation property of Grothendieck, it suffices to build a continuous map to the Schwartz $\epsilon$ product (which coincides with the completed projective tensor product, and by density and agreement on smooth maps, it will canonically be the inverse) 
$$(\mathcal{D}'_{\Gamma,\Gamma}(U,\mathcal{F}),\mathcal{I})\to(\mathcal{D}'_{\Gamma_1,\Gamma_1}(U_1,\mathcal{F}),\mathcal{I})\epsilon (\mathcal{D}'_{\Gamma_2,\Gamma_2}(U_2,\mathcal{F}) ,\mathcal{I})\simeq \mathcal{D}'_{\Gamma_1,\Gamma_1}(U_1,\mathcal{F})\hat{\otimes}_\epsilon \mathcal{D}'_{\Gamma_2,\Gamma_2}(U_2,\mathcal{F}) .$$
(resp. $$(\mathcal{D}'_{\lambda,\overline{\lambda}}(U,\mathcal{C}),\mathcal{I})\to (\mathcal{D}'_{\Lambda_1,\overline{\Lambda_1}}(U_1,\mathcal{C}_1),\mathcal{I})\epsilon (\mathcal{D}'_{\Lambda_2,\overline{\Lambda_2}}(U_2,\mathcal{C}_2),\mathcal{I}) \simeq \mathcal{D}'_{\Lambda_1,\overline{\Lambda_1}}(U_1,\mathcal{C}_1)\hat{\otimes}_\epsilon \mathcal{D}'_{\Lambda_2,\overline{\Lambda_2}}(U_2,\mathcal{C}_2) \ \ \ \ \ ),$$
But as we reminded, $(\mathcal{D}'_{\Gamma_1,\Gamma_1}(U_1,\mathcal{F}),\mathcal{I})\epsilon (\mathcal{D}'_{\Gamma_2,\Gamma_2}(U_2,\mathcal{F}),\mathcal{I})$ is the set of $\epsilon$-hypocontinuous bilinear maps on the product of Arens duals
$$[(\mathcal{D}'_{\Gamma_1,\Gamma_1}(U_1,\mathcal{F}),\mathcal{I})_{c}'\otimes_{\beta e}(\mathcal{D}'_{\Gamma_2,\Gamma_2}(U_2,\mathcal{F}),\mathcal{I})_{c}']'\simeq [(\mathcal{D}'_{\Lambda_1,\overline{\Lambda_1}}(U_1,\mathcal{K}),\mathcal{I})\hat{\otimes}_{\beta e} (\mathcal{D}'_{\Lambda_2,\overline{\Lambda_2}}(U_2,\mathcal{K}) ,\mathcal{I})]',$$
since by  propositon \ref{FAGeneral2} the Arens topologies coincide with $\mathcal{I}$ on the duals (and we went to the completions for completed tensor product for a more uniform notation). Respectively, $(\mathcal{D}'_{\Lambda_1,\overline{\Lambda_1}}(U_1,\mathcal{C}_1),\mathcal{I})\epsilon (\mathcal{D}'_{\Lambda_2,\overline{\Lambda_2}}(U_2,\mathcal{C}_2),\mathcal{I})$ is the set of $\epsilon$-hypocontinuous bilinear maps (which coincides here with hypocontinuous by  the computation of equicontinuous sets in \cite[lemma 27]{Dab14a} 
) on the product of Arens duals
$$[(\mathcal{D}'_{\Lambda_1,\overline{\Lambda_1}}(U_1,\mathcal{C}_1))_{c}'\otimes_\beta(\mathcal{D}'_{\Lambda_2,\overline{\Lambda_2}}(U_2,\mathcal{C}_2))_{c}']'\simeq [(\mathcal{D}'_{\Gamma_1,\Gamma_1}(U_1,(\mathscr{O}_{\mathcal{C}_1})^o),\mathcal{I}_{ibi})\hat{\otimes}_{\beta } (\mathcal{D}'_{\Gamma_2,\Gamma_2}(U_2,(\mathscr{O}_{\mathcal{C}_2})^o) ,\mathcal{I}_{ibi})]',$$
since by \cite[Corol 26]{Dab14a} 
 the Arens topology coincides with $\mathcal{I}_{ibi}$.

The maps we want are thus obtained by composition with the map \eqref{hypoMap} (resp. \eqref{hypoMap2}), 
We could check the continuity of the first but we won't need it. 

Again for the second isomorphism, both sides are complete and thus we have only to identify the topologies on the algebraic tensor product which is dense on both sides. 
 By the maps \eqref{hypoMap}, \eqref{hypoMap2} again, the topology induced by the bornological tensor product is stronger. But we just identified their duals (not topologically though) 
 thus both are weaker than the Mackey topology of duality by Mackey-Arens theorem. Finally, in \cite[Corol 26]{Dab14a} 
  we checked the Mackey topology on $\mathcal{D}'_{\Lambda,\overline{\Lambda}}(U,\mathcal{K})$  (resp. $\mathcal{D}'_{\gamma,\overline{\gamma}}(U,(\mathscr{O}_\mathcal{C})^o)$) is the same as $\mathcal{I}_{iii}=\mathcal{I}_{ibi}=\mathcal{I}$ (resp. $\mathcal{I}_{ibi}$) thus all topologies coincide as expected. (Note we may also use \cite[\S 21. 4.(5)]{Kothe} to identify the Mackey topology of a space and the one of the completion on the tensor product dense subspace of $(\mathcal{D}'_{\Lambda,\overline{\Lambda}}(U,\mathcal{K}),\mathcal{I}_{iii})$ etc). 
 We thus deduced the second isomorphism. 

Let us now get our contradiction in showing the completion of the  maps \eqref{hypoMap}, \eqref{hypoMap2} are actually not surjective.
But from our previous proposition \ref{hypocontinuityImproved}, after completing the map obtained on induced tensor products, we checked any $u$ in the image satisfies the supplementary bound (taking $D=3^k$ as we can)
$$\text{sup}_{(\xi,\eta)\in C_{3^k,k,1}, \eta\in V_2}(1+|\xi|+|\eta|)^l|\mathcal{F}(u(\varphi_1\otimes\varphi_2))(\xi,\eta)|<\infty,$$ as soon as $\text{supp}(\varphi_2)\times V_2\cap {\Lambda_2}=\emptyset.$ 
 It thus suffices to prove this condition is not satisfied by any $u$ in  $(\mathcal{D}'_{\Lambda,\overline{\Lambda}}(U,\mathcal{K}),\mathcal{I}_{iii}),$ (respectively in the second case exchanging $\Lambda_i$'s and $\gamma_i$'s, we only write the first case and leave the notational changes to the reader in the second). Note that for $\mathcal{I}_{ibi}$ this is not written in the previous proposition, but in the closed case, we have full hypocontinuity which goes to the bornologification and in the open case, we reason as in proposition \ref{hypocontinuity} (ii) where we only have to check separate continuity from barelledness, and then our proposition \ref{hypocontinuityImproved} is enough.

We will use a specific case of H\"ormander's example 8.2.4 \cite[p.~188]{Hormander-97} with $s=-2l, \rho=1/k.$ Fix $\xi_0$ a vector, $|\xi_0|=1$ such that say $(x_0,\xi_0)\in \Lambda_1\times U_2\times \{0\}.$ This points exists of course if $\Lambda_1\neq \emptyset,$ (we can assume that up to exchanging the indexes $1,2$). Of course,  we take $x_0=(x,y)\in U_1\times U_2$ and since we are free about $y$, we choose $y$ such that there is a non-empty closed cone $V_2$ with $\{y\}\times V_2\cap \Lambda_2=\emptyset$ and it exists since $\Lambda_2^c\neq \emptyset$ (without loss of generality, since either $\Lambda_1^c= \emptyset,$ and this is forced by our assumption, or otherwise $\Lambda_1^c\neq \emptyset,$  and if $\Lambda_2^c= \emptyset$ then $\Lambda_2\neq \emptyset$ and we could exchange again indexes $1,2$ to satisfy our two assumptions). We of course take $\varphi_1,\varphi_2$ smooth compactly supported on  a neighborhood of $x,y$ respectively.

Let $\chi\in C^\infty(\R,[0,1])$ be equal to 1
in $(-\infty,1/2)$ and to 0 in $(1,+\infty)$, 
with $0\le \chi\le 1$. 
Take an orthonormal basis 
$(e_1=\xi_0,e_2,...,e_d)$ such that $(e_{1},...,e_{d_1})$ is an orthonormal basis of $ \R^{d_1}\times\{0\},$ $(e_{d_1+1},...,e_{d})$ an orthonormal basis of $\{0\}\times \R^{d_2},$  and write coordinates in this coordinate system, so that for instance $C_{3^k,k,1}=\{( |\xi_{d_1+1}|^2+...+|\xi_{d}|^2)\geq\frac{1}{3}(|\xi_1|^2+...+|\xi_{d_1}|^2)^{1/k}\}$.
Note that $A=\{( \xi_2=...=\xi_{d_1}=0, \frac{1}{2}|\xi_1|^{2\rho}\geq (|\xi_{d_1+1}|^2+...+|\xi_{d}|^2)\geq\frac{1}{3}|\xi_1|^{2/k}, (\xi_{d_1+1},...,\xi_{d})\in V_2\}\subset C_{3^k,k}\cap \{ (|\xi_{d_1+1}|^2+...+|\xi_{d}|^2)/|\xi_1|^{2\rho}\leq \frac{1}{2}\} $ has  non-empty intersection with every complement of compact sets (for $\rho=1/k$).

Define $u_{\xi_0,s}\in \mathcal{S}'(\R^n)$, for $s\in \R$, by
\begin{eqnarray*}
\widehat{u_{\xi_0,s}}(\xi) &=& (1-\chi(\xi_1)) \xi_1^{-s}
  \chi((\xi_2^2+\dots+\xi_n^2)/\xi_1^{2\rho}).
\end{eqnarray*}
 Then
$WF(u_{\xi_0,s})=\{(0;\xi); \xi_2=\dots=\xi_n=0, \xi_1>0\}=\{0\}\times \R_+^*\xi_0$
and $u_{\xi_0,s}$ coincides with a function in $\mathcal{S}(\R^n)$
outside a neighborhood of the origin~\cite[p.~188]{Hormander-97}. Of course we can consider the obvious translation $u=\chi(T_{x_0}u_{\xi_0,s})$, $\chi$ smooth with compact support $\chi(x_0)=1$ with $DWF(u)=WF(u)=\{x_0\}\times \R_+^*\xi_0$ so that $u\in \mathcal{D}'_{\Lambda,\overline{\Lambda}}(U,\mathcal{K}).$

Moreover, since $w=(1-\chi(\varphi_1\otimes \varphi_2))(T_{x_0}u_{\xi_0,s})\in \mathcal{S}(\R^d)$ from the results above, $$(1+|\xi|+|\eta|)^l|\mathcal{F}(u(\varphi_1\otimes\varphi_2))(\xi,\eta)|\geq (1+|\xi|+|\eta|)^l|\mathcal{F}(T_{x_0}u_{\xi_0,s})(\xi,\eta)| -(1+|\xi|+|\eta|)^l|\mathcal{F}(w)(\xi,\eta)|$$ and $(\xi,\eta)\mapsto (1+|\xi|+|\eta|)^l|\mathcal{F}(w)(\xi,\eta)|$ which is rapidly decreasing is not a problem at infinity.
Since we can bound below the sup we are interested in  a supremum on $A$, on which $\chi((\xi_2^2+\dots+\xi_n^2)/\xi_1^{2\rho})=1$, one deduces :$$\text{sup}_{(\xi,\eta)\in C_{3^k,k,1}, \eta\in V_2}(1+|\xi|+|\eta|)^l|\mathcal{F}(T_{x_0}u_{\xi_0,s})(\xi,\eta)|\geq \text{sup}_{(\xi,\eta)\in A}(1+|\xi|+|\eta|)^l|(1-\chi(\xi_1)) \xi_1^{-s}|=\infty,$$ at least for $s$ negative.
Thus $u$ is not in the image of the completion of  \eqref{hypoMap} concluding to our statement and to the expected contradiction.
\end{proof}

\section{Control of tensor products by wave front set conditions}
Our final goal in the next paper of this series will be to describe all the tensor products we considered here in terms of microlocal conditions, notably in order to obtain extra functional analytic properties  not given by the general theory. But since the lack of continuity presented in proposition \ref{NONContinuity} gave an obstruction to an easy identification of these tensor products (see the isomorphism from which we obtained a contradiction in the proof), we will still need in our applications a way of checking that a distribution gives an element of a completed tensor product by checking a (dual) wave front set condition. This is the goal of this section to give such conditions. 

For our convenience in order to state results in the context of most topologies of Theorem \ref{FAGeneral2}, we will write for a cone $\gamma$ \begin{equation}\label{notation}\mathcal{J}_1=\mathcal{J}_2=\mathcal{I}_{ppp}, \mathcal{J}_3=\mathcal{J}_4=\mathcal{I}_{ibi}, \ \ \ \gamma(1)=\gamma(3)=\gamma,\gamma(2)=\gamma(4)=\overline{\gamma},\end{equation}

and $\overline{j}=j$ for $j=1,4$, $\overline{j}=5-j$ for $j=2,3$ so that our duality results (for instance for Arens duals) are conveniently summarized as $$(\mathcal{D}'_{\gamma,\gamma(j)}(U,\mathcal{C};E),\mathcal{J}_{j})'_c=(\mathcal{D}'_{-\gamma^c,-\gamma^c(\overline{j})}(U,(\mathscr{O}_{\mathcal{C}})^o;E^*),\mathcal{J}_{\overline{j}}).$$

\begin{proposition}\label{InjectionProjectiveProduct}
Let $\gamma_i
\subset \dot{T}^*U_i$ be 
cones and $\mathcal{C}_i$ enlargeable polar families of closed sets of $U_i$, $E_i\to U_i$ vector bundles. 
Let $\gamma=(\gamma_1^c\dot{\times}\gamma_2^c)^c,$ $\Gamma=((\overline{\gamma}_1)^c{\times}(\overline{\gamma}_2)^c)^c$  
and $\mathcal{C}=(\mathcal{C}_1^o\times \mathcal{C}_2^o)^{o}.$
There are continuous injections for any $j\in\{1,2,4\}$: \begin{align*}(&\mathcal{D}'_{\gamma,\gamma(j)}(U_1\times U_2, \mathcal{C};E_1\otimes E_2),\mathcal{J}_{j})\hookrightarrow(\mathcal{D}'_{\gamma_1,\gamma_1(j)}(U_1,\mathcal{C}_1;E_1),\mathcal{J}_{j})\epsilon (\mathcal{D}'_{\gamma_2,\gamma_2(j)}(U_2,\mathcal{C}_2;E_2),\mathcal{J}_{j})\\&\ \ \ \hookrightarrow(\mathcal{D}'_{\gamma_1,\gamma_1(j)}(U_1,\mathcal{C}_1;E_1),\mathcal{J}_{j})\hat{\otimes}_{\pi} (\mathcal{D}'_{\gamma_2,\gamma_2(j)}(U_2,\mathcal{C}_2;E_2),\mathcal{J}_{j})\hookrightarrow (\mathcal{D}'_{\Gamma,\Gamma}(U_1\times U_2, \mathcal{F};E_1\otimes E_2),\mathcal{I}_{ppp}),\end{align*}
and for any $j\in\{1,3\}$ :
\begin{align*}(\mathcal{D}'_{\gamma,{\gamma}}(U_1\times U_2, \mathcal{C}),\mathcal{J}_{j})\hookrightarrow(\mathcal{D}'_{\gamma_1,\overline{\gamma_1}}(U_1,\mathcal{C}_1),\mathcal{I}_{ppp})\epsilon (\mathcal{D}'_{\gamma_2,{\gamma_2}}(U_2,\mathcal{C}_2),\mathcal{J}_{j}).\end{align*} 
\end{proposition}
\begin{proof}
We know from \cite[Prop 33]{Dab14a} 
 that all spaces involved in the first line are nuclear for the stated topologies
 , thus completed projective product and injective product coincide, and moreover, they have the approximation property (see e.g.\cite[p 110]{Schaeffer}), thus by \cite[prop 11 p 46]{Schwarz}, the $\epsilon$ product is a dense subspace of this completed injective product, we only have to build a map to/from Schwartz $\epsilon$-product (since the last space $(\mathcal{D}'_{\Gamma,\Gamma}(U_1\times U_2, \mathcal{F};E_1\otimes E_2),\mathcal{I}_{ppp})$ is complete). 
 Using injectivity of the $\epsilon$ product of maps \cite[p 20]{Schwarz} (and the equality $\mathcal{I}_{pip}=\mathcal{I}_{iii}$ in the closed DWF/WF case and injections in \cite[lemma 21]{Dab14a}
) we have $$(\mathcal{D}'_{\gamma_1,\gamma_1(j)}(U_1,\mathcal{C}_1;E_1),\mathcal{J}_{j})\epsilon \mathcal{D}'_{\gamma_2,\gamma_2(j)}(U_2,\mathcal{C}_2;E_2),\mathcal{J}_{j})\hookrightarrow (\mathcal{D}'_{\overline{\gamma_1},\overline{\gamma_1}}(U_1,\mathcal{C}_1;E_1),\mathcal{I}_{ppp})\epsilon (\mathcal{D}'_{\overline{\gamma_2},\overline{\gamma_2}}(U_2,\mathcal{C}_2;E_2),\mathcal{I}_{ppp}).$$ Using also our proposition \ref{continuity} and composing with the previous one, the second map is thus known (and injective via inclusions in distributions).

But thanks to the identifications of the Arens duals in the same corollary \begin{align*}(\mathcal{D}'_{\gamma_1,\gamma_1(j)}(U_1,\mathcal{C}_1;E_1)&,\mathcal{J}_{j})\epsilon (\mathcal{D}'_{\gamma_2,\gamma_2(j)}(U_2,\mathcal{C}_2;E_2),\mathcal{J}_{j})\\&\simeq [(\mathcal{D}'_{-\gamma_1^c,-\gamma_1^c(\overline{j})}(U_1,(\mathscr{O}_{\mathcal{C}_1})^o;E_1^*),\mathcal{J}_{\overline{j}})\otimes_{\beta e} (\mathcal{D}'_{-\gamma_2^c,-\gamma_2^c(\overline{j})}(U_2,(\mathscr{O}_{\mathcal{C}_2})^o;E_2^*),\mathcal{J}_{\overline{j}})]'.\end{align*}

Our first injective map is thus built in dualizing the dense range tensor multiplication map of proposition \ref{hypocontinuity} (ii,iii), using also $(\mathscr{O}_\mathcal{C})^o=(\mathscr{O}_{(\mathcal{C}_1^{o}\times \mathcal{C}_2^{o})^{o}})^o=(\mathscr{O}_{\mathcal{C}_1}^o\times\mathscr{O}_{\mathcal{C}_2}^o)^{oo}$,  (as is easily checked, using crucially enlargeability for $\supset$). 
To prove continuity, it suffices to prove the tensor multiplication map sends tensor products of equicontinuous sets to equicontinuous sets. For $j=4$ this is obvious since equicontinuous sets are the same as bounded sets. For $j=1,2$ 
we identified in lemma \ref{BoundedRelated} equicontinuous sets in $\mathcal{D}'_{-\gamma_i^c,-\gamma_i^c}(U_i,(\mathscr{O}_{\mathcal{C}_i})^o;E_i^*)$
to bounded sets of $(\mathcal{D}'_{\Pi_i,\Pi_i}(U_i,(\mathscr{O}_{\mathcal{C}_i})^o ;E_i^*),\mathcal{I}_{H,iii})$ with a closed 
$\Pi_i\subset -\gamma_i^c$. But, since we have hypocontinuity at level of each term of the defining inductive limits $\mathcal{I}_{iii},\mathcal{I}_{ibi}$, and since an equicontinuous set $B$ defines an equicontinuous family of maps $u\otimes \cdot$ which thus preserve those bounded sets of $(\mathcal{D}'_{\Pi_i,\Pi_i}$ and thus also preserve equicontinuous sets in $\mathcal{D}'_{-\gamma_i^c,-\gamma_i^c}(U_i,(\mathscr{O}_{\mathcal{C}_i})^o;E_i^*)$, as expected.

For the second inclusion, we can translate what we are looking for through duality (recall now $j=1,3$): 
$$(\mathcal{D}'_{\gamma,{\gamma}}(U_1\times U_2, \mathcal{C}),\mathcal{J}_{j})\hookrightarrow
[(\mathcal{D}'_{-\gamma_1^c,-\gamma_1^c}(U_1,(\mathscr{O}_{\mathcal{C}_1})^o),\mathcal{I}_{b})\otimes_{\beta e} (\mathcal{D}'_{-\gamma_2^c,-{\gamma_2^c}(j)}(U_2,(\mathscr{O}_{\mathcal{C}_2})^o),\mathcal{J}_{\overline{j}}=\mathcal{I}_{ppp})]'.$$

We used the computation of Arens duals summarized in 
Theorem \ref{FAGeneral2}.

Our second injective map is thus built in dualizing the dense range tensor multiplication map of proposition \ref{hypocontinuity} (iv). To check continuity, again, it suffices to see the tensor multiplication map sends products of equicontinuous sets to equicontinuous sets. But fixing $B$ equicontinuous in the first space (thus we can assume $\gamma_1$ closed by our identification of equicontinuous sets), $u\otimes.$ gives an equicontinuous family, and in the case $j=1$ via the identification of equicontinuous sets to those having wave front set in a closed subcone in the inductive limit, this reduces to the case $j=2$ (with even closed $\gamma_2$), in which case bounded sets and equicontinuous sets agree (at target partly from our reduction to $\gamma_1$ closed) . This concludes to our last result. 
\end{proof}

\section{Spaces of multilinear maps}

Our description of tensor products gives us access to abstract description of spaces of hypocontinuous multilinear maps on generalized H\"ormander spaces of distributions. For convenience we still use notation \eqref{notation}.

We summarize the descriptions of the spaces of interests which are $\epsilon$-products and give relations with spaces controlled by wave front set conditions. All spaces of hypocontinuous multilinear maps will be given their canonical topology as $\epsilon$ products, namely of uniform convergence on products of equicontinuous sets, see \cite{Schwarz}.

\begin{proposition}\label{multilinear}Let $\lambda_i=-\gamma_i^c, \lambda=-\gamma^c$ be cones, $E_i\to U_i,E\to U$ vector bundles  and $k_i\in\{1,2,3,4\},i=1,...,n;k\in\{2,4\}$
\begin{align*}L&({\bigotimes}_{\beta e, i\in [1,n]} (\mathcal{D}'_{\gamma_i,{\gamma}_i(k_i)}(U_i,\mathcal{C}_i;E_i),\mathcal{J}_{k_i}); (\mathcal{D}'_{\gamma,\overline{\gamma}}(U,\mathcal{C};E),\mathcal{J}_{k})) \\&\simeq  [\mathbb{{\varepsilon}}_{ i\in [1,n]} (\mathcal{D}'_{\lambda_i,{\lambda_i}(\overline{k_i})}(U_i,(\mathscr{O}_{\mathcal{C}_i})^o;E_i^*),\mathcal{J}_{\overline{k_i}})]\varepsilon (\mathcal{D}'_{\gamma,\overline{\gamma}}(U,\mathcal{C};E),\mathcal{J}_{k})
\end{align*}
is nuclear. It is also complete when all $k_i\in \{3,4\}$ in which case it is also the space of hypocontinuous multilinear maps, the corresponding completed projective tensor product and even the space of bounded multilinear maps when $k_i=k=4.$
We also have the following continuous inclusions for $\gamma_a=(-\gamma_1\dot{\times}...\dot{\times} -\gamma_n\dot{\times} \gamma^c)^c$ and $\mathscr{C}=((\mathscr{O}_{\mathcal{C}_1})^{oo}\times...\times (\mathscr{O}_{\mathcal{C}_n})^{oo}\times\mathcal{C}^o)^o$ on $V=U_1\times ...\times U_n\times U$:
$$(\mathcal{D}'_{\gamma_a,\gamma_a(K)}(V,\mathscr{C};E_1^*\otimes ...\otimes E_n^*\otimes E),\mathcal{J}_K)\hookrightarrow L({\bigotimes}_{\beta e, i\in [1,n]} (\mathcal{D}'_{\gamma_i,{\gamma}_i(k_i)}(U_i,\mathcal{C}_i;E_i),\mathcal{J}_{k_i}); (\mathcal{D}'_{\gamma,\overline{\gamma}}(U,\mathcal{C};E),\mathcal{J}_{k}),$$
for $k_1=...=k_{n-1}$ equal in the cases $k=\overline{k_1}=\overline{k_n}=K\in\{4,2\}$, or $K=3$ with either  $k_1=k_n=4,k=2$, or $k_n\in\{4,2,3\},k_1=3,k=2$ and finally, the case 
$K=1,k_1=3,k_n\in\{3,1\},k=2.$ We can also have the case $k_i\in\{3,4\},k\in\{2,4\},K=3$ partially covered before.
\end{proposition}
 \begin{proof}
 We computed in Theorem \ref{FAGeneral2} the various Arens=Mackey duals. The identification with the $\epsilon$-product comes from \cite[Prop 4 p 30]{Schwarz}, using only the completeness of the target space. Since from \cite[Prop 11 p 46]{Schwarz}, the $\epsilon$-product has for completion the completed injective tensor product, in the case where spaces have the approximation property, one deduces they are all nuclear as their completion.  They are complete when all the spaces of the $\epsilon$ product are complete \cite[Prop 3 p 29]{Schwarz}. The statement of hypocontinuity comes from the identification of bounded and equicontinuous sets as explained in lemma \ref{BoundedRelated} and in the bornological case, the identification of hypocontinuous linear maps and bounded linear maps.
 
 Applying \cite[Prop 1 p 20, Prop 7 p 38]{Schwarz} we get inductively from proposition \ref{InjectionProjectiveProduct} the various stated injections. For instance in the case $k=\overline{k_1}=\overline{k_n}=K\in\{4,2\}$ we only have to use repeatedly the first injection, the point is that all the spaces of the $\epsilon$-product involved are complete
 , so that we can apply associativity of $\epsilon$ product combined with the preservation of injectivity of $\epsilon$ products of continuous injections.

When $K\in\{3,1\}$ we first apply the second injection of  
 proposition \ref{InjectionProjectiveProduct}, once composed on the right with the completion maps to go from case $j\in\{3,1\}$ to $j+1\in\{4,2\}$ we can then apply again associativity freely and get the cases $K=3$ with either  $k_1=k_n=4,k=2$, or $k_n\in\{4,3\},k_1=3,k=2$ and finally, the case 
$K=1,k_1=3,k_n=3,k=2,$ and without equalities of $k_i$'s, $k_i\in\{3,4\},k\in\{2,4\},K=3$.
 
The cases $K=3,k_n=2,k_1=3,k=2$ and $K=1=k_1,k_n=3,k=2$ need a supplementary remark since one space is never complete.
  Here, we can gather terms which are complete in an $\epsilon$-product. Indeed, this is a consequence of \cite[ Prop 7 p 38]{Schwarz} with two applications of \cite[Prop 4 p 30]{Schwarz} which enables to gather complete spaces in the space of value, then apply associativity, and then come back to the $\epsilon$ product with associativity applied. In this way, since there is only one term not complete in the $\epsilon$ product above, one gathers all others and use the injections as above inductively.
\end{proof}
The easiest case is for the bornological complete topologies with all indexes $j=4$. At first reading, the reader should probably only consider this case where everything is straightforward. Since it is not clear at this point if this will be sufficient for applications (especially because the topology is quite tricky even to define in this case), we consider more general situations.

Our next result makes explicit nice spaces with controlled (dual) wave front set included in spaces of operators, and the various continuity of compositions on the full space and on restrictions. As we will see the only technical assumption is made for the target space where composition is made. We write $\epsilon_{eq}$ the family of $\epsilon$-equihypocontinuous \cite[p 18]{Schwarz} parts of $L({\bigotimes}_{\beta e, i\in [1,m]} (\mathcal{D}'_{\gamma'_i,\gamma'_i(k_i')}(U,\mathcal{C}'_i;E_i'),\mathcal{J}_{k_i'}); (\mathcal{D}'_{\gamma',\overline{\gamma'}}(U,\mathcal{C}';E'),\mathcal{J}_{k'})).$

\begin{proposition}\label{MultilinearComposition}In the setting of our previous proposition, and with supplementary cones $\gamma_i',\gamma', i\in[1,m],$ polar enlargeable families $\mathcal{C}'_i,\mathcal{C}',$ vector bundles $E_i'\to U_i',E'\to U'$ such that for a fixed index $j$, $\gamma_{j}'=\gamma$,
  $\mathcal{C}'_j=\mathcal{C}$, $E_j'=E,U_j'=U$. Assume we write $\gamma''_i=\gamma'_{i}$ for $i\in[1,j-1]$, $\gamma''_i=\gamma_{i-j}$ for $i\in[j,j+n-1]$,$\gamma''_i=\gamma'_{i-n+1}$ for $i\in[j+n,m+n-1]$, and similarly for $\mathcal{C}''_j$, $E''_j$.
Consider also $k_i\in \{1,2,3,4\}, i=1...n,  k_i'\in \{1,2,3,4\},i=1...m, k,k'\in \{2,4\}$ with $k_j'=k.$ Also write $k''_i=k'_{i}$ for $i\in[1,j-1]$, $k''_i=k_{i-j}$ for $i\in[j,j+n-1]$,$k''_i=k'_{i-n+1}$ for $i\in[j+n,m+n-1].$

   Then the map below corresponding to composition in the j-th variable  is hypocontinuous when $k=4$ and $\epsilon_{eq}-\beta$-hypocontinuous (thus separately continuous) when $k=2$ (always with the $\epsilon$-product topologies of proposition \ref{multilinear}).
\begin{align*}\circ_j:L
(&{\bigotimes}_{\beta e, i\in [1,m]}
 (\mathcal{D}'_{\gamma'_i,\gamma'_i(k_i')}(U_i',\mathcal{C}'_i;E_i'),\mathcal{J}_{k_i'})
; (\mathcal{D}'_{\gamma',\overline{\gamma'}}(U',\mathcal{C}';E'),\mathcal{J}_{k'}))\times \\& L({\bigotimes}_{\beta e, i\in [1,n]} (\mathcal{D}'_{\gamma_i,{\gamma}_i(k_i)}(U_i,\mathcal{C}_i;E_i),\mathcal{J}_{k_j}); (\mathcal{D}'_{\gamma,\overline{\gamma}}(U,\mathcal{C};E)),\mathcal{J}_k))) \\& \ \ \ \ \to L({\bigotimes}_{\beta e, i\in [1,n+m-1]} (\mathcal{D}'_{\gamma''_i,\gamma''_i(k''_i)}(U_i'',\mathcal{C}''_i;E''_i),\mathcal{J}_{k''_i}); (\mathcal{D}'_{\gamma',\overline{\gamma'}}(U',\mathcal{C}';E'),\mathcal{J}_{k'}))
\end{align*}
which restricts also to  $\epsilon$-hypocontinuous maps to the (often continuously) embedded subspaces given in proposition \ref{multilinear}.  
Explicitly if $\gamma_a=(-\gamma_1\dot{\times}...\dot{\times} -\gamma_n\dot{\times} \gamma^c)^c$,$\gamma_b=(-\gamma'_1\dot{\times}...\dot{\times} -\gamma'_m\dot{\times} (\gamma')^c)^c$,$\gamma_c=(-\gamma''_1\dot{\times}...\dot{\times} -\gamma''_{n+m-1}\dot{\times} (\gamma')^c)^c,\mathscr{C}=((\mathscr{O}_{\mathcal{C}_1})^{oo}\times...\times (\mathscr{O}_{\mathcal{C}_n})^{oo}\times\mathcal{C}^o)^o$, $V=U_1\times ...\times U_n\times U, \mathscr{E}=E_1^*\otimes...\otimes E_n^*\otimes E\to V$ and similarly for $\mathscr{C}',\mathscr{C}'',V',V''=U_1''\times ...\times U_{n+m-1}''\times U', \mathscr{E}',\mathscr{E}''$ the following maps coming from restriction are  $\epsilon$-hypocontinuous  :
$$\circ_j:(\mathcal{D}'_{\gamma_b,{\gamma_b}(\kappa_1)}(V',\mathscr{C}';\mathscr{E}'),\mathcal{J}_{\kappa_1})\otimes_{\beta e} (\mathcal{D}'_{\gamma_a,{\gamma_a(\kappa_2)}}(V,\mathscr{C};\mathscr{E}),\mathcal{J}_{\kappa_2})\to (\mathcal{D}'_{\gamma_c,{\gamma_c}(\kappa_3)}(V'',\mathscr{C}'';\mathscr{E}''),\mathcal{J}_{\kappa_3}),$$
with either $\kappa_1=\kappa_2=\kappa_3\in\{1,3,4\}$ or $\kappa_3=2, \kappa_1=\overline{\kappa_2}\in\{2,3\}.$ We even get full hypocontinuity when $\kappa_3\in\{3,4\}.$
\end{proposition}

\begin{proof}
There is no problem to define the composition in any case.
In order to check the space of value is indeed in $\epsilon$-equicontinuous maps by composition, it suffices to note that a product of equicontinuous sets is sent by a map in the second argument to an equicontinuous set. For, note that evaluated in all but one variable the maps above gives an equicontinuous family of continuous linear maps, which thus sends bounded sets (as equicontinuous sets) to bounded sets (see \cite[\S 39.3.(1)]{Kothe2}), but in $(\mathcal{D}'_{\gamma,\overline{\gamma}}(U,\mathcal{C};E),\mathcal{J}_{k}),k=2,4$ 
we saw bounded sets are equicontinuous as explained in lemma \ref{BoundedRelated}.
.

Then using the remark above we explain the claimed $\epsilon-$hypocontinuity of a composition $A\circ_jB.$ Note that this boils down to getting equicontinuity of the family of maps obtained by evaluating all but one variable of $A\circ_jB$ to equicontinuous sets. But for this there are 2 cases, the case where all the arguments of $B$ are evaluated, in which case the equicontinuity of the image just proved concludes since the $j$-th argument of $A$ then receives such an equicontinuous family. The second case is where all the arguments of $A$ but the $j$-th one are evaluated to equicontinuous sets, and then the equicontinuouity follows by (the obvious
) composition of equicontinuous families. 

To prove hypocontinuity of $\circ_j$, take $B_i$ equicontinuous in $(\mathcal{D}'_{\gamma''_i,\gamma''_i(k''_i)}(U,\mathcal{C}''_i;E''_i),\mathcal{J}_{k''_i})$, $A,A'$ bounded in $L
({\bigotimes}_{\beta e, i\in [1,m]}
 (\mathcal{D}'_{\gamma'_i,\gamma'_i(k_i')}(U,\mathcal{C}'_i;E_i'),\mathcal{J}_{k_i'})
; (\mathcal{D}'_{\gamma',\overline{\gamma'}}(U,\mathcal{C}';E'),\mathcal{J}_{k'}))$  and 

\noindent $L({\bigotimes}_{\beta e, i\in [1,n]} (\mathcal{D}'_{\gamma_i,{\gamma}_i(k_i)}(U,\mathcal{C}_i;E_i),\mathcal{J}_{k_j}); (\mathcal{D}'_{\gamma,\overline{\gamma}}(U,\mathcal{C};E),\mathcal{J}_k))$  respectively, and finally $C$ equicontinuous in $(\mathcal{D}'_{\gamma',\overline{\gamma'}}(U,\mathcal{C}';E'),\mathcal{J}_{k'})'\simeq (\mathcal{D}'_{-(\gamma')^c,-(\gamma')^c(\overline{k'})}(U,(\mathscr{O}_{\mathcal{C}'})^o;(E')^*),\mathcal{I}_{ibi})$.

We have to bound $$\sup_{b_i\in B_i,c\in C, u\in A}|\langle(u\circ_j v)(b_1,....,b_{n+m-1}),c\rangle|, \sup_{b_i\in B_i,c\in C, v\in A'}|\langle(u\circ_j v)(b_1,....,b_{n+m-1}),c\rangle|$$ since without the sup over $A,A'$ one gets a seminorm of the target space so that we have to bound this uniformly in $A,A'$ by a seminorm for $v,u$ respectively to get hypocontinuity.

The second case is easy since $A''=A'(B_j,...,B_{j+n-1})$ is bounded in the target space for $v$ namely $(\mathcal{D}'_{\gamma,\overline{\gamma}}(U,\mathcal{C};E),\mathcal{J}_{k}),k=2,4$ by definition of $A'$ bounded
, but, as we said, in this space, any bounded set is equicontinuous 
 thus so is $A''$, and :
$$\sup_{b_i\in B_i,c\in C, v\in A'}|\langle(u\circ_j v)(b_1,....,b_{n+m-1}),c\rangle|\leq \sup_{b_i\in B_i,v_j\in A'',c\in C}|\langle(u(b_1,....,b_{j-1},v_{j},b_{j+n},...,b_{n+m-1}),c\rangle|$$ is the bound by a seminorm for $u$. 

The first case will require the restriction $k=4$ to get full hypocontinuity. 
Indeed, likewise $A(B_1,....,B_{j-1},\cdot,B_{j+n},...,B_{n+m-1})$ is, by \cite[Prop 2 bis p 28]{Schwarz}, equicontinuous   from $(\mathcal{D}'_{\gamma,\overline{\gamma}}(U,\mathcal{C};E),\mathcal{I}_{ibi})$ to $(\mathcal{D}'_{\gamma',\overline{\gamma'}}(U',\mathcal{C}';E'),\mathcal{J}_{k'})$ (where we used our computations of strong duals $(\mathcal{D}'_{\gamma,\overline{\gamma}}(U,\mathcal{C};E),\mathcal{I}_{ibi})$ of Arens duals $(\mathcal{D}'_{-(\gamma)^c,-(\gamma)^c(\overline{k})}(U,(\mathscr{O}_\mathcal{C'})^o;E^*),\mathcal{I}_{ibi})$  of $(\mathcal{D}'_{\gamma,\overline{\gamma}}(U,\mathcal{C};E),\mathcal{J}_{k})$ to formulate things under Schwartz hypothesis. The strong dual computation uses proposition \ref{quasiCompletion}. This is where we will need $k=4$ so that the first space does not change).  
 thus, by  \cite[\S 39.3.(4)]{Kothe2},  the set of adjoint maps send  the equicontinuous set  like  $C$ to a bounded set $C'=[A(B_1,....,B_{j-1},\cdot,B_{j+n},...,B_{n+m-1})]'(C)$ which is again equicontinuous when $k=4$ because of the target space has also topology $\mathcal{I}_{ibi}$ where bounded sets are equicontinuous. 
 
Thus we have the expected bound 
$$\sup_{b_i\in B_i,c\in C, v\in A'}|\langle(u\circ_j v)(b_1,....,b_{n+m-1}),c\rangle|\leq \sup_{b_i\in B_i,c\in C', v\in A'}|\langle v(b_j,....,b_{j+n-1}),c\rangle|.$$ 
In case $k=2$ for the first argument fixed to a single element thus  $A=\{u\}$ is $\epsilon$-equihypocontinuous by definition, or also more generally for $A$ $\epsilon$-equihypocontinuous, thus 

\noindent $[A(B_1,....,B_{j-1},\cdot,B_{j+n},...,B_{n+m-1})]$ is equicontinuous thus by \cite[\S 39.3.(4)]{Kothe2} the set of adjoint maps  $[A(B_1,....,B_{j-1},\cdot,B_{j+n},...,B_{n+m-1})]'$ send equicontinuous sets in the dual to equicontinuous sets which enables to conclude as before.

For the induced maps, we get continuity by realizing the restriction as a canonical composition of 3 maps. The first is the tensor map $t$ obtained in lemma \ref{hypocontinuity} (ii) (case $\kappa_3=3,4$),(iii) (case $\kappa_3=1$) and (iv) (case $\kappa_3=2$, explaining the various classes of hypocontinuity obtained) with target space  $\mathcal{D}'_{\gamma_a{\times}\gamma_b,{(\gamma_a\dot{\times}\gamma_b)}(\kappa_3)}(V\times V',(\mathscr{C}\times \mathscr{C}')^{oo};\mathscr{E}\otimes \mathscr{E}'),\mathcal{J}_{\kappa_3})$ (we even got to the completion of the target space obtained there to get a common target in all our cases. 

The second map is a pullback via the map (recall $U_j'=U$) $$f_j:\mathscr{U}:= U_1'\times ... \times U_m'\times U'\times U_1\times ...\times U_n\to V\times V'$$ with $f_j(x_1,....,x_m,y,z_1,...,z_n)=(x_1,....,x_m,y,z_1,...,z_n,x_j)$,
so that $df_j^*(\xi_1,....,\xi_m,\eta,\zeta_1,...,\zeta_n,\eta')=(\xi_1,....,\xi_j+\eta',...,\xi_m,\eta,\zeta_1,...,\zeta_n)$ and thus $df_j^*(\gamma_a\dot{\times}\gamma_b)=\{(x_1,....,x_m,y,z_1,...,z_n,\xi_1,....,\xi_j+\eta',...,\xi_m,\eta,\zeta_1,...,\zeta_n): (x_1,....,x_m,y,z_1,...,z_n,x_j,\xi_1,....,\xi_m,\eta,\zeta_1,...,\zeta_n,\eta')\in (\gamma_b\dot{\times}\gamma_a)\cap \dot{T}^*(V\times V')\}$. If $X=(x_1,....,x_m,y,z_1,...,z_n,\xi_1,....,\xi_j+\eta',...,\xi_m,\eta,\zeta_1,...,\zeta_n)\in df_j^*(\gamma_b\dot{\times}\gamma_a)\cap \dot{T}^*(V\times V')$ and assume moreover $\xi_j+\eta'=0.$ Let us show that $X\not\in (-\gamma_1'\dot{\times}...-\gamma_{j-1}'\dot{\times}\{0\}...\dot{\times} -\gamma_m'\dot{\times} (\gamma')^c\dot{\times}-\gamma_1\dot{\times}...\dot{\times}-\gamma_{n})$. If it were, we would have ($(x_i,\xi_i)\in -\gamma_i'$ or $\xi_i=0$) for all $i\neq j$, and ($(z_i,\zeta_i)\in -\gamma_i$ or $\zeta_i=0$) for all $i$ and moreover $(y,\eta)\in (\gamma')^c$ or $\eta=0$. Thus if $\eta'=\xi_j=0$ one contradicts the assumption $(x_1,....,x_m,y,z_1,...,z_n,x_j,\xi_1,....,\xi_m,\eta,\zeta_1,...,\zeta_n,\eta')\in (\gamma_b\dot{\times}\gamma_a)\cap \dot{T}^*(U^{n+m+2})$ (since from the assumptions on $(z_i,\zeta_i)$ one cannot have ($(x_j,\eta')\in \gamma^c$ or $\eta'=0$) especially $\eta'=0$ is forbidden because of the definition of $\gamma_a$ and also similarly, one cannot have ($(x_j,\xi_j)\in -\gamma_j=\gamma$ or $\xi_j'=0$ from  the definition of $\gamma_b$) or otherwise if $\eta'\neq 0$ and thus in order not to contradict $(z_1,...,z_n,x_j,\zeta_1,...,\zeta_n,\eta')\in \gamma_a$ one must have $(x_j,\eta')\in \gamma$  and similarly since then $\xi_j\neq 0$, in order not to contradict $(x_1,....,x_m,y,\xi_1,....,\xi_m,\eta)\in \gamma_b$ one must have $(x_j,\xi_j)\in (-\gamma'_j)^c=(-\gamma)^c$ and this contradicts $\xi_j=-\eta$.
 Altogether we deduced 
$$df_j^*(\gamma_a\dot{\times}\gamma_b)\cap T^*U^{j-1}\times\{0\}\times T^*U^{n+m+1-j}\subset (-\gamma_1'\dot{\times}...-\gamma_{j-1}'\dot{\times}\{0\}...\dot{\times} -\gamma_m'\dot{\times} (\gamma')^c\dot{\times}-\gamma_1\dot{\times}...\dot{\times}-\gamma_{n})^c.$$

Similarly, let us show that \begin{align*}&f_j^{-1}((\mathscr{C}'\times \mathscr{C})^{oo})\subset \mathscr{C}'''\\&:=((\mathscr{O}_{\mathcal{C}_1'})^{oo}\times...\times (\mathscr{O}_{\mathcal{C}_{j-1}'})^{oo}\times\mathcal{K}^o\times(\mathscr{O}_{\mathcal{C}_{j+1}'})^{oo}\times...\times (\mathscr{O}_{\mathcal{C}_{m}'})^{oo} \times\mathcal{C'}^o\times (\mathscr{O}_{\mathcal{C}_1})^{oo}\times...\times (\mathscr{O}_{\mathcal{C}_{n}})^{oo})^o\end{align*}
Thus take $O_i\in (\mathscr{O}_{\mathcal{C}_i'})^{oo}, V_i\in(\mathscr{O}_{\mathcal{C}_i})^{oo}, W\in \mathcal{K}^o,W'\in \mathcal{C'}^o$ and let $A=O_1\times...\times O_{j-1}\times W\times O_{j+1}\times ... O_m\times W'\times V_1\times ...\times V_n$, and let $B\in \mathscr{C}',C\in \mathscr{C}$, then $A\cap f_j^{-1}(B\times C)\subset f_j^{-1}(f_j(A)\cap (B\times C))\subset f_j^{-1}[(A'\cap B)\times (A''\cap C)]$ with $A'=O_1\times...\times O_{j-1}\times W\times O_{j+1}\times ... O_m\times W', A''= V_1\times ...\times V_n\times W$. But note that if $p_j$ is the projection on the $j$-th coordinate, we have $T:=p_j(A'\cap B)\in (\mathscr{O}_{\mathcal{C}_j'})^{o}$ since if $W''\in (\mathscr{O}_{\mathcal{C}_j'})^{oo},$ $W''\cap p_j(A'\cap B)\subset p_j(\overline{B'\cap B})$ which is compact, since $B':=O_1\times...\times O_{j-1}\times W''\times O_{j+1}\times ... O_m\times W'\in (\mathscr{C}')^o$. Likewise $S:=p_{n+1}(A''\cap C)\in \mathcal{C}=\mathcal{C}'_j $. Let $\mathscr{U}_1\times U_j'\times \mathscr{U}_2:= U_1'\times ... \times U_m'\times U',\mathscr{U}_3:= U_1\times ... \times U_n.$ But from this and  the definition of $f_j$ we obviously have \begin{align*}f_j^{-1}[(A'\cap B)\times (A''\cap C)]&=f_j^{-1}[(A'\cap B)\cap(\mathscr{U}_1\times T\times \mathscr{U}_2)\times (A''\cap C\cap (\mathscr{U}_3\times S))]\\&=f_j^{-1}[(A'\cap B)\cap(\mathscr{U}_1\times S\times \mathscr{U}_2)\times (A''\cap C\cap (\mathscr{U}_3\times T))].\end{align*}
thus $A\cap f_j^{-1}(B\times C)\subset  f_j^{-1}[\overline{(D'\cap B)}\times \overline{(D''\cap C)}]$ with $D'=O_1\times...\times O_{j-1}\times Int(S_\epsilon)\times O_{j+1}\times ... O_m\times W'\in (\mathscr{C}')^o, D''= V_1\times ...\times V_n\times Int(T_\epsilon)\in \mathscr{C}^o$,for some $\epsilon_i>0$, by enlargeability, so that $\overline{(D'\cap B)}, 
\overline{(D''\cap C)}$ are compact and thus since $Im(f_j)$ is closed and $f_j$ has a continuous inverse on its image, one deduces $f_j^{-1}[\overline{(D'\cap B)}\times \overline{(D''\cap C)}]$ and then $\overline{A\cap f_j^{-1}(B\times C)}$ compact and thus the stated inclusion $f_j^{-1}((\mathscr{C}'\times \mathscr{C})^{oo})\subset \mathscr{C}'''$ since by \cite{Dab14a}{Ex 5} 
 any closed set in there is included in a set of the form above $f_j^{-1}(B\times C)$. Again $\overline{A\cap (f_j^{-1}(B\times C))_\epsilon}\subset(\overline{A_{\delta(\epsilon)}\cap (f_j^{-1}(B\times C))})_{u\epsilon} $ which is also compact for $\epsilon=\epsilon(B,C)$ small enough, thus we even have $(f_j)_e^{-1}((\mathscr{C}'\times \mathscr{C})^{oo})\subset \mathscr{C}''',$ for some function $\epsilon$ given above (recall the notation $(f_j)_e^{-1}$ introduced in proposition \ref{pullback} contains an implicit $\epsilon$).

From the application of proposition \ref{pullback}, we thus got our second map after composition by a canonical map :$$f_j^*:(\mathcal{D}'_{\gamma_a{\times}\gamma_b,{(\gamma_a\dot{\times}\gamma_b)(\kappa_3)}}(V\times V',(\mathscr{C}\times \mathscr{C}')^{oo};\mathscr{E}\otimes \mathscr{E}'),\mathcal{J}_{\kappa_3})
\to (\mathcal{D}'_{\gamma_d,\gamma_d(\kappa_3)}(V\times \mathscr{U}_3,\mathscr{C}''';f_j^*(\mathscr{E}\otimes \mathscr{E}')),\mathcal{J}_{\kappa_3})$$ with \begin{align*}\gamma_d:=&[(-\gamma_1'\dot{\times}...-\gamma_{j-1}'\dot{\times}\{0\}...\dot{\times} -\gamma_m'\dot{\times} (\gamma')^c\dot{\times}-\gamma_1\dot{\times}...\dot{\times}-\gamma_{n})^c\cap (T^*\mathscr{U}_1\times\{0\}\times T^*(\mathscr{U}_2\times \mathscr{U}_3))] \\&\cup T^*\mathscr{U}_1\times\dot{T}^*(U_j')\times T^*(\mathscr{U}_2\times \mathscr{U}_3)\end{align*} (our computation above implies the right condition in the hypothesis on the 0 section, $\gamma_d\subset \dot{T}^*(V\times \mathscr{U}_3)$). 

The third map necessary to recover the restricted composition map is easier to build. Recall $\mathscr{E}''=(E_1')^*\otimes ...\otimes (E_{j-1}')^*\otimes (E_{1})^*\otimes ...\otimes(E_{n})^*\otimes (E_{j+1}')^*\otimes...\otimes (E_{n}')^*\otimes E'\to V''$ and define $\mathscr{E}'''=(E_1')^*\otimes ...\otimes (E_{j-1}')^*\otimes (\C)^*\otimes (E_{j+1}')^*\otimes...\otimes (E_{n}')^*\otimes E'\otimes (E_{1})^*\otimes ...\otimes(E_{n})^*\to V\times \mathscr{U}_3.$ The third map is then defined from the projection map composed with a permutation $g_j(x_1,....,x_m,y,z_1,...,z_n)=(x_1,....,x_{j-1},z_1,...,z_n,x_{j+1},....,x_{m},y),$ so that with the first part of proposition \ref{pullback} again, we get $$g_j^*:(\mathcal{D}'_{-\gamma_c^c,-\gamma_c^c(\overline{\kappa_3})}(V'',(\mathscr{O}_{\mathscr{C}''})^o;(\mathscr{E}'')^*),\mathcal{J}_{\overline{\kappa_3}})\to (\mathcal{D}'_{-\gamma_d^c,-\gamma_d^c(\overline{\kappa_3})}(V\times \mathscr{U}_3,(\mathscr{O}_{\mathscr{C}'''})^o;g_j^*((\mathscr{E}'')^*)=(\mathscr{E}''')^*),\mathcal{J}_{\overline{\kappa_3}})$$ since by a computation already used in the proof of proposition \ref{InjectionProjectiveProduct} (using enlargeability) 
\begin{align*}(\mathscr{O}_{\mathscr{C}'''})^o&=(\mathcal{C}_1'\times...\times \mathcal{C}_{j-1}'\times\mathcal{F}\times\mathcal{C}_{j+1}'\times...\times {\mathcal{C}_{m}'} \times\mathscr{O}_{\mathcal{C'}}^o\times {\mathcal{C}_1}\times...\times \mathcal{C}_{n})^{oo}\\&\supset g_j^{-1}(\mathscr{O}_{\mathscr{C}''}^o)=g_j^{-1}((\mathcal{C}_1'\times...\times \mathcal{C}_{j-1}'\times {\mathcal{C}_1}\times...\times \mathcal{C}_{n} \times\mathcal{C}_{j+1}'\times...\times {\mathcal{C}_{m}'} \times\mathscr{O}_{\mathcal{C'}}^o)^{oo})\end{align*}
where we used again \cite{Dab14a}{Ex 5}
 to see any element of $(\mathcal{C}_1'\times...\times \mathcal{C}_{j-1}'\times {\mathcal{C}_1}\times...\times \mathcal{C}_{n} \times\mathcal{C}_{j+1}'\times...\times {\mathcal{C}_{m}'} \times\mathscr{O}_{\mathcal{C'}}^o)^{oo}$ is included in some product and thus $g_j^{-1}((\mathcal{C}_1'\times...\times \mathcal{C}_{j-1}'\times {\mathcal{C}_1}\times...\times \mathcal{C}_{n} \times\mathcal{C}_{j+1}'\times...\times {\mathcal{C}_{m}'} \times\mathscr{O}_{\mathcal{C'}}^o)^{oo})\subset [g_j^{-1}((\mathcal{C}_1'\times...\times \mathcal{C}_{j-1}'\times {\mathcal{C}_1}\times...\times \mathcal{C}_{n} \times\mathcal{C}_{j+1}'\times...\times {\mathcal{C}_{m}'} \times\mathscr{O}_{\mathcal{C'}}^o))]^{oo}.$


Taking the adjoint of $g_j^*$ continuous between Mackey duals, one gets the third map we wanted :$(g_j)_*:=(g_j^*)^*\circ Tr:(\mathcal{D}'_{\gamma_d,\gamma_d(\kappa_3)}(V\times \mathscr{U}_3,\mathscr{C}''';f_j^*(\mathscr{E}\otimes \mathscr{E}')),\mathcal{J}_{\kappa_3})\to (\mathcal{D}'_{\gamma_c,\gamma_c(\kappa_3)}(V'',\mathscr{C}'';\mathscr{E}''),\mathcal{J}_{\kappa_3})$ where $Tr$ is the map induced from the bundle map $E\otimes E^*\to \C^*.$ 
Thus $\circ_j=(g_j)_*\circ f_j^*\circ t$ and it suffices to check it agrees on a dense set thus everywhere with the restricted map.
\end{proof}

Beyond composition studied in the previous lemma, we will be interested in tensor products on spaces of multilinear maps and their continuity properties. The useful easy results are gathered in the next lemma which is a consequence of associativity of $\epsilon$ product \cite[Prop 7 p 38]{Schwarz} and the canonical map between projective tensor product and $\epsilon$-product (both in the complete case (with approximation property for the second point requiring identification with completed injective product).

\begin{lemma}\label{MultilinearTensor}In the setting above, and with supplementary cones $\gamma_i',\gamma', i\in[1,m]$ polar enlargeable families $\mathcal{C}'_i,\mathcal{C}'$ 
assume we write $\gamma''_i=\gamma'_{i}$ for $i\in[1,m]$, $\gamma''_{m+1}=(-\gamma')^c$ ,$\gamma''_i=\gamma_{i-m-1}$ for $i\in[m+2,m+n+1]$, and similarly for $\mathcal{C}''_j$ (especially $\mathcal{C}_{m+1}''=\mathscr{O}_{\mathcal{C}'})^o$). Then  we have a continuous map corresponding to tensor multiplication at the level of distribution spaces 
\begin{align*}.\otimes.&:L({\bigotimes}_{\beta e, i\in [1,n]} (\mathcal{D}'_{\gamma_i',{\gamma}_i'}(U,\mathcal{C}_i),\mathcal{I}_{b}); (\mathcal{D}'_{\gamma',\overline{\gamma'}}(U,\mathcal{C}')),\mathcal{I}_{ppp})))\times\\& \ \  L({\bigotimes}_{\beta e, i\in [1,n]} (\mathcal{D}'_{\gamma_i,{\gamma}_i}(U,\mathcal{C}_i),\mathcal{I}_{b}); (\mathcal{D}'_{\gamma,\overline{\gamma}}(U,\mathcal{C})),\mathcal{I}_{ppp})))\\& \ \ \ \ \to L({\bigotimes}_{\beta e, i\in [1,n+m-1]} (\mathcal{D}'_{\gamma''_i,\gamma''_i}(U,\mathcal{C}''_i),\mathcal{I}_{b}); (\mathcal{D}'_{\gamma,\overline{\gamma}}(U,\mathcal{C})\mathcal{I}_{ppp}))
\end{align*}

\end{lemma}

\part{Spaces of smooth functionals with wave front set conditions.}
\section{Topology and algebraic structure}
In this section, we want to study topological (and algebraic) properties of some spaces of  conveniently smooth \cite{KrieglMichor} (vector valued) functionals on an open 
 $\mathscr{U}\subset \mathcal{E}(U,E)$ (open in the usual Fr\'echet  space topology thus $c^\infty$-open, \cite[Th 4.11]{KrieglMichor}) for $E\to U$ a bundle as above. These spaces will be variants of microcausal functionals (studied in this generality in \cite{RibeiroBF}) when specified to specific cones. We will moreover study two classes of spaces. The goal is to solve easily an issue appeared there for defining Poisson (and retarded) brackets with field dependent propagator. This uses crucially our study of multilinear maps on generalized H\"ormander spaces of distributions.
 
The first variant is the more similar to \cite{RibeiroBF}. Fix $m\in \N^*$. Let $\lambda=(\lambda_n)$ a family of 
 cones $\lambda_n\subset (\dot{T}^*U^n)^J$ for some set $J,$ $\lambda_n=(\lambda_{n,j})_{j\in J}$. We will often assume :
\medskip

\begin{minipage}{15,5cm} \textbf{Assumption 1}:
 $\lambda_{n,j}\dot{\times} \lambda_{N,j}\subset \lambda_{n+N,j}, \forall j\in J;n,N\geq 1.$ 
 \end{minipage}
 \medskip

We also consider $\mathcal{C},\mathcal{C}'$ polar enlargeable families of closed sets. Note that $(\mathcal{C}^n)^{oo}$ is then also enlargeable. We will fix $\mathfrak{E}$ a separated locally convex space of value (as motivated for instance by \cite{RejznerYangMills,RejznerYangMills2,Rejzner, RejznerBF}, but note that our investigation at this stage is only really preliminary, and the assumption on $\mathfrak{E}$ won't be suitable to apply it as is to these papers, a more general treatment will require the investigation of the third part of this paper). This locally convex space will enable to treat field dependent retarded product in a uniform way as functionals valued in a locally convex space. 

For a smooth map $F:\mathscr{U}\to \mathfrak{E}$ it is known \cite[Corol 5.11]{KrieglMichor} that the $n$-th differential $F^{(n)}\in C^{\infty}(U,L_b(\mathcal{E}(U,E)^{\otimes_\beta n};\mathfrak{E}))$ is bounded (even symmetric) multilinear map. But since $\mathcal{E}(U,E)$ is a Fr\'echet space it is known that $\mathcal{E}(U,E)^{\hat{\otimes}_\beta n}=\mathcal{E}(U,E)^{\hat{\otimes}_\pi n}=\mathcal{E}(U^n,E^{\otimes n})$ (see e.g. \cite[Cor of Th 34.1,Th 51.6]{Treves} in the non-bundle case) which is bornological so that  $F^{(n)}\in C^{\infty}(U,\mathcal{E}'(U^n,(E^{\otimes n})^*)\epsilon \mathfrak{E})\subset C^{\infty}(U,\mathcal{D}'(U^n,(E^{\otimes n})^*)\epsilon \mathfrak{E}).$ Indeed smooth maps depend only of the bornology \cite[Corol 1.8]{KrieglMichor} which is the same between $\mathcal{E}'(U^n,E^{\otimes n})\epsilon \mathfrak{E}$ and $L_b(\mathcal{E}(U,E)^{\otimes_\beta n};\mathfrak{E})$ (see \cite[corol 2 p 34]{Schwarz} with our identification above and since $\mathcal{E}(U^n,E^{\otimes n})$ and its dual have their Mackey topology since they are bornological, we also use their equicontinuous sets are exactly their bounded set for the same reason).

Note that from \cite[lemma 2.3.8]{RibeiroBF} there is a natural notion of support of a functional (when $\mathfrak{E}=\C$), such that in the smooth case $\text{supp}(F)=\overline{\bigcup_{\varphi\in \mathscr{U}}\text{supp}(F^{(1)}[\varphi])}.$ For the case of more general $\mathfrak{E}$, we suppose given a bornology $\mathscr{B}$ on some subspace of  $\mathfrak{E}'.$ Then we can define for $B\in \mathscr{B},$ $\text{supp}(F,B)=\overline{\bigcup_{\varphi\in \mathscr{U}, \psi\in B}\text{supp}(F^{(1)}[\varphi](\psi))},$ where, as seen above, $F^{(1)}[\varphi](\psi)\in \mathcal{E}'(U,E^*).$ When not specified, we take the bornology $\mathscr{B}_f$ of all subsets of the dual, especially in the case $\mathfrak{E}=\C$ (in this case we only write $\text{supp}(F)$). In this subsection, this bornology
can be ignored, it will be used in the next subsection. At this stage, for instance, there is no reason to distinguish between scalarly compact support and real compact support (for $\mathscr{B}_f$). We also fix a family of locally convex spaces with continuous injections $\mathfrak{E}\to \mathfrak{E}_i, i\in J.$

We define a first space of functionals on an open $\mathscr{U}\subset \mathcal{E}(U,E)$ :\begin{align*}&\mathscr{F}_{\lambda}(\mathcal{E}(U,E),\mathscr{U},\mathcal{C},\mathcal{C}';(\mathfrak{E},(\mathfrak{E}_i)_{i\in J},\mathscr{B}))=\{F:\mathscr{U}\to \mathfrak{E} |\ F \text{smooth},\ \forall n\in \N^*,\ \forall j\in J,\\& F^{(n)}\in C^\infty(\mathscr{U}, (\mathcal{D}'_{\lambda_{n,j},\overline{\lambda_{n,j}}}(U^n,(\mathcal{C}^n)^{oo};(E^{\otimes n})^*),\mathcal{I}_{ibi})\epsilon \mathfrak{E}_i); \ \forall B\in \mathscr{B}, \text{supp}(F,B)\in \mathcal{C}'\}.\end{align*}

The second variant using more explicitly spaces of multilinear maps will be called an operadic variant, depending on $\Gamma=\{(\gamma_j,\gamma_j')\}_{j\in J}$
a set of pairs of  cones, then we define:

\begin{align*} \mathscr{F}_{\Gamma}&(\mathscr{U},\mathcal{C};(\mathfrak{E},(\mathfrak{E}_i)_{i\in J},\mathscr{B}))=\{F:\mathscr{U}\to \mathfrak{E} |\ F\ \text{smooth}\ ,  \ \forall B\in \mathscr{B}, \text{supp}(F,B)\in\mathcal{C},\ \forall i\in J, \\&\forall n\in \N^*, F^{(n)}\in \mathcal{C}^\infty(\mathscr{U},L_{\beta e}[(\mathcal{D}'_{\gamma_i,\overline{\gamma_i}}(U;E),\mathcal{I}_{ibi})^{n-1};(\mathcal{E}'_{\gamma_i',\overline{\gamma_i'}}(U;E^*),\mathcal{I}_{ibi})]\epsilon \mathfrak{E}_i)\}.\end{align*}

In most cases all $\mathfrak{E}_i=\mathfrak{E}$ and we only write $(\mathfrak{E},\mathscr{B})$ instead of $(\mathfrak{E},(\mathfrak{E})_{i\in J},\mathscr{B})$,  and $\mathfrak{E}$ instead of $(\mathfrak{E},\mathscr{B}_f)$ as we said.

We will often assume :
\medskip

\begin{minipage}{15,5cm} \textbf{Assumption 2}:
 $-\gamma_i\subset(\gamma_i')^c, \forall i\in J.$ 
 \end{minipage}
 \medskip

Let us give examples :

\begin{ex}\label{microlocal}\textbf{Local functionals :}
Let $J=\{1\}$, and $\lambda_n=\lambda_{n,1}=C_n=\{(x,...,x; \xi_1,...,\xi_n)\in \dot{T}^*(U^n): x\in U, \xi_1+...+\xi_n=0\}$ the conormal bundle to the small diagonal and write $\lambda=C.$ Note that $C_n\subset (-\gamma^{\dot{\times}(n-1)}\dot{\times}\gamma^c)^c$ as soon as $\gamma$ is a semigroup for addition. Then $\mathscr{F}_{C}(\mathcal{E}(U,E),\mathscr{U},\mathcal{K},\mathcal{K};\C)$ has an interesting closed subspace (for the topologies defined below)~:
$$\mathscr{F}_{\mu loc}(\mathscr{U};\C)=\{F\in \mathscr{F}_{C}(\mathcal{E}(U,E),\mathscr{U},\mathcal{K},\mathcal{K};\C): F \text{additive}\}$$
is one of the possible definitions of microlocal functionals in \cite{RibeiroBF}, where $F$ is said to be additive if $\forall \varphi_2\in \mathscr{U}$ and $\varphi_1,\varphi_3\in \mathcal{E}(U,E)$ such that $\varphi_1+\varphi_2,\varphi_2+\varphi_3, \varphi_1+\varphi_2+\varphi_3\in \mathscr{U}$  and $\text{supp}(\varphi_1)\cap\text{supp}(\varphi_3)=\emptyset $ then :
$$F(\varphi_1+\varphi_2+\varphi_3)=F(\varphi_1+\varphi_2)+F(\varphi_3+\varphi_2)-F(\varphi_2).$$
We will thus see below it is complete nuclear for the topologies we will soon define. Note that $C$ does not satisfy assumption 1.
\end{ex}

\begin{ex}\label{microcausal}\textbf{Microcausal functionals :}
Let $J=\{1,2\}$, $U=M$ a (time oriented) globally hyperbolic manifold (for a metric $g$) and $\mu_{n,1}^c=\overline{V_+}^n=\{(x_1,...,x_n; \xi_1,...,\xi_n)\in \dot{T}^*(M^n):  \xi_i\in \overline{V_+(x_i)}\}$ with  $\overline{V_+(x_i)}$ the closure of  $V_+(x_i)$ the future light cone at $x_i$  and  $\mu_{n,2}^c=\overline{V_-}^n$ analogously. Then the minimal variant of the space of microcausal functionals (see e.g. \cite{RibeiroBF}) with $DWF(F^{(n)})\subset \mu_{n,1}\cap\mu_{n,2}$ (instead of the same condition on wave front sets in the usual approach) is: $$\mathscr{F}_{\mu c}((M,g),\mathscr{U};\mathfrak{E})=\mathscr{F}_{\mu}(\mathcal{E}(U,E),\mathscr{U},\mathcal{K},\mathcal{K};\mathfrak{E}).$$

We have the operadic variant  if we define $\gamma_{mc,1}=\overline{V_+}, \gamma_{mc,1}'=(\overline{V_-})^c,\gamma_{mc,2}=\overline{V_-}, \gamma_{mc,2}'=(\overline{V_+})^c,$ gathered in $\Gamma_{mc}=\{(\gamma_{mc,j},\gamma_{mc,j}')\}_{j=1,2}$ which satisfies assumption 2. We can define :

$$\mathscr{F}_{\mu c}((M,g),\mathscr{U};\mathfrak{E})\hookrightarrow\mathscr{F}_{mc}((M,g),\mathscr{U};\mathfrak{E}):=\mathscr{F}_{\Gamma_{mc}}(\mathscr{U},\mathcal{K};\mathfrak{E}).$$
The injection (which will soon be continuous) follows from proposition \ref{multilinear} since for instance $((-\gamma_{mc,1})^{\dot{\times}(n-1)}\dot{\times}(\gamma_{mc,1}')^c)=(\overline{V_-}^{\dot{\times}n})^c=\mu_{n,2}.$
\end{ex}

\begin{ex}\label{Opmicrocausal}\textbf{Operadically Microcausal functionals :}

The inconvenience in some applications of the last example is that $\gamma\neq \gamma'$ so that it is not possible to compose multilinear maps obtained via derivatives even in the case $E=\C$. For $G$ a family of globally hyperbolic metrics on $M$, consider $J=G\times\{1,2\}.$
We define $\gamma_{omc,(g,1)}=\gamma_{omc,(g,1)}'={V_+(g)}, \gamma_{omc,(g,2)}=\gamma_{omc,(g,2)}'={V_-(g)},$ gathered in $\Gamma_{omc}=\{(\gamma_{omc,(g,j)},\gamma_{omc,(g,j)}')\}_{(g,j)\in J}$, which satisfies assumption 2, and the space of operatorially microcausal functionals :
$$\mathscr{F}_{omc}(M,G,\mathscr{U};\mathfrak{E}):=\mathscr{F}_{\Gamma_{omc}}(\mathscr{U},\mathcal{K};\mathfrak{E}).$$
We also have a variant with closed cones $\Gamma_{omc2}=\{(\overline{\gamma_{omc,(g,j)}},\overline{\gamma_{omc,(g,j)}'})\}_{(g,j)\in J}$:
$$\mathscr{F}_{omc2}(M,G,\mathscr{U};\mathfrak{E}):=\mathscr{F}_{\Gamma_{omc2}}(\mathscr{U},\mathcal{K};\mathfrak{E})\hookrightarrow\mathscr{F}_{mc}((M,g),\mathscr{U};\mathfrak{E}), (g\in G).$$

To compare with examples already considered, let us make explicit the variants with wave front set conditions.

 Define $o\mu_{n,(g,i)}=((-\gamma_{omc,(g,i)})^{\dot{\times}(n-1)}\dot{\times}(\gamma_{omc,(g,i)}')^c)=({V_\pm(g)}^{\dot{\times}(n-1)}\dot{\times}{V_\mp(g)}^c)^c$ and then $o\mu_n=(o\mu_{n,(g,i)})_{(g,i)\in J},$ with variant $o\mu'_{n,(g,i)}=((-\gamma_{omc2,(g,i)})^{\dot{\times}(n-1)}\dot{\times}(\gamma_{omc2,(g,i)}')^c)=({\overline{V_\pm(g)}}^{\dot{\times}(n-1)}\dot{\times}[{\overline{V_\mp(g)}}]^c)^c$ and $o\mu'_n=(o\mu'_{n,(g,i)})_{(g,i)\in J},$ and then deduce
$$\mathscr{F}_{o\mu c}(M,G,\mathscr{U};\mathfrak{E}):=\mathscr{F}_{o\mu}(\mathcal{E}(U,E),\mathscr{U},\mathcal{K},\mathcal{K};\mathfrak{E})\hookrightarrow\mathscr{F}_{omc}(M,G,\mathscr{U};\mathfrak{E}),$$
$$\mathscr{F}_{o\mu c2}(M,G,\mathscr{U};\mathfrak{E}):=\mathscr{F}_{o\mu'}(\mathcal{E}(U,E),\mathscr{U},\mathcal{K},\mathcal{K};\mathfrak{E})\hookrightarrow\mathscr{F}_{omc2}(M,G,\mathscr{U};\mathfrak{E})\cap \mathscr{F}_{\mu c}((M,g),\mathscr{U};\mathfrak{E}), (g\in G).$$

Note that $\mathscr{F}_{o\mu c2}(M,\{g\},\mathscr{U};\mathfrak{E})$ is the space introduced in \cite[p 14]{ColliniH} (once considered the symmetry of differentials that allows not to state symmetric conditions with respect to which variable is the last variable). Note that $\mathscr{F}_{\mu loc}(\mathscr{U};\C)\subset \mathscr{F}_{o\mu c}(M,G,\mathscr{U};\C)$ using as explained  in the example for multilocal functionals that $V_\pm(g)$ are semigroups for addition.
\end{ex}

In the case $\mathcal{C}'=\mathcal{F}$, there are 2 natural topologies. We can put  the topology of uniform convergence on images of finite dimensional compacts by smooth functions $\mathcal{I}_{s}$ or the topology of uniform convergence on images of compacts by smooth curves $\mathcal{I}_{c}$. $\mathcal{I}_{s}$ (resp. $\mathcal{I}_{c}$) is the initial topology generated, for any $f_\phi:\R^n\to \mathcal{E}(U,E)$ for $n\geq 1$ (resp. for $n=1$) by the maps $f_\phi^{*(k,i)}:\mathscr{F}_{\lambda}(\mathcal{E}(U,E),\mathscr{U},\mathcal{C},\mathcal{C}';(\mathfrak{E},\mathscr{B}))\to \mathcal{E}[\R^n,(\mathcal{D}'_{\lambda_{k,i},\overline{\lambda_{k,i}}}(U^k,(\mathcal{C}^k)^{oo};(E^{\o n})^*),\mathcal{I}_{ibi})\epsilon \mathfrak{E}]$ (this last space being given the usual topology of uniform convergence on compacts of all derivatives). We have $$\mathcal{E}[\R^n,(\mathcal{D}'_{\lambda_{k,i},\overline{\lambda_{k,i}}}(U^k,(\mathcal{C}^k)^{oo};(E^{\o n})^*),\mathcal{I}_{ibi})\epsilon \mathfrak{E}]\simeq \mathcal{E}[\R^n]\epsilon[(\mathcal{D}'_{\lambda_{k,i},\overline{\lambda_{k,i}}}(U^k,(\mathcal{C}^k)^{oo};(E^{\otimes n})^*),\mathcal{I}_{ibi})\epsilon \mathfrak{E}]$$ even if $\mathfrak{E}$  is not quasi-complete\footnote{Indeed it is enough to deal with the case $k=0$ and $\mathcal{E}[\R^n,\mathfrak{E}]$ defines an element of $\mathcal{E}[\R^n]\epsilon  \mathfrak{E}=L_\epsilon(\mathfrak{E}_c',\mathcal{E}[\R^n])$ but conversely, an element of the $\epsilon$ product defines a scalarly smooth map, thus composed with curves, a scalarly smooth curve, but from \cite[Corol 1.8]{KrieglMichor} smoothness of curves depends only of the bornology, thus our element of $\mathcal{E}[\R^n]\epsilon  \mathfrak{E}$ defines a conveniently smooth map (with the original topology of $E$) and, from \cite[Corol 3.14]{KrieglMichor}, this is the usual notion of smooth maps. Then the identification of topologies follows as in \cite[Th 44.1]{Treves}}, 
and the second bracket can be removed in the quasi-complete case (\cite{Schwarz} for associativity of $\epsilon$ product). Of course  
$f_\phi^{*(k,i)}(F)=F^{(k)}\circ f_\phi.$

For the second type of spaces, we likewise define $\mathcal{I}_{s}$, $\mathcal{I}_{c}$ using : $$f_\phi^{*(k,i)}:\mathscr{F}_{\Gamma}(\mathscr{U},\mathcal{C};(\mathfrak{E},\mathscr{B}))\to \mathcal{E}(\R^n,L_{\beta e}[(\mathcal{D}'_{\gamma_i,\overline{\gamma_i}}(U;E),\mathcal{I}_{ibi})^{n-1};(\mathcal{E}'_{\gamma_i',\overline{\gamma_i'}}(U;E^*),\mathcal{I}_{ibi})]\epsilon \mathfrak{E}).$$

The topologies above are also defined by the following families of seminorms. For $K\subset \R^m$ a compact and $f:\R^m\to \mathscr{U}\subset \mathcal{E}(U)$ a smooth map, $C$ an equicontinuous set in $(\mathcal{D}'_{-\lambda_n^c,-\overline{\lambda_n^c}}(U^n,(\mathcal{C}^n)^{o};E^{\o n})),\mathcal{I}_{ibi})$, $CE$ an equicontinuous set in $\mathfrak{E}'_c$, one defines for the first case: 

$$p_{f,K,CE}(F)=\sup_{\phi\in f(K), e\in CE}| \langle F(\phi),e\rangle|,$$
$$p_{n,f,K,C,CE}(F)= \sup_{\phi\in f(K)}\sup_{v\in C, e\in CE}|\langle F^{(n)}(\phi),v\otimes e\rangle|.$$

For the second case, we keep the seminorms $p_{f,K,CE}$ and replace the second family, using equicontinuous=bounded sets $C\subset (\mathcal{D}'_{\gamma_i,\overline{\gamma_i}}(U),\mathcal{I}_{ibi})$ and $D\subset (\mathcal{D}'_{-(\gamma_i')^c,-\overline{(\gamma_i')^c}}(U),\mathcal{I}_{ibi})$ 
 (cf lemma \ref{BoundedRelated}) and consider the seminorms~:

$$p_{n,f,K,C,D}(F)= \sup_{\phi\in f(K)}\sup_{v\in D,u_1,...,u_{n-1}\in C,e\in CE}|\langle F^{(n)}(\phi)[e],u_1\otimes...\otimes u_{n-1}\otimes v\rangle|.$$

It is crucial in the case $\mathfrak{E}$ not necessarily quasi-complete that we see $$L_{\beta e}[(\mathcal{D}'_{\gamma_i,\overline{\gamma_i}}(U;E))^{n-1};\mathcal{E}'_{\gamma_i',\overline{\gamma_i'}}(U;E^*)]\epsilon \mathfrak{E}\simeq L_\epsilon(\mathfrak{E}'_c;L_{\beta e}[(\mathcal{D}'_{\gamma_i,\overline{\gamma_i}}(U;E),\mathcal{I}_{ibi})^{n-1};(\mathcal{E}'_{\gamma_i',\overline{\gamma_i'}}(U;E^*),\mathcal{I}_{ibi})])$$ thanks to \cite[Corol 2 p 34]{Schwarz}.

For cases with specified support, $\mathcal{C}'$ one of course uses inductive and projective limits, for $\alpha=s,c$ (with the same name for induced topologies on a subspace)
\begin{align*}(\mathscr{F}_{\lambda}&(\mathcal{E}(U,E),\mathscr{U},\mathcal{C},\mathcal{C}';(\mathfrak{E},\mathscr{B})),\mathcal{I}_{\alpha i})=\underrightarrow{\lim}_{C\in \mathcal{C'}} (\mathscr{F}_{\lambda}(\mathcal{E}(U,E),\mathscr{U},\mathcal{C},\{F\in \mathcal{F}, F\subset C\};(\mathfrak{E},\mathscr{B})),\mathcal{I}_{\alpha})
 \\&\to (\mathscr{F}_{\lambda}(\mathcal{E}(U,E),\mathscr{U},\mathcal{C},\mathcal{C}';(\mathfrak{E},\mathscr{B})),\mathcal{I}_{\alpha p}):=\underleftarrow{\lim}_{O\in (\mathcal{C}')^o}(\mathscr{F}_{\lambda}(\mathcal{E}(U,E),\mathscr{U},\mathcal{C},\{O\}^o;(\mathfrak{E},\mathscr{B})),\mathcal{I}_{\alpha i}),\\(\mathscr{F}_{\Gamma}&(\mathscr{U},\mathcal{C};(\mathfrak{E},\mathscr{B})),\mathcal{I}_{\alpha i})=\underrightarrow{\lim}_{C\in \mathcal{C}} (\mathscr{F}_{\Gamma}(\mathscr{U},\{F\in \mathcal{F}, F\subset C\};(\mathfrak{E},\mathscr{B})),\mathcal{I}_{\alpha})
 \\&\to (\mathscr{F}_{\Gamma}(\mathscr{U},\mathcal{C};(\mathfrak{E},\mathscr{B})),\mathcal{I}_{\alpha p}):=\underleftarrow{\lim}_{O\in \mathcal{C}^o}(\mathscr{F}_{\Gamma}(\mathscr{U},\{O\}^o;(\mathfrak{E},\mathscr{B})),\mathcal{I}_{\alpha i}).\end{align*}

The spaces above are related in using proposition \ref{InjectionProjectiveProduct} and its consequence explained in proposition \ref{multilinear}. We recover in special cases statements of our previous examples.

\begin{ex}\textbf{Wave-Front variant of operadically controlled functionals} If for $\Gamma$ as  above, one defines $\lambda(\Gamma)$ by $\lambda_{n,i}(\Gamma)=((-\gamma_i)^{\dot{\times}(n-1)}\dot{\times}(\gamma_i')^c)^c; i\in J.$ Note that we easily check assumption $1$ namely  $\lambda_{n,i}(\Gamma)\dot{\times}\lambda_{N,i}(\Gamma)\subset(\lambda_{n,i}(\Gamma)^c\dot{\times}\lambda_{N,i}(\Gamma)^c)^c\subset \lambda_{n+N,i}(\Gamma)$ as soon as $-\gamma_i\subset(\gamma_i')^c,$ namely under assumption 2. 
One gets continuous embeddings for $\alpha\in\{si,ci,sp,cp\}$:
\begin{align*}(\mathscr{F}_{\lambda(\Gamma)}&(\mathcal{E}(U,E),\mathscr{U},\mathcal{K},\mathcal{C};(\mathfrak{E},\mathscr{B})),\mathcal{I}_{\alpha })\hookrightarrow(\mathscr{F}_{\Gamma}(\mathscr{U},\mathcal{C};(\mathfrak{E},\mathscr{B})),\mathcal{I}_{\alpha }),\end{align*}
where we used proposition \ref{multilinear} and \cite[Prop 1 p 20]{Schwarz} for $\epsilon$ product of injective continuous maps.

\end{ex}

Our first main general result generalizing and improving the known  properties of microcausal functionals is the following :

\begin{theorem}\label{AppliedFA}
The spaces $\mathscr{F}_{\lambda}(\mathcal{E}(U,E),\mathscr{U},\mathcal{K},\mathcal{C};(\mathfrak{E},\mathscr{B}))
,\mathscr{F}_{\Gamma}(\mathscr{U},\mathcal{C};(\mathfrak{E},\mathscr{B}))$ with topologies $\mathcal{I}_{sp}$ or $\mathcal{I}_{cp}$ (or $\mathcal{I}_{si}$ or $\mathcal{I}_{ci}$ if $\mathcal{C}$ is countably generated) are complete locally convex spaces (resp. nuclear) as soon as $\mathfrak{E}$ is. Assume moreover $\mathfrak{E}$ is a quasi-complete locally convex algebra with continuous product and assume assumption 1 (resp. assumption 2), then $\mathscr{F}_{\lambda}(\mathcal{E}(U,E),\mathscr{U},\mathcal{K},\mathcal{C};(\mathfrak{E},\mathscr{B}))$ is an algebra with hypocontinuous multiplication defined pointwise for $\mathcal{I}_{si}$ or $\mathcal{I}_{ci}$ (resp. $\mathscr{F}_{\Gamma}(\mathscr{U},\mathcal{C};(\mathfrak{E},\mathscr{B}))$  an algebra with continuous multiplication if $\mathcal{C}=\mathcal{F}$ for $\mathcal{I}_{s}$ or $\mathcal{I}_{c}$ and otherwise hypocontinuous for $\mathcal{I}_{si}$ or $\mathcal{I}_{ci}$).
\end{theorem}
The reader will note in the proof that, in order to use \cite{Schwartz3}, it is crucial that in all cases $\mathfrak{E}$ has a continuous and not only hypocontinuous multiplication map.

\begin{proof}

 \setcounter{Step}{0} 
 \begin{step}
Nuclearity.
\end{step}
By definition, In the case $\mathcal{C}'=\mathcal{F}$, $\mathcal{I}_s,\mathcal{I}_{c}$ are locally convex kernels of nuclear spaces $\mathcal{E}[\R^n,(\mathcal{D}'_{\lambda_{k,i},\overline{\lambda_{k,i}}}(U^k,(\mathcal{C}^k)^{oo};(E^{\o n})^*),\mathcal{I}_{ibi})\epsilon \mathfrak{E}],$ (including $\mathcal{E}[\R^n, \mathfrak{E}],$ in case $k=0$) 

\noindent $\mathcal{E}(\R^n,L_{\beta e}[(\mathcal{D}'_{\gamma_i,\overline{\gamma_i}}(U;E),\mathcal{I}_{ibi})^{n-1};(\mathcal{E}'_{\gamma_i',\overline{\gamma_i'}}(U;E^*),\mathcal{I}_{ibi})]\epsilon \mathfrak{E})$. From \cite{Schwarz} and our identification of these spaces with $\epsilon$ products, the completion of these spaces are the same with $\mathfrak{E}$ replaced by its completion which is still nuclear. Then $\epsilon$-products are completed projective products of nuclear spaces thus nuclear.
Then from projective limits, subspaces and countable projective limits, one deduces all the support cases.

 \begin{step}
Completeness.
\end{step}
In the case $\mathcal{C}'=\mathcal{F}$ all topologies are defined as locally convex kernels. 
As made explicit below,
it remains to check those locally convex kernels are closed subspaces of corresponding products (then they will be complete as closed subspaces of complete locally convex spaces, see \cite[Prop 5.4]{Treves}, since continuous and smooth functions considered valued in complete locally convex spaces are complete, see Tr\`eves \cite[Prop 44.1]{Treves}).  Said otherwise the proof is the same as for the classical smooth structure in convenient analysis (cf \cite[Def 3.11]{KrieglMichor}).

\begin{step}
(Hypo)continuity.
\end{step}
We only consider the cases with support $\mathcal{F}$, the general case follows by regular (since countable strict) inductive limits and \cite{Melnikov}.

Using composition, product \cite[corol 3.13]{KrieglMichor} and smoothness of bounded multilinear maps \cite[lemma 5.5]{KrieglMichor}, one gets pointwise product of smooth maps are smooth and the derivative given by ordinary Leibniz rule. 

From Schwartz' crossing theorems \cite[Prop 2 p 18]{Schwartz3}, and replacing by completeness and nuclearity an $\epsilon$ product by a quasi-completed projective tensor product (here any $\hat{\otimes}$ is such a quasi-completion), on gets a canonical map  :

\begin{align*}\Gamma_{\beta,\pi}&:[(\mathcal{D}'_{\lambda_{n,j},\overline{\lambda_{n,j}}}(U^n,(\mathcal{C}^n)^{oo};(E^{\otimes n})^*),\mathcal{I}_{ibi})\hat{\o}_\pi \mathfrak{E}]\otimes_{\beta} (\mathcal{D}'_{\lambda_{m,j},\overline{\lambda_{m,j}}}(U^m,(\mathcal{C}^m)^{oo};(E^{\otimes m})^*),\mathcal{I}_{ibi})\epsilon \mathfrak{E})\\&\to [\mathcal{D}'_{\lambda_{n,j},\overline{\lambda_{n,j}}}(U^n,(\mathcal{C}^n)^{oo};(E^{\otimes n})^*),\mathcal{I}_{ibi})\hat{\o}_\beta(\mathcal{D}'_{\lambda_{m,j},\overline{\lambda_{m,j}}}(U^m,(\mathcal{C}^m)^{oo};(E^{\otimes m})^*),\mathcal{I}_{ibi})]\epsilon [\mathfrak{E}\hat{\otimes}_{\pi}\mathfrak{E}]
\end{align*}

Note that the extension to the hypocontinuous product follows from $1,2$  in his result because $sup(\beta,\pi)=\beta$ and because any bounded set in  $(\mathcal{D}'_{\lambda_{n,j},\overline{\lambda_{n,j}}}(U^n,(\mathcal{C}^n)^{oo};(E^{\otimes n})^*),\mathcal{I}_{ibi})\hat{\o}_\pi \mathfrak{E}$ is a $\beta-\pi$ decomposable (even  $\gamma-\pi$-decomposable) by \cite[Prop 1.4 p 16]{Schwartz3} by nuclearity of the strong=Arens dual of $(\mathcal{D}'_{\lambda_{n,j},\overline{\lambda_{n,j}}}(U^n,(\mathcal{C}^n)^{oo};(E^{\otimes n})^*),\mathcal{I}_{ibi})$. Now if $M$ is the multiplication map on $\mathfrak{E}$, we thus obtain a map $[\iota\circ (.\otimes .)\epsilon M]\circ \Gamma_{\beta,\pi}$ from the tensor multiplication map of \ref{hypocontinuity} and a canonical injection from assumption 1, with value in
\begin{align*}&(\mathcal{D}'_{\lambda_{n,j}\dot{\times}\lambda_{m,j},\overline{\lambda_{n,j}\dot{\times}\lambda_{m,j}}}(U^{n+m},(\mathcal{C}^{n+m})^{oo};(E^{\otimes (n+m)})^*),\mathcal{I}_{ibi})\epsilon \mathfrak{E}\\&\to (\mathcal{D}'_{\lambda_{n+m,j},\overline{\lambda_{n+m,j}}}(U^{n+m},(\mathcal{C}^{n+m})^{oo};(E^{\otimes (n+m)})^*),\mathcal{I}_{ibi})\epsilon \mathfrak{E}.\end{align*}
 This is the expected space of value.
 Using Leibniz rule, in order to get the derivative of $(FG)^{(n+m)}$ and reasoning  in composing the above map with $\phi\mapsto (F^{(n)}(\phi)\otimes G^{(m)}(\phi))$  one deduces the stated hypocontinuity in the first case.
 
 In the second case, one reasons similarly with $\Gamma_{\pi,\pi}$ which is continuous instead \cite[Prop 2.3 p 18-19]{Schwartz3}, the tensor multiplication map of lemma \ref{MultilinearTensor} and the canonical injection one needs to use with assumption 2 is the composition with $(\mathcal{D}'_{\gamma_i,\overline{\gamma_i}}(U;E),\mathcal{I}_{ibi})\to (\mathcal{D}'_{(-\gamma_i')^c,\overline{(-\gamma_i')^c}}(U;E),\mathcal{I}_{ibi}).$ 



\end{proof}
\section{Supplementary algebraic structure on spaces of smooth functionals}
Our final task is to get a general result to recover Poisson brackets or retarded brackets 
on our spaces of functionals
. We especially want to cover field dependent variants as in \cite{RibeiroBF}. The reader not familiar with vector valued distributions should assume $\mathfrak{E}=\C$ later at first reading.

If $\mathfrak{E}$ is a quasi-complete algebra with continuous multiplication $M$, then one can extend duality pairings to an hypocontinuous $\mathfrak{E}$-valued pairing defining for
$f\in (\mathcal{E}'_{\gamma_i',\overline{\gamma_i'}}(U;E),\mathcal{I}_{ibi})\epsilon \mathfrak{E}$, $g\in (\mathcal{D}'_{(-\gamma_i')^c,\overline{-(\gamma_i')^c}}(U;E),\mathcal{I}_{ibi})\epsilon \mathfrak{E}$ a pairing  $\langle f, g\rangle \in \mathfrak{E}$. This comes from \cite{Schwartz3} (using $(\langle.,.\rangle \epsilon M)\circ \Gamma_{\beta,\pi}$) as in the proof of our previous theorem.

Assume moreover given an action $A:\mathfrak{F}\hat{\otimes}_{\pi}\mathfrak{E}\to \mathfrak{E}$ with $\mathfrak{F}$ a quasi-complete locally convex space.

One can  extend $\circ_j$ from lemma \ref{MultilinearComposition} to :
\begin{align*}\circ_j:& L
({\bigotimes}_{\beta e, i\in [1,m]}
 (\mathcal{D}'_{\gamma'_i,\overline{\gamma'_i}}(U_i',\mathcal{C}'_i;E_i'),\mathcal{I}_{ibi})
; (\mathcal{D}'_{\gamma',\overline{\gamma'}}(U',\mathcal{C}';E'),\mathcal{I}_{ibi}))\epsilon \mathfrak{F} \times \\&  L({\bigotimes}_{\beta e, i\in [1,n]} (\mathcal{D}'_{\gamma_i,\overline{\gamma_i}}(U_i,\mathcal{C}_i;E_i),\mathcal{I}_{ibi}); (\mathcal{D}'_{\gamma,\overline{\gamma}}(U,\mathcal{C};E)),\mathcal{I}_{ibi}))\epsilon \mathfrak{E} \\& \ \ \ \ \to  L({\bigotimes}_{\beta e, i\in [1,n+m-1]} (\mathcal{D}'_{\gamma''_i,\overline{\gamma''_i}}(U_i'',\mathcal{C}''_i;E''_i),\mathcal{I}_{ibi}); (\mathcal{D}'_{\gamma',\overline{\gamma'}}(U',\mathcal{C}';E'),\mathcal{I}_{ibi}))\epsilon \mathfrak{E}
\end{align*}
using again $\Gamma_{\beta,\pi}$ of \cite[Prop 2 p 18]{Schwartz3}  to obtain the above map as $([\circ_j]\epsilon A)\circ \Gamma_{\beta,\pi}.$ One proves that this map is  hypocontinuous  by composition and using $\Gamma_{\beta,\pi}$ hypocontinuous as proved before.

Note also that $L_{\beta e}[(\mathcal{E}'_{\gamma_i',\overline{\gamma_i'}}(U;E^*),\mathcal{I}_{ibi});(\mathcal{D}'_{\gamma_i,\overline{\gamma_i}}(U;E),\mathcal{I}_{ibi})]\epsilon \mathfrak{F}\subset L_{\beta e}[\mathcal{E}(U;E^*);\mathcal{D}'(U;E)]\epsilon \mathfrak{F}$ and one can define \begin{align*}\mathfrak{F}(\Gamma)=[&\cap_{i\in J}L_{\beta e}[(\mathcal{E}'_{\gamma_i',\overline{\gamma_i'}}(U;E^*),\mathcal{I}_{ibi});(\mathcal{D}'_{\gamma_i,\overline{\gamma_i}}(U;E),\mathcal{I}_{ibi})]\epsilon \mathfrak{F}]\\&\cap[\cap_{i\in J}L_{\beta e}[(\mathcal{E}'_{(-\gamma_i^c),\overline{(-\gamma_i^c)}}(U;E^*),\mathcal{I}_{ibi});(\mathcal{D}'_{(-\gamma_i')^c,\overline{(-\gamma_i')^c}}(U;E),\mathcal{I}_{ibi})]\epsilon \mathfrak{F}]\\&\cap[\cap_{i\in J}L_{\beta e}[(\mathcal{E}'_{\gamma_i',\overline{\gamma_i'}}(U;E^*),\mathcal{I}_{ibi});(\mathcal{D}'_{(-\gamma_i')^c,\overline{-(\gamma_i')^c}}(U;E),\mathcal{I}_{ibi})]\epsilon \mathfrak{F}]\end{align*}
 in this space with the projective kernel topology.
We can associate continuous maps for $j\in J$ \begin{align*}&\mathfrak{F}(\Gamma)\to \mathfrak{F}_{(j,1)}(\Gamma):=[L_{\beta e}[(\mathcal{E}'_{\gamma_i',\overline{\gamma_i'}}(U;E^*));(\mathcal{D}'_{\gamma_i,\overline{\gamma_i}}(U;E))]\epsilon \mathfrak{F}]\\&\ \ \ \ \ \ \ \ \ \ \ \ \ \ \ \ \ \ \ \ \ \ \ \ \cap[L_{\beta e}[(\mathcal{E}'_{(-\gamma_i^c),\overline{(-\gamma_i^c)}}(U;E^*));(\mathcal{D}'_{(-\gamma_i')^c,\overline{(-\gamma_i')^c}}(U;E))]\epsilon \mathfrak{F}],
\\&\mathfrak{F}(\Gamma)\to \mathfrak{F}_{(j,2)}(\Gamma):=L_{\beta e}[(\mathcal{E}'_{\gamma_i',\overline{\gamma_i'}}(U;E^*),\mathcal{I}_{ibi});(\mathcal{D}'_{(-\gamma_i')^c,\overline{-(\gamma_i')^c}}(U;E),\mathcal{I}_{ibi})]\epsilon \mathfrak{F}.\end{align*}
 
  We reach the point where we need to introduce a bornology $\mathscr{B}(\Gamma,\mathscr{B},\mathcal{C})$ on a subspace of $\mathfrak{F}$ and depending on $(\mathfrak{E},\mathscr{B})$. For sets $A\subset\mathcal{E}'_{\gamma_i',\overline{\gamma_i'}}(U;E^*)\epsilon \mathfrak{E}, A'\subset\mathcal{E}'_{\gamma_i',\overline{\gamma_i'}}(U;E^*)\epsilon \mathfrak{E}$ uniformly supported in $\mathcal{C}$ for $\mathscr{B}$, namely e.g. for any $B\in \mathscr{B},$ $\overline{\bigcup_{a\in A,b\in B}\text{supp}(a(b))}\in \mathcal{C}.$
Then the bornology $\mathscr{B}(\Gamma,\mathscr{B},\mathcal{C})$ is by definition generated by sets C of form $\{F\mapsto \psi(\langle F\circ_1(a),a'\rangle), a\in A,a'\in A', \psi\in B\}$ for $A,A'$ as above and $B\in \mathscr{B}.$ Those bornologies have only been introduced to state a natural condition for preservation of support conditions.

\begin{theorem}\label{AppliedComposition}
Assume $\mathfrak{E},\mathfrak{F}$ quasi-complete with continuous multiplications $\mathfrak{F}\hat{\otimes}_{\pi}\mathfrak{E}\to \mathfrak{E},\mathfrak{E}\hat{\otimes}_{\pi}\mathfrak{E}\to \mathfrak{E}$.
Let $\Gamma'=\{(\gamma_i,-\gamma_i^c),(\gamma_i,\gamma_i')\}_{i\in J}$ indexed by $J\times \{1,2\}$ and 
$$R\in {F}_{\Gamma'}(\mathscr{U},\mathcal{C};(\mathfrak{F}(\Gamma),(\mathfrak{F}_g(\Gamma))_{g\in J\times\{1,2\}},\mathscr{B}(\Gamma,\mathscr{B}_f,\mathcal{C}))).$$ 

Then the product defined by $$[R(F,G)](\varphi)=\langle R(\varphi)[F^{(1)}(\varphi)],G^{(1)}(\varphi)\rangle$$
is an hypocontinuous product on $\mathscr{F}_{\Gamma}(\mathscr{U},\mathcal{C};(\mathfrak{E},\mathscr{B}_f))$ 
 with topologies $\mathcal{I}_{si}$ or $\mathcal{I}_{ci}$.

\end{theorem}
We of course wrote $R(\varphi)[F^{(1)}(\varphi)]$ for $R(\varphi)\circ_1[F^{(1)}(\varphi)].$ Note we could formulate an obvious variant for $\mathscr{F}_{\lambda(\Gamma)}(\mathcal{E}(U,E),\mathscr{U},\mathcal{K},\mathcal{C};(\mathfrak{E},\mathscr{B}_f))$ (but notationally less convenient since $\epsilon$ products don't recover wave front set conditions and we could not formulate the assumption on $R$ as a vector valued statement). 
\begin{proof}
The smoothness of $R(F,G)$ comes from composition of bounded(= smooth) multilinear maps and smooth maps. By Leibniz rule, the derivatives are expressed using $\langle R^{(l)}(\varphi)[F^{(m+1)}(\varphi)],G^{(n+1)}(\varphi)\rangle.$ Let us explain the meaning of this expression in more detail.
First note that
\begin{align*}&R^{(l)}(\varphi)\in L_{\beta e}[(\mathcal{D}'_{\gamma_i,\overline{\gamma_i}}(U;E),\mathcal{I}_{ibi})^{l-1};(\mathcal{E}'_{(-\gamma_i)^c,\overline{(-\gamma_i)^c}}(U;E^*),\mathcal{I}_{ibi})]\epsilon \mathfrak{F}_{(i,1)}(\Gamma)\hookrightarrow\\& L_{\beta e}[(\mathcal{D}'_{\gamma_i,\overline{\gamma_i}}(U;E),\mathcal{I}_{ibi})^{l-1};(\mathcal{E}'_{(-\gamma_i)^c,\overline{(-\gamma_i)^c}}(U;E^*),\mathcal{I}_{ibi})]\epsilon[L_{\beta e}[(\mathcal{E}'_{\gamma_i',\overline{\gamma_i'}}(U;E^*),\mathcal{I}_{ibi});(\mathcal{D}'_{\gamma_i,\overline{\gamma_i}}(U;E),\mathcal{I}_{ibi})]\epsilon \mathfrak{F}]\\&\simeq L_{\beta e}[(\mathcal{D}'_{\gamma_i,\overline{\gamma_i}}(U;E),\mathcal{I}_{ibi})^{l}\times(\mathcal{E}'_{\gamma_i',\overline{\gamma_i'}}(U;E^*),\mathcal{I}_{ibi});(\mathcal{D}'_{\gamma_i,\overline{\gamma_i}}(U;E),\mathcal{I}_{ibi})]\epsilon \mathfrak{F}
\end{align*}
Thus from the extended map before the proof $R^{(l)}(\varphi)[F^{(m+1)}(\varphi)]=R^{(l)}(\varphi)\circ_{l+1}[F^{(m+1)}(\varphi)]\in L_{\beta e}[(\mathcal{D}'_{\gamma_i,\overline{\gamma_i}}(U;E),\mathcal{I}_{ibi})^{l+m};(\mathcal{D}'_{\gamma_i,\overline{\gamma_i}}(U;E), \mathcal{I}_{ibi})]\epsilon \mathfrak{E}$ and from the way it is composed from smooth and bounded multilinear maps, it is again (conveniently) smooth in $\varphi.$

Since (using a partial transpose which gives the same map by symmetry of derivatives) we have  $G^{(n+1)}(\varphi)\in  L_{\beta e}[(\mathcal{D}'_{\gamma_i,\overline{\gamma_i}}(U;E),\mathcal{I}_{ibi})^{n-1}\times(\mathcal{D}'_{(-\gamma_i')^c,\overline{-(\gamma_i')^c}}(U;E), \mathcal{I}_{ibi});(\mathcal{E}'_{(-\gamma_i)^c,\overline{(-\gamma_i)^c}}(U;E^*) )]\epsilon \mathfrak{E}.$

This shows one can also pair $\langle R^{(l)}(\varphi)[F^{(m+1)}(\varphi)],G^{(n+1)}(\varphi)\rangle$ with the $\mathfrak{E}$-valued duality pairing to get a smooth map with value in $L_{\beta e}[(\mathcal{D}'_{\gamma_i,\overline{\gamma_i}}(U;E),\mathcal{I}_{ibi})^{l+m+n-1};(\mathcal{E}'_{\gamma_i',\overline{\gamma_i'}}(U;E^*),\mathcal{I}_{ibi})]\epsilon \mathfrak{E}$ if the last variable (evaluated by duality) is a variable below $G^{(n+1)}(\varphi)$ (one uses multiple times associativity of $\epsilon$-products to do this).

If the last variable is below $F$, we have a partial transpose on the term with $F^{(m+1)}(\varphi)$ as above for $G$, and uses then a $$R^{(l)}(\varphi)\in  L_{\beta e}[(\mathcal{D}'_{\gamma_i,\overline{\gamma_i}}(U;E),\mathcal{I}_{ibi})^{l}\times(\mathcal{E}'_{(-\gamma_i^c),\overline{(-\gamma_i)^c}}(U;E^*),\mathcal{I}_{ibi});(\mathcal{D}'_{(-\gamma_i')^c,-\overline{(\gamma_i')^c}}(U;E),\mathcal{I}_{ibi})]\epsilon \mathfrak{F}$$ and the reasoning is similar. Finally, if the last variable is below $R$, one finally uses $\mathfrak{F}_{(i,2)}(\Gamma)$ to get :\begin{align*}R^{(l)}(\varphi)\in&  L_{\beta e}[(\mathcal{D}'_{\gamma_i,\overline{\gamma_i}}(U;E),\mathcal{I}_{ibi})^{l}\\&\times\mathcal{D}'_{(-\gamma_i')^c,\overline{(-\gamma_i')^c}}(U;E),\mathcal{I}_{ibi})^{l}\times(\mathcal{E}'_{\gamma_i',\overline{\gamma_i'}}(U;E^*),\mathcal{I}_{ibi});(\mathcal{D}'_{(-\gamma_i')^c,-\overline{(\gamma_i')^c}}(U;E),\mathcal{I}_{ibi})]\epsilon \mathfrak{F}.\end{align*} From all  hypocontinuities of our building maps, it is then obvious to get the stated hypocontinuity in the full support case. It remains to check the support condition for the  distributions given for $\psi\in B\subset\mathscr{B}$ \begin{align*}h&\mapsto \psi([R(F,G)]^{(1)}(\varphi)[h])\\&=\psi(\langle R^{(1)}(\varphi)[h,F^{(1)}(\varphi)],G^{(1)}(\varphi)\rangle) +\psi(\langle R(\varphi)[F^{(2)}(\varphi)[h]],G^{(1)}(\varphi)\rangle) + \psi(\langle R(\varphi)[F^{(1)}(\varphi)],G^{(2)}(\varphi)[h]\rangle)\end{align*}
The two last terms are supported respectively in $\text{supp}(F,\mathfrak{E}'),$ $\text{supp}(G,\mathfrak{E}')$. Indeed let $h\in \mathcal{D}(U)$ with $\text{supp}(h)\cap \text{supp}(F,\mathfrak{E}')=\emptyset$, and any $g\in \mathcal{D}(U), \psi\in \mathfrak{E}'$  then $\psi(F^{(1)}(\varphi)[h])=0$ thus if we differentiate in $\phi$ $\langle F^{(2)}(\varphi)[h], g\otimes \psi\rangle=\psi(F^{(2)}(\varphi)[h,g])=0$ where $F^{(2)}(\varphi)[h]$ is seen in $(\mathcal{E}'_{(-\gamma_i)^c,\overline{(-\gamma_i)^c}}(U;E^*) )]\epsilon \mathfrak{E}$ and thus by definition and density  of $\mathcal{D}(U)\subset (\mathcal{E}'_{(-\gamma_i)^c,\overline{(-\gamma_i)^c}}(U;E^*) )'_c$ one gets 
$F^{(2)}(\varphi)[h]=0$ in this space and thus $\psi(\langle R(\varphi)[F^{(2)}(\varphi)[h]],G^{(1)}(\varphi)\rangle)=0.$

It remains to consider the first term. 

Let $C=\{A\mapsto \psi(\langle A\circ_1[F^{(1)}(\varphi_1)],G^{(1)}(\varphi_2)\rangle), \varphi_i\in \mathscr{U},\psi \in B\}$ then $C\in \mathscr{B}(\Gamma,\mathscr{B},\mathcal{C})$ by definition since $\{F^{(1)}(\varphi_1), \varphi_i\in \mathscr{U}\}\subset (\mathcal{E}'_{(-\gamma_i)^c,\overline{(-\gamma_i)^c}}(U;E^*) )\epsilon \mathfrak{E}$ and is uniformly supported uniformly in elements of $\mathscr{B}$ in $\mathcal{C}.$

Thus for $\psi'\in C$, $h\in \mathcal{D}(U),$ $\text{supp}(h)\cap \text{supp}(R,C)=\emptyset$ we have $\psi'(R^{(1)}(\varphi)[h])=[\psi'(R(\varphi))]^{(1)}[h]=0$ and especially $\psi(\langle R^{(1)}(\varphi)[h,F^{(1)}(\varphi)],G^{(1)}(\varphi)\rangle=0$. As a conclusion , we have $$\text{supp}(R(F,G),B)=\overline{\bigcup_{\varphi\in \mathscr{U},\psi\in B}\text{supp}([R(F,G)]^{(1)}(\varphi)[\psi])}\subset C\cup \text{supp}(F,\mathfrak{E}')\cup\text{supp}(G,\mathfrak{E}')\in \mathcal{C},$$
 since $\mathcal{C}$ is polar thus stable by finite unions and this concludes to the expected support condition. This also shows compatibility with the strict inductive limits and thus by \cite{Melnikov}, hypocontinuity needs only be checked in the case $\mathcal{C}=\mathcal{F}$ and this is obvious by our formulas for derivative in terms of our various hypocontinuous products. 
\end{proof}

\section{Application to Retarded and advanced products }\label{Retarded}As we will explain in the next remark, the construction of retarded products explained in \cite{RibeiroBF} fails, as stated there, on microcausal functionals. Let us apply our result in the case \textrm{$\mathfrak{E}= \mathfrak{F}=\C$} for simplicity. We mostly follow their physical setting.

We stick to the setting of \cite{BarGinouxPfaffe} (but with reversed convention on what is called retarded/advanced to follow \cite{RibeiroBF}) to which we refer for terminology, for the first result which gathers in our functional analytic language known results similar to \cite[Corol 3.2.5]{RibeiroBF}. Especially, $E^*$ is now obtained using a volume form for any (fixed) Riemanian metric on $M$ (cf. p 167) and a vector bundle isomorphism $*:E^*\to E$ for some scalar product identifying sections of $E^*$ and $E$.

We call $L(g)=[\overline{V_+(g)}\cap (V_+(g))^c] \cup[\overline{V_-(g)}\cap (V_-(g))^c]\subset \dot{T}^*M$  the light cone bundle  (with the notation of example \ref{Opmicrocausal}).

\begin{proposition}\label{ExtendedRetarded}
Let $(M,g)$ a globally hyperbolic Lorentzian manifold, $P$ a normally hyperbolic operator acting on sections of a vector bundle $E\to M$ (especially its principal symbol is related as usual to $\hat{g}\o Id_E$ for $\hat{g}$ globally hyperbolic). We assume $ V_{\mp}(\hat{g})\subset V_{\mp}(g).$ Then the retarded Green's operator $\Delta_{P}^{ret}$ (resp. the advanced Green's operator $\Delta_{P}^{adv}$) from $\mathcal{D}(M,E)\to \mathcal{E}(M,E)$ with $\text{supp}(\Delta_{P}^{ret}(v))\subset J^{+,\hat{g}}(\text{supp}(v))\subset J^{+,{g}}(\text{supp}(v))$ (resp. $\text{supp}(\Delta_{P}^{adv}(v))\subset J^{-,\hat{g}}(\text{supp}(v))$) extends continuously to $\mathcal{D}'(M,\mathcal{K}_P(g),E)\to \mathcal{D}'(M,\mathcal{K}_P(g);E)$ (resp. $\mathcal{D}'(M,\mathcal{K}_F(g),E)\to \mathcal{D}'(M,\mathcal{K}_F(g);E)$ with the notation of \cite[Ex 16]{Dab14a}
. 
Moreover  for any $\gamma\subset L(\hat{g})^c$ or $\gamma\supset L(\hat{g})$, any $\gamma\subset\Lambda\subset \overline{\gamma},$ either equal to $\overline{\gamma}$ or satisfying the same condition as $\gamma$, any $\mathcal{C}_P\in\{\mathcal{K}_P(g),\mathcal{SK}_P(g)\},\mathcal{C}_F\in\{\mathcal{K}_F(g),\mathcal{SK}_F(g)\}$ and finally any  $\mathcal{I}$ among $\mathcal{I}_{ppp},\mathcal{I}_{ibi}$ it can also be extended to :
$$\Delta_{P}^{ret}:(\mathcal{D}_{\gamma,\Lambda}'(M,\mathcal{C}_P,E),\mathcal{I})\to (\mathcal{D}_{\gamma,\Lambda}'(M,\mathcal{C}_P;E),\mathcal{I}),$$
$$\Delta_{P}^{adv}:(\mathcal{D}_{\gamma,\Lambda}'(M,\mathcal{C}_F,E),\mathcal{I})\to (\mathcal{D}_{\gamma,\Lambda}'(M,\mathcal{C}_F;E),\mathcal{I}).$$
Finally, assume a family of $P$ as above is continuously parametrized by a compact set as seen in $L(\mathcal{D}(M,\mathcal{K};E),\mathcal{D}(M,\mathcal{K};E))$ with the weak topology and gives bounded set of $$\Delta_{P}^{ret}\in L_{b}(\mathcal{D}(M,E),\mathcal{E}(M,E)).$$ If moreover the inclusion of light cones are uniform for the symbols of $P$ in the family then the corresponding extensions form also bounded sets in the spaces they are stated to live in, for the topology of convergence on bounded sets.
\end{proposition}
Note that our statement implies that the if part of \cite[Corol 3.2.5 (c)]{RibeiroBF} is wrong in taking $\gamma=V_+$ there may be points in the wave front set in $\overline{V_+}$ (with $DWF$ in $V_+$) without any point in the light cone in $WF(\Delta_P^{ret}(v))$ outside of $\overline{V_+}.$

\begin{proof}
From \cite[Prop 3.4.8 p 91]{BarGinouxPfaffe}, we know sequential continuity $\Delta_{P}^{ret}: \mathcal{D}(M,E)\to \mathcal{D}(M,\mathscr{O}_{\mathcal{SK}_P(\hat{g})},E)\to \mathcal{D}(M,\mathscr{O}_{\mathcal{SK}_P({g})},E)$  (resp. $\Delta_{P}^{adv}: \mathcal{D}(M,E)\to \mathcal{D}(M,\mathscr{O}_{\mathcal{SK}_F(g)},E)$, as stated we only know with value in $\mathcal{SK}(\hat{g})$ but the defining support condition implies this automatically). Note that in a bounded family case one gets something bounded here since the support control by $g$ gives a uniform control on support. Since the spaces involved are (quasi)-complete separated locally convex space whose strong dual is a Schwartz space using \cite[Prop 8]{Dab14a}
, an application of \cite[lemma 21]{BrouderDabrowski} implies they are Mackey-sequentially continuous, thus bounded, thus continuous since the spaces are (ultra)bornological. Thus we have well defined continuous adjoint maps (on Mackey=strong duals) $(\Delta_{P}^{ret})^*: \mathcal{D}'(M,\mathcal{K}_F({g}),E^*)\to\mathcal{D'}(M,E^*)$ (resp. $(\Delta_{P}^{adv})^*: \mathcal{D}'(M,\mathcal{K}_P({g}),E^*)\to\mathcal{D'}(M,E^*)$).
But from \cite[lemma 3.4.4 p 89]{BarGinouxPfaffe}, they correspond to extensions of $\Delta_{P^*}^{adv}$ (resp. $\Delta_{P^*}^{ret}$). Exchanging $P,P^*$ and using again the support condition on the dense space of smooth maps one deduces the first extension with modified target.
Moreover, since $J^{\pm,\hat{g}}(J^{\pm,g}(.))\subset J^{\pm,g}(.)$, evaluated at compacts or for $\Sigma$ Cauchy surface for  $g$ thus $\hat{g}$, $\mathcal{SK}_P(\hat{g})\subset \mathcal{SK}_P({g})$ and thus by duality  $\mathcal{K}_P({g})\subset \mathcal{K}_P(\hat{g})$ and both support conditions are stable by  $\Delta_{P}^{ret}$. Thus one gets the stated spaces of value of the extension. The bounded family yields bounded family of extensions as above.

 From propagation of singularity theorems for wave front set of solutions (cf. e.g.  \cite{Duistermaat} or \cite[Th 13,15]{stromaier}) and the statement on support, it is well known that the spaces claimed to be left stable by $\Delta_{P}^{ret/adv}$ are indeed stable in case $\gamma=\Lambda.$ It remains to see continuity, first for $\mathcal{I}_{ppp}$. 
Since open cones containing $\gamma$ can follow the same constraint (since $L(\hat{g})$ closed) we are reduced using \eqref{Alternativeppp} to the case $\gamma$ open (using also the completion to go to $\Lambda=\overline{\gamma}.)$  But in the open case $\mathcal{I}_{pmp}=\mathcal{I}_{iii}$ is both ultrabornological by  \cite[Th 23]{Dab14a}
 and quasi-LB thus strictly webbed by Theorem \ref{FAGeneral2} (since either the support condition or its dual is countably generated). Applying De Wilde's closed graph theorem \cite[\S 35.2.(2)]{Kothe2}, it suffices to check our map is sequentially closed. But from the continuity at level $\mathcal{D}'$ and the continuous injection, this is obvious. Going to 
$\Lambda=\overline{\gamma}$ by completion, it only remains to check a known wave front set condition to get general $\Lambda$ (since the topology in this case is induced). The case $\mathcal{I}_{ibi}$ is a consequence by bornologification.

For the boundedness statement, it only remains to consider the extensions with wave front set conditions. We start again from the case $\mathcal{I}_{ppp}$ for $\gamma=\Lambda\supset L(\hat{g})$ (for all symbols in the family of $P$) open. By barrelledness and a standard result \cite[\S 39.3.(2)]{Kothe2}, it suffices to see the family is pointwise bounded (and we even get an equicontinuous family not only a bounded one in this case). Thus take $u$ with $WF(u)=\Gamma\subset \gamma$, we want to prove boundedness of $\{\Delta_{P}^{ret}(u)\}\subset \mathcal{D}_{\Gamma,\Gamma}'(M,\mathcal{F};E)$ (since from the known support conditions and the inductive limit definition in the $\gamma$  open case this will be enough). But our family of $P$ are continuously parametrized on a compact in the sense stated above which implies that  described in coordinates, so are the coefficients of the differential operators and all their derivatives in positions. But now, a cone in $\gamma^c$ has been assumed not to intersect any $char(P)$ in the family, thus from the continuity in parameters, we can apply the argument in \cite[Th 8.3.1]{Hormander}, get (8.3.5) uniformly on parameters the boundedness of coefficients of $R_j$'s also uniformly on our compact set of parameters, and thus also the conclusion, which gives exactly the expected boundedness $\mathcal{D}_{\Gamma,\Gamma}'(M,\mathcal{F};E)$ (the seminorms of $\mathcal{D}'$ have been treated before). 

We thus got $\{\Delta_{P}^{ret}\}\subset L_b(\mathcal{D}_{\gamma,\gamma}'(M,\mathcal{C}_P;E);\mathcal{D}_{\gamma,\gamma}'(M,\mathcal{C}_P;E))$ is bounded.
From the description of bounded sets in the completion in lemma \ref{quasiCompletion} one obtains the case $\Lambda=\overline{\gamma}$ by completion, then by the projective description \eqref{Alternativeppp} for any $\gamma$ in the case $\Lambda=\overline{\gamma}$ and then as above by induced topology for any $\Lambda.$

For the topology $\mathcal{I}_{ibi}$, a bornologification argument doesn't work here but we reason as follows. Consider first the case $\Lambda=\overline{\gamma}$ and look at adjoint maps $(\Delta_{P}^{ret})^*$ thus defined for ultrabonological (thus barrelled) topologies $\mathcal{I}_b$. The boundedness proved for the adjoints prove strong pointwise boundedness (and the strong topology is known to be $\mathcal{I}_b$ by \cite[Prop 33]{Dab14a}
). From \cite[\S 39.3.(2)]{Kothe2} as above the family $(\Delta_{P}^{ret})^*$ is thus equicontinuous thus bounded, thus bounded as above at the completion level and since in this case equicontinuous sets and bounded sets coincide, we deduce the adjoint are bounded for the convergence on bounded sets for $\mathcal{I}_{ibi}$ in the case $\Lambda=\overline{\gamma}.$ The general $\Lambda$ case follows since the topology is induced as before. This finishes all the cases in the case $\gamma\supset L(\hat{g}).$ By duality and exchanging advanced and retarded propagators, one gets the case $\gamma\subset L(\hat{g})^c$ in considering $\lambda=-\gamma^c\supset L(\hat{g}).$ The assumption on $P$ gives the boundedness in $\lambda$ of coefficients of the adjoint differential operators and their derivatives that is what we used for the support case, and since we already know the case of full wave front set, one gets $\Delta_{P}^{ret}\in L_{b}(\mathcal{E}'(M,E),\mathcal{D}'(M,E))$ bounded and thus by duality the assumed boundedness at the adjoint level. 
As in the proof before one gets boundedness for $\mathcal{I}_{ibi}$ of the adjoints, thus of $\Delta_{P}^{ret}$ and one goes from this to $\mathcal{I}_{ppp}$ using the previous was its bornologification.


\end{proof}

The next result is a variant adapted to our setting of \cite[Prop 3.2.8]{RibeiroBF}. Since the proof is not included there and would need, in our view, to go back to parameter dependence of $\Delta_{P}^{ret}$ along smooth curve, at least at the level of smooth functions that can in principal be obtained by using parameter dependent variants of all local geometric objects (exponential maps, parallel transport) used in \cite{BarGinouxPfaffe} and controlling global geometric data uniformly by given controlling globally hyperbolic metrics, we don't enter in these details and only assume a reasonable minimal result needed to make work functional analytic argument. Going much further in the construction of retarded Green's operator is clearly not the purpose of this article.

 Following them and \cite{DutschBF}, we call generalized Lagrangian a  map $\mathscr{L}:\mathcal{D}(U)\to \mathscr{F}_{\mu loc}( \mathscr{U},\C)$ (this space is defined in example \ref{microlocal}) with $\text{supp}(\mathscr{L}(f))\subset \text{supp}(f)$ and $\mathscr{L}$ additive in $f$ as defined in the cited example.

From the definition we know that $\mathscr{L}(f)^{(2)}(\varphi)\in \mathcal{D}'_{C_2,C_2}(M^2,\mathcal{K},(E^{\otimes 2})^*)$ and is smooth in $\varphi.$ Noting that $C_2\subset (-\gamma\dot{\times}\gamma^c)^c$ for any cone $\gamma\subset \dot{T}^*M$, and composing with the canonical map built in proposition \ref{multilinear} one gets an image $$P_{\mathscr{L}(f),\varphi}\in L((\mathcal{D}'_{\gamma,\overline{\gamma}}(M,\mathcal{C};E),\mathcal{I}_{ibi}),(\mathcal{D}'_{\gamma,\overline{\gamma}}(M,\mathcal{C};E^*),\mathcal{I}_{ibi})).$$
Noting that if $f,g\in \mathcal{D}(M)$ equal $1$ on a neighborhood of $\text{supp}(u)\in \mathcal{K}$, $P_{\mathscr{L}(f),\varphi}(u)=P_{\mathscr{L}(g),\varphi}(u)$ starting with $u$ smooth, since by \cite[lemma 3.1.2]{RibeiroBF} $\text{supp}(\mathscr{L}^{(1)}(f)-\mathscr{L}^{(1)}(g))\subset \text{supp}(f-g)$ so that just after one derivative evaluating it to $u$ make vanish the difference. 
As a consequence, take now $u\in \mathcal{D}'_{\gamma,\overline{\gamma}}(M,\mathcal{C};E), v\in [\mathcal{D}'_{\gamma,\overline{\gamma}}(M,\mathcal{C};E)]'$ thus supported respectively in $C\in \mathcal{C}, D\in (\mathscr{O}_\mathcal{C})^o$ with $C_\epsilon\cap D_\epsilon$ compact, and approximating by $u,v$ smooth one gets, assuming this time $f,g$ equal to $1$ on a neighborhood of $C_\epsilon\cap D_\epsilon$ $\langle P_{\mathscr{L}(f),\varphi}(u),v\rangle =\langle P_{\mathscr{L}(f),\varphi}(uh),v\rangle =\langle P_{\mathscr{L}(g),\varphi}(uh),v\rangle=\langle P_{\mathscr{L}(g),\varphi}(u),v\rangle,$ using $h$ equal to $1$ on a neighborhood of $\text{supp}(v)$ supported in $D_\epsilon$ and the fact that additivity of $\mathscr{L}(f),\mathscr{L}(g)$ implies its second derivative is supported on the diagonal \cite[prop 2.3.11]{RibeiroBF}.

Arguing with the inductive limit definition of the topologies (also on the dual) one can thus define $\langle P_{\mathscr{L},\varphi}(u),v\rangle=\langle P_{\mathscr{L}(f),\varphi}(u),v\rangle,$ $f\in \mathcal{D}(M)$ equal $1$ on a sufficiently huge neighborhood of $\text{supp}(u)\cap \text{supp}(v)$ and one gets for any cone $\gamma$  a continuous map (agreeing for different $\gamma$ and polar enlargeable families $\mathcal{C}$):
$$P_{\mathscr{L},\varphi}\in L((\mathcal{D}'_{\gamma,\overline{\gamma}}(M,\mathcal{C};E),\mathcal{I}_{ibi}),\mathcal{D}'_{\gamma,\overline{\gamma}}(M,\mathcal{C};E^*),\mathcal{I}_{ibi})).$$

\begin{proposition}\label{SmoothRetarded} Let $(M,g),(M,g')$ globally hyperbolic. Let $\mathscr{L}$ a generalized Lagrangian (on  $\mathscr{U}\subset \mathcal{E}(M,E)$ smooth sections of a bundle $E\to M$).
Assume given  a smooth curve on $\R$, $\lambda\mapsto \varphi_\lambda\in\mathscr{U}$ and the associated  $P:\lambda\mapsto P_\lambda:=P_{\mathscr{L},\varphi_\lambda}$ and assume $*P_\lambda$ are  normally hyperbolic  differentials operators on $E\to M$ with metric $\hat{g}_\lambda$ globally hyperbolic and with $${V_{\pm}(g')}\supset {V_{\pm}(\hat{g}_\lambda)}\supset  {V_{\pm}(g)}.$$ 
Then for any $\gamma\in\{V_{\pm}(g), V_{\pm}(g)^c\},$ we have the smoothmness of $$P\in C^\infty[\R;L_{\beta e}(\mathcal{D}'_{\gamma,\overline{\gamma}}(M,\mathcal{C};E);\mathcal{D}'_{\gamma,\overline{\gamma}}(M,\mathcal{C};E^*))]$$ Also assume $\Delta_{*P_\lambda}^{ret}\in L_b(\mathcal{D}(M,E);\mathcal{E}(M,E))$  is locally bounded in $\lambda$, then for any $\gamma$ as above $\mathcal{C}_P\in\{\mathcal{K}_P(g'),\mathcal{SK}_P(g')\}$, the extension $\Delta_{P_\lambda}^{ret}=(\Delta_{*P_\lambda}^{ret})*\in L_{\beta e}[(\mathcal{D}_{\gamma,\overline{\gamma}}(M,\mathcal{C}_P,E^*),\mathcal{I}_{ibi})\to (\mathcal{D}_{\gamma,\overline{\gamma}}'(M,\mathcal{C}_P;E),\mathcal{I}_{ibi})]$ of the previous proposition is smooth in $\lambda$ too with :$$\frac{\partial}{\partial \lambda} \Delta_{P_\lambda}^{ret}=
-\Delta_{P_\lambda}^{ret}\dot{P}_\lambda\Delta_{P_\lambda}^{ret}.$$

There are corresponding results for advanced propagators.
\end{proposition}
\begin{proof}
The smoothness of $P$ follows from the construction above and the definition of seminorms that only involves evaluating on sets uniformly supported in $\mathcal{C}$ and $(\mathscr{O}_\mathcal{C})^o.$ More precisely, one defines the candidate derivatives as above by inserting $f$ and uses a weak characterization as in \cite[corol 1.9 p 14]{KrieglMichor} but instead of evaluating boundedness on all elements of the dual, one evaluates it by the definition for the $\epsilon$ product on tensor product of equicontinuous sets for which the $f$ could be chosen uniformly. To get smoothness of $\Delta_{P_\lambda}^{ret}$ it suffices to prove the formula for the first derivative and apply induction using smoothness (boundedness, even hypocontinuity) of the composition map by proposition \ref{MultilinearComposition} between our spaces.

We first note that on compact sets of $\lambda$'s we have uniform boundedness in the stated spaces of $\Delta_{P_\lambda}^{ret}$ since $V_{\pm}(g)\subset L(\hat{g}_\lambda)^c$ for all $\lambda$ and the corresponding inclusion for complements that enable to use the boundedness statement in proposition \ref{ExtendedRetarded}.

To compute the first derivative, one also wants to check a resolvent equation : $$\Delta_{P_\lambda}^{ret}=\Delta_{P_\mu}^{ret}+\Delta_{P_\lambda}^{ret}(P_\mu-P_\lambda)\Delta_{P_\mu}^{ret}.$$

Indeed from the continuity of the operators used one can extend the defining relation $\Delta_{P_\lambda}^{ret}P_\lambda=Id$ to functions with support in $\mathcal{C}_P$
so that we have  $\Delta_{P_\mu}^{ret}=\Delta_{P_\lambda}^{ret}P_\lambda\Delta_{P_\mu}^{ret}$ and  also (the easier relation when applied to compact) $\Delta_{P_\lambda}^{ret}=\Delta_{P_\lambda}^{ret}P_\mu\Delta_{P_\mu}^{ret}$. Taking the difference, one gets the relation, first on compactly supported function, and then by continuity and density on distributions supported on $\mathcal{C}_P.$ 


Then from the equation, and continuity of composition, one checks $\Delta_{P_\lambda}^{ret}$ is continuous (for Mackey convergence since $\frac{1}{\lambda-\mu}(\Delta_{P_\mu}^{ret}-\Delta_{P_\lambda}^{ret})$ is bounded, using the uniform boundedness derived above from proposition \ref{ExtendedRetarded} and some boundedness from \cite[p 9]{KrieglMichor} since $P_\lambda$ smooth) and then from the usual computation of difference quotients, one gets the usual derivatives again by hypocontinuity of compositions and Mackey convergence argument using the proof \cite[Corol 1.9]{KrieglMichor} to get that after dividing by $\lambda-\mu$ the next expression is again bounded~: \begin{align*}\frac{1}{\lambda-\mu}&(\Delta_{P_\lambda}^{ret}-\Delta_{P_\mu}^{ret})+\Delta_{P_\mu}^{ret}\dot{P}_\mu\Delta_{P_\mu}^{ret}\\&=\Delta_{P_\mu}^{ret}[\frac{1}{\lambda-\mu}(P_\mu-P_\lambda)+\dot{P}_\mu]\Delta_{P_\mu}^{ret}+ (\Delta_{P_\lambda}^{ret}-\Delta_{P_\mu}^{ret})\frac{1}{\lambda-\mu}(P_\mu-P_\lambda)\Delta_{P_\mu}^{ret}.\end{align*}

\end{proof}
We are now ready to get retarded products from the results of our previous section. Some of the assumptions are probably redundant but we don't want to go back to the construction of retarded products in this paper to make them minimal, even though this is probably easy following \cite{BarGinouxPfaffe}.
\begin{proposition}\label{RetardedAlgebra}
Let $(M,g),(M,g')$ globally hyperbolic. Let $\mathscr{L}$ a generalized Lagrangian (on  $\mathscr{U}\subset \mathcal{E}(M,E)$ smooth sections of a bundle $E\to M$).
Assume that for any smooth curve  $\lambda\mapsto \varphi_\lambda\in\mathscr{U}$ and the associated   $*P_{\mathscr{L},\varphi_\lambda}$ are  normally hyperbolic  differential operators on $E\to M$ with metric $\hat{g}_\lambda$ globally hyperbolic and with ${V_{\pm}(g')}\supset {V_{\pm}(\hat{g}_\lambda)}\supset  {V_{\pm}(g)}.$ Finally assume that for any curve as above $\Delta_{*P_\lambda}^{ret}\in L_b(\mathcal{D}(M,E);\mathcal{E}(M,E))$  is locally bounded in $\lambda$. 

Let $\Gamma_{omc}$ be defined as in example \ref{Opmicrocausal} for $G=\{g\}$
 and define $R(\varphi)=\Delta^{ret}_{*P_{\mathscr{L},\varphi}}*$
and recall $\Gamma_{omc}'=\{(V_+(g),(V_-(g))^c),(V_+(g),V_+(g)),(V_-(g),(V_+(g))^c),(V_-(g),V_-(g))\}$ indexed by $\{+,-\}\times\{1,2\}$
Then $R\in {F}_{\Gamma_{omc}'}(\mathscr{U},\mathcal{K};(\C(\Gamma_{omc}),(\C_g(\Gamma_{omc})_{g\in \{+,-\}\times\{1,2\}},\mathscr{B}(\Gamma_{omc},\mathscr{B}_f,\mathcal{K})))$ and thus defines an hypocontinuous retarded bracket  on $\mathscr{F}_{\Gamma_{omc}}(\mathscr{U},\mathcal{K};\C)$  with topologies $\mathcal{I}_{si}$ or $\mathcal{I}_{ci}$. Assuming the corresponding assumption for an advanced bracket giving a map $A$, $\mathscr{F}_{\Gamma_{omc}}(\mathscr{U},\mathcal{C};(\mathfrak{E},\mathscr{B}_f))$ is thus a Poisson algebra (with hypocontinuous multiplication maps) with Poisson bracket $\{.,.\}=R(.,.)-A(.,.).$
\end{proposition}

\begin{proof}
Note that, using $[V_+(g)]^c\supset V_-(g)$
we  have the simplified formula \begin{align*}\mathfrak{\C}(\Gamma_{omc})=&L_{\beta e}[(\mathcal{E}'_{V_+(g),\overline{V_+(g)}}(U;E^*),\mathcal{I}_{ibi});(\mathcal{D}'_{V_+(g),\overline{V_+(g)}}(U;E),\mathcal{I}_{ibi})]\\&\cap L_{\beta e}[(\mathcal{E}'_{[V_+(g)]^c,\overline{[V_+(g)]^c}}(U;E^*),\mathcal{I}_{ibi});(\mathcal{D}'_{[V_+(g)]^c,\overline{[V_+(g)]^c}}(U;E),\mathcal{I}_{ibi})]\\&\cap L_{\beta e}[(\mathcal{E}'_{V_-(g),\overline{V_-(g)}}(U;E^*),\mathcal{I}_{ibi});(\mathcal{D}'_{V_-(g),\overline{V_-(g)}}(U;E),\mathcal{I}_{ibi})]\\&\cap L_{\beta e}[(\mathcal{E}'_{[V_-(g)]^c,\overline{[V_-(g)]^c}}(U;E^*),\mathcal{I}_{ibi});(\mathcal{D}'_{[V_-(g)]^c,\overline{[V_-(g)]^c}}(U;E),\mathcal{I}_{ibi})]\end{align*} 

and recall the notation (we now often implicitly assume topology $\mathcal{I}_{ibi}$):\begin{align*}&\mathfrak{\C}_{(\pm,1)}(\Gamma_{omc})=L_{\beta e}[\mathcal{E}'_{V_{\pm}(g),\overline{V_{\pm}(g)}}(U;E^*);\mathcal{D}'_{V_{\pm}(g),\overline{V_{\pm}(g)}}(U;E)]\\&\ \ \ \ \ \ \ \ \ \ \ \ \ \ \ \ \ \ \ \ \ \ \ \ \ \ \ \ \ \ \cap L_{\beta e}[\mathcal{E}'_{[V_{\mp}(g)]^c,\overline{[V_{\mp}(g)]^c}}(U;E^*);\mathcal{D}'_{[V_{\mp}(g)]^c,\overline{[V_{\mp}(g)]^c}}(U;E)],\\&
\mathfrak{\C}_{(\pm,2)}(\Gamma_{omc})=L_{\beta e}[\mathcal{E}'_{[V_{\pm}(g)],\overline{[V_{\pm}(g)]}}(U;E^*);\mathcal{D}'_{[V_{\mp}(g)]^c,\overline{[V_{\mp}(g)]^c}}(U;E)].\end{align*}

Thus our previous proposition \ref{SmoothRetarded} already explains why $R$ is (conveniently) smooth with value $\mathfrak{\C}(\Gamma_{omc}).$ Then we have to control the space of value of $R^{(n)}(\varphi).$
 From the computation of $R^{(1)}$ and induction, one sees that it is a linear combination of  compositions of $R(\varphi)$ and differentials of $$\langle d^{(k)}P_{\mathscr{L},\varphi}(h_1,...,h_k)(u),v\rangle=\mathscr{L}(f)^{(k+2)}(\varphi)[h_1,....h_k,u,v]$$ with $f$ chosen depending only of the supports of $u,v$ as above. Using the condition on $\mathscr{L}$, $d^{(k)}P_{\mathscr{L},\varphi}$ is smooth in $\varphi$ with value in \begin{align}\label{dP14}(\mathcal{D}'_{C_{k+2},C_{k+2}}(M^{k+2},\mathcal{K}; &(E^{\o (k+2)})^*),\mathcal{I}_{ibi})\hookrightarrow L({\bigotimes}_{\beta e, i\in [1,k]} (\mathcal{D}'_{V_{\pm}(g),\overline{V_{\pm}(g)}}(M,\mathcal{F};E),\mathcal{I}_{ibi})\nonumber\\&\otimes_{\beta e}(\mathcal{D}'_{V_{\pm}(g),\overline{V_{\pm}(g)}}(M,\mathcal{SK}_P(g');E),\mathcal{I}_{ibi}); (\mathcal{D}'_{V_{\pm}(g),\overline{V_{\pm}(g)}}(M,\mathcal{SK}_P(g');E^*),\mathcal{I}_{ibi})).\end{align}
The first space of value is only with $f$ fixed, but the second one obtained via proposition \ref{multilinear} also holds by commutation of (regular) inductive limit on support and hypocontinuous map for general $u,v$ and also using locality of $\mathscr{L}(f)$ as for the definition of $P$.

Now, as explained above a typical derivative is a linear combination of composition of the form $[R(\varphi)\circ[d^{(k)}P_{\mathscr{L},\varphi}\circ_{k+1}R(\varphi)]\circ_{k+1}d^{(l)}P_{\mathscr{L},\varphi}]\circ_{l+k+1}R(\varphi)$ with an arbitrary length of composition (the term above appears in the $(k+l)$-th differential. From what we stated and proposition \ref{MultilinearComposition} (especially boundedness of composition), the differential above is smooth in $\varphi$ with value in a space of the form $$L({\bigotimes}_{\beta e, i\in [1,k+l]} (\mathcal{D}'_{V_{\pm}(g),\overline{V_{\pm}(g)}}(M,\mathcal{F};E))\otimes_{\beta e}(\mathcal{D}'_{V_{\pm}(g),\overline{V_{\pm}(g)}}(M,\mathcal{SK}_P(g');E)); (\mathcal{D}'_{V_{\pm}(g),\overline{V_{\pm}(g)}}(M,\mathcal{SK}_P(g');E^*))).$$
Note it is here that it is crucial for composition that we developed a theory for general support left stable by propagators like $\mathcal{SK}_P(g').$
Using a partial transpose in the last coordinate, we also have \begin{align}\label{dP15}d^{(k)}P_{\mathscr{L},\varphi}\in& L({\bigotimes}_{\beta e, i\in [1,k]} (\mathcal{D}'_{V_{\pm}(g),\overline{V_{\pm}(g)}}(M,\mathcal{F};E))\nonumber\\& \otimes_{\beta e}(\mathcal{D}'_{(V_{\mp}(g))^c,\overline{(V_{\mp}(g))^c}}(M,\mathcal{SK}_P(g');E)); (\mathcal{D}'_{(V_{\mp}(g))^c,\overline{(V_{\mp}(g))^c}}(M,\mathcal{SK}_P(g');E^*))).\end{align}

Arguing for compositions as before, this gives the smoothness of derivatives valued in $\mathfrak{\C}_{(\pm,1)}(\Gamma_{omc})$ corresponding to cones $(V_{\pm}(g),(V_{\mp}(g))^c)\in \Gamma_{omc}'.$

Arguing as for the first map, we also know that 
\begin{align*}d^{(k)}P_{\mathscr{L},\varphi}\in& L({\bigotimes}_{\beta e, i\in [1,k-1]} (\mathcal{D}'_{V_{\pm}(g),\overline{V_{\pm}(g)}}(M,\mathcal{F};E))\otimes_{\beta e}(\mathcal{D}'_{(V_{\pm}(g),\overline{V_{\pm}(g)}}(M,\mathcal{K}_F(g');E))\\& \otimes_{\beta e}(\mathcal{D}'_{(V_{\pm}(g)),\overline{(V_{\pm}(g))}}(M,\mathcal{SK}_P(g');E)); (\mathcal{D}'_{V_{\pm}(g),\overline{V_{\pm}(g)}}(M,\mathcal{K};E^*))),\end{align*}

  and thus again by partial transpose but with respect to the next-to-last coordinate, we get :
\begin{align}\label{dP16}d^{(k)}P_{\mathscr{L},\varphi}\in& L({\bigotimes}_{\beta e, i\in [1,k-1]} (\mathcal{D}'_{V_{\pm}(g),\overline{V_{\pm}(g)}}(M,\mathcal{F};E))\otimes_{\beta e}(\mathcal{D}'_{(V_{\mp}(g))^c,\overline{(V_{\mp}(g))^c}}(M,\mathcal{F};E))\nonumber\\& \otimes_{\beta e}(\mathcal{D}'_{(V_{\pm}(g)),\overline{(V_{\pm}(g))}}(M,\mathcal{SK}_P(g');E)); (\mathcal{D}'_{(V_{\mp}(g))^c,\overline{(V_{\mp}(g))^c}}(M,\mathcal{SK}_P(g');E^*)).\end{align}
Then by composing \eqref{dP14},\eqref{dP16} and \eqref{dP15} in this order, one gets the smoothness of derivatives of $R$ in the space  valued in $\mathfrak{\C}_{(\pm,2)}(\Gamma_{omc})$ corresponding to cones $(V_{\pm}(g),(V_{\pm}(g)))\in \Gamma_{omc}'.$

It only remains to check the global support condition namely for $B\in \mathscr{B}(\Gamma_{omc},\mathscr{B}_f,\mathcal{K})$, $\text{supp}(R,B)\in \mathcal{K}.$

Without loss of generality, one can assume $B=\{F\mapsto (\langle F\circ_1(a),a'\rangle), a\in A,a'\in A'\}$ for $A\subset\mathcal{E}'_{V_{\pm}(g),\overline{V_{\pm}(g)}}(U;E^*), A'\subset\mathcal{E}'_{V_{\pm}(g),\overline{V_{\pm}(g)}}(U;E^*)$ uniformly supported in $K_1,K_2\in\mathcal{K}$.

From the computation of the first derivative of $R(\varphi)$, one considers the support of the distribution for $a\in A,a'\in A'$ $$h\mapsto\langle R(\varphi)\circ dP_{\mathscr{L},\varphi}(h)\circ R(\varphi)[a], a'\rangle=\mathscr{L}(f)^{(3)}(\varphi)(h,R(\varphi)[a],A(\varphi)[a']).$$
But by support property of retarded and advanced propagators, $\text{supp}(R(\varphi)[a])\subset J^{+,g'}(K_1),$ $\text{supp}(A(\varphi)[a'])\subset J^{-,g'}(K_2),$ and by locality of $\mathscr{L}(f)$ one deduces the support above is in $J^{-,g'}(K_2)\cap J^{+,g'}(K_1)$ which is compact by global hyperbolicity of $g'$.
Once the various hypocontinuity proved, we can reason as in \cite{RibeiroBF} to check one obtains a Poisson algebra structure.
  \end{proof}
\begin{remark}
The argument to prove microcausal functionals form a Poisson algebra even stable by a retarded product $R$ in \cite{RibeiroBF} have the following flaw. The stability of the retarded product is not convincingly checked. As we pointed out to the authors, in the formula for $R(F,G)^{(3)}{[\phi]}(\psi_1,\psi_2,\psi_3)$ there is a term $-F^{(2)}[\phi](\psi_1,D^1\Delta^R[\phi](\psi_2,G^{(2)}[\phi](\psi_3)))= \mathscr{L}(f)^{(3)}(\Delta^A[\phi]F^{(2)}[\phi](\psi_1),\psi_2,\Delta^R[\phi](G^{(2)}[\phi](\psi_3))).$ The distribution  for this term in $R(F,G)^{(3)}{[\phi]}$ is  the multiplication of  $\mathscr{L}(f)^{(3)}$ with $\Delta^A[\phi]F^{(2)}[\phi]\otimes 1\otimes \Delta^R[\phi]G^{(2)}[\phi].$

Of course one can take $\xi_1, \xi_2$ space-like vectors such that $\xi_1+\xi_2=\xi_3\in V_{+}(x)$. Because of the $(x,x,\xi,-\xi)$ terms in the wave front set of $\Delta^A$, $\Delta^R$, the wave front set of  $\Delta^A[\phi]F^{(2)}[\phi]$ contains a priori all the wave front set of $F^{(2)}[\phi]$, which, in the microcausal functional case,  can contain  $(x,x,\xi_1,\xi_3)$, and respectively $G^{(2)}[\phi]$ can contain $(x,x,\xi_2,\xi_3)$. From there one can find a point  $(x,x,x,\xi_3,\xi_1+\xi_2,\xi_3)\in V_{+}(x)^3$ in the wave front set of the term above (the $\xi_1+\xi_2$ appearing because of locality of $\mathscr{L}$). The motivation of the study in this last section comes from the observation of this issue in the argument. One can of course take polynomial functionals with exactly distributional kernels of derivatives given by H\"ormander's example with only one point above in the wave front set in order to show this argument indeed prevents microcausal functionals to be stable by $R$.
\end{remark}

\end{document}